\documentclass[11pt,lettersize,onecolumn]{IEEEtran}

\usepackage[utf8]{inputenc}
\usepackage[T1]{fontenc}
\usepackage{lmodern}
\usepackage{amsmath,amssymb,amsfonts,amsthm,mathtools}
\usepackage{bm}
\usepackage{booktabs}
\usepackage{array}
\usepackage{enumitem}
\usepackage{graphicx}
\usepackage{tikz}
\usetikzlibrary{positioning,arrows.meta,calc,fit,shapes.geometric}
\usepackage{pgfplots}
\pgfplotsset{compat=1.17}
\usepackage{xcolor}
\usepackage{cite}
\usepackage{url}
\usepackage[colorlinks=true,linkcolor=blue!50!black,
            citecolor=blue!50!black,urlcolor=blue!50!black]{hyperref}
\usepackage{cleveref}
\usepackage{setspace}

\theoremstyle{plain}
\newtheorem{theorem}{Theorem}
\newtheorem{proposition}{Proposition}
\newtheorem{lemma}{Lemma}
\newtheorem{corollary}{Corollary}

\theoremstyle{definition}
\newtheorem{definition}{Definition}

\newtheorem{example}{Example}

\theoremstyle{remark}
\newtheorem{remark}{Remark}

\crefname{theorem}{Theorem}{Theorems}
\crefname{proposition}{Proposition}{Propositions}
\crefname{lemma}{Lemma}{Lemmas}
\crefname{corollary}{Corollary}{Corollaries}
\crefname{definition}{Definition}{Definitions}
\crefname{remark}{Remark}{Remarks}
\crefname{assumption}{Assumption}{Assumptions}
\crefname{example}{Example}{Examples}
\Crefname{theorem}{Theorem}{Theorems}
\Crefname{proposition}{Proposition}{Propositions}
\Crefname{lemma}{Lemma}{Lemmas}
\Crefname{corollary}{Corollary}{Corollaries}
\Crefname{definition}{Definition}{Definitions}
\Crefname{remark}{Remark}{Remarks}
\Crefname{assumption}{Assumption}{Assumptions}
\Crefname{example}{Example}{Examples}
\crefname{section}{Section}{Sections}
\Crefname{section}{Section}{Sections}
\crefname{appendix}{Appendix}{Appendices}
\Crefname{appendix}{Appendix}{Appendices}

\newcommand{\R}{\mathbb{R}}
\newcommand{\N}{\mathbb{N}}
\newcommand{\Z}{\mathbb{Z}}
\newcommand{\E}{\mathbb{E}}
\newcommand{\PR}{\mathbb{P}}
\newcommand{\one}{\mathbf{1}}
\newcommand{\zero}{\mathbf{0}}

\newcommand{\defeq}{\mathrel{\vcentcolon=}}

\DeclareMathOperator*{\minimize}{minimize}

\newcommand{\Vset}{\mathcal{V}}
\newcommand{\Cset}{\mathcal{C}}
\newcommand{\Eset}{\mathcal{E}}
\newcommand{\Gn}{G_n}

\newcommand{\RolesV}{\mathcal{R}_{\mathrm V}}
\newcommand{\RolesC}{\mathcal{R}_{\mathrm C}}

\newcommand{\rolemap}{\rho}

\newcommand{\Templates}{\Theta}
\newcommand{\Sockets}{\mathbb{T}}
\newcommand{\socketrole}[1]{r({#1})}
\newcommand{\socketcheckrole}[1]{s({#1})}

\newcommand{\Xstar}{X^{*}}
\newcommand{\Aobs}{A}      
\newcommand{\Zobs}{Z}      
\newcommand{\Mset}{\mathcal{M}}
\newcommand{\Uset}{\mathsf{U}}        
\newcommand{\Vbias}{\beta}            

\newcommand{\epsV}[1][]{\epsilon^{\mathrm V}\ifx\relax#1\relax\else_{\,#1}\fi}
\newcommand{\epsC}[1][]{\epsilon^{\mathrm C}\ifx\relax#1\relax\else_{\,#1}\fi}

\newcommand{\Qchan}{Q}      

\newcommand{\force}{\Gamma}  

\newcommand{\pmsg}{p}        
\newcommand{\hmsg}{h}        
\newcommand{\pbar}{\bar p}   
\newcommand{\forceprob}{\varphi}  
\newcommand{\DEmap}{\Phi}    
\newcommand{\DEjac}{D\DEmap} 
\newcommand{\Pdee}{P_{\mathrm{DE}}}
\newcommand{\Pbit}{P_{\mathrm{bit}}}

\newcommand{\Designspace}{\mathcal{D}}
\newcommand{\Cost}{\mathrm{Cost}}

\newcommand{\rhead}[1]{\par\smallskip\noindent\textit{#1}\enspace\ignorespaces}

\title{On the Reliability of Networks of AI Agents:\\
Density Evolution, Stopping Sets, and Architecture Optimization}

\author{Ehsan~Aghazadeh%
\thanks{E.~Aghazadeh is with the Department of Computer Science,
University of Massachusetts Amherst, Amherst, MA 01003, USA
(e-mail: \texttt{eaghazadeh@umass.edu}).}
and Hossein~Pishro-Nik,~\IEEEmembership{Senior~Member,~IEEE}%
\thanks{H.~Pishro-Nik is with the Department of Electrical and
Computer Engineering, University of Massachusetts Amherst, Amherst,
MA 01003, USA (e-mail: \texttt{pishro@umass.edu}).}%
\thanks{This work was supported in part by the U.S.\ National Science
Foundation under Grants CNS-2528914 and CNS-2150832.}}

\begin{document}
\maketitle

\begin{abstract}
Modern AI systems increasingly solve a task not with a single model
call but with several imperfect agents working together: some propose
pieces of a solution, others verify them, and the results are
combined.  These systems often outperform any single model, yet it is
rarely clear why they succeed or when they will fail.  We model such
a system as message passing on a sparse graph, the structure that
underlies low-density parity-check (LDPC) codes, and extend the
density-evolution machinery of coding theory to this richer setting.

In our model a task is a set of coupled binary subclaims, and an agent
architecture is a sparse, role-typed factor graph whose check nodes
are noisy \emph{Boolean verifier nodes}, each computing a local Boolean
function of the subclaims it touches.  Three distinct failure
modes, all modeled as erasures (an agent abstaining, a verifier
returning no usable output, and a message lost between two agents),
propagate as the agents exchange set-valued messages.  The check agents combine these
messages by a single \emph{logical-forcing} rule that specializes to
XOR, AND, OR, implication, and Horn constraints.  This is more than a
relabeling of LDPC theory: the verifier functions are nonlinear and
value-asymmetric, and the three failure modes do not reduce to a
single effective channel, so they require new threshold, finite-length,
and converse results rather than a direct reuse of parity-check
density evolution.

We prove a density-evolution theorem that predicts the asymptotic
fraction of unresolved subclaims on random role-typed architectures,
with an extension to deterministic, locally tree-like graph
sequences.  The XOR case recovers the classical
LDPC recursion on the binary erasure channel (BEC); the AND case
exposes an asymmetry between positive and negative verifier
certificates.  We then
establish a recovery threshold and show that the three failure modes
are separate design knobs, not a single effective noise
level.  We characterize finite-length failures by
\emph{certificate-stopping sets} and give augmentation conditions that
remove small failure patterns.  Finally, we formulate cost-constrained
architecture optimization whose shadow prices identify the most
valuable knob to improve, and prove a converse showing that the
logical-forcing rule is asymptotically optimal among sound local
rules.  The theory is first-order in scope: surviving certificates are
assumed correct, so confidently wrong messages and correlated failures
lie outside it and require non-erasure extensions.
\end{abstract}


\section{Introduction}\label{sec:intro}
\IEEEPARstart{C}{omplex} cognitive tasks are increasingly executed
by \emph{networks} of imperfect artificial-intelligence agents
rather than by single isolated model calls.  Examples include
formal theorem proving, verified software generation, multi-step
action planning in structured environments, and target classification
from noisy multi-modality observations.  A modern system may decompose a problem into
subgoals, route each subgoal to a specialized model or tool, invoke
verifiers and tests on intermediate artifacts, refactor messages
between roles, and aggregate the surviving evidence into a final
answer.  These systems can substantially outperform their single-component
baselines, sometimes by a wide margin.  Hilbert reaches $99.2\%$ on
miniF2F by orchestrating four cooperating roles around a Lean
type-checker \cite{varambally2025hilbert}.  CodeR reaches $28.33\%$ on
SWE-bench-lite, well above single-LLM baselines, using a task graph of
code generators, test runners, and patch aggregators
\cite{chen2024coder}.  In debate-based oversight, having persuasive
LLMs argue opposing sides helps a weaker judge identify the truthful
answer more reliably \cite{khan2024debate}.

Yet the reason these gains occur is often unclear.  Is the
improvement due to ensembling many independent samples?  To a stronger
verifier?  To the topology of the agent communication graph?  To the
way roles are allocated and how compute is distributed across them?
Practitioners design these systems by trial and error, tuning agent
counts, verifier placements, and communication patterns until a
benchmark goes up.  Recent work argues that agent reliability deserves
study as a discipline in its own right, and finds that rapid capability
gains have so far produced only modest reliability gains
\cite{rabanser2026science}.  We lack the analogue, for networks of AI agents,
of what density evolution gave the LDPC community: a
sharp asymptotic prediction of when a sparse network will be
reliable, which architectural choices push it across the
reliability threshold, and what structural patterns cause it to
fail at finite size.

This paper takes a step toward such a theory.  The starting point is
simple: a broad class of agent networks can be modeled as message
passing on a sparse graph, with agents at the nodes and partial
results sent along the edges.  We track how those messages evolve from
round to round and characterize when the network resolves its task.
The main tool is the density-evolution machinery of LDPC codes, which
we generalize to a setting with distinct agent roles, general Boolean
verifiers that can themselves fail, and messages that can be lost
between particular role pairs.  The classical LDPC erasure recursion is
the special case in which every verifier is a parity check.  On
top of this asymptotic theory we add three things.  First, a
finite-length failure analysis built on a generalization of stopping
sets we call \emph{certificate-stopping sets}, together with an
augmentation theorem that removes all small failure patterns by
targeted cross-verification.  Second, a cost-constrained design method
whose shadow prices say which reliability knob is most valuable to
improve.  Third, a converse showing that the message-passing rule is
optimal among sound local rules.

\subsection{Hard tasks decompose into coupled subclaims}
\label{sec:intro-coupled-subclaims}

A natural first attempt treats every agent as observing a noisy
version of a single global ground truth $y^* \in \{1, \ldots, K\}$.
If agent errors are conditionally independent given $y^*$ and each
agent does a little better than chance, then majority or plurality
voting drives the error to zero exponentially fast in the number of
agents.  This regime is real and important; it covers calibration
questions on multiple-choice benchmarks.  But the communication
topology barely matters here: when every agent sees the same label and
the aggregator only needs a majority vote, the graph does almost no
work.

Hard cognitive tasks have internal structure.  A formal proof
decomposes into a sequence of lemma applications, definitional
unfoldings, and case splits.  A program decomposes into modules,
functions, and tests, with local consistency constraints between them.
A structured plan decomposes into actions and dependencies, each of
which can be checked against its neighbors but not on its own.  We model such a
task by a hidden vector of coupled binary \emph{subclaims}
\begin{equation}
\label{eq:intro-Xstar}
   \Xstar = (\Xstar_1, \ldots, \Xstar_n) \in \{0, 1\}^n.
\end{equation}
Here $\Xstar_i = 1$ means that subclaim $i$ is correct (the proof step
type-checks, the test passes, or the local rule is satisfied), and
$\Xstar_i = 0$ means it is not.  $\Xstar$ is a fixed ground truth and
does not depend on the graph's check operations: whether a check
computes AND, XOR, or another Boolean relation, $\Xstar$ is the same.
The one convention we adopt is to read ``correct'' in the
\emph{validator-relative} sense, by the standard the system applies to
each subclaim rather than by an abstract notion.  This keeps the
verifiers sound: a verifier can abstain, but never certifies a wrong
value.  Sometimes the verifier is only an approximate test of true
correctness: a fixed set of tests can all pass even though the code is
still wrong.  We place this confidently-wrong case outside our
first-order theory (\Cref{sec:applications}).  Each local check tests
whether a small subset of subclaims is jointly consistent.  The
coupling between subclaims is therefore relational: it lives in these
checks, not in the prior on $\Xstar$ itself.
We take the entries of $\Xstar$ to be independent given their roles,
though the same density-evolution limits hold more generally
(\Cref{sec:model-observations}).

\emph{In words.} Replacing one global label by a vector of coupled
subclaims is what makes the topology of agent communication matter.
Some unresolved subclaims can be filled in by neighboring checks; some
form residual clusters that no amount of additional rounds will fix.
The reliability of the overall system depends on which structural
patterns the agent graph creates and which it avoids.

The contrast across the three regimes that any agent-system theory
must distinguish is shown in Table~\ref{tab:three-regimes}.

\begin{table}[t]
\centering
\renewcommand{\arraystretch}{1.25}
\caption{Three regimes for multi-agent reasoning, distinguished by what
the hidden truth looks like and where the analytical action is.  This
paper develops the theory of the second and third rows.}
\label{tab:three-regimes}
\small
\begin{tabular}{p{0.27\linewidth}p{0.27\linewidth}p{0.36\linewidth}}
\toprule
\textbf{Setting} & \textbf{Natural method} & \textbf{Why it matters here} \\
\midrule
One hidden answer, independent agents
  & Majority vote, self-consistency
  & Topology and roles barely matter; LLN dominates. \\
Coupled subclaims with local checks
  & Message passing on a sparse graph
  & Topology creates recovery and failure patterns. \\
Noisy verifiers and channels on top of a coupled-subclaim task
  & Message passing with noisy, role-typed verifiers
  & Role-typed density evolution and certificate-stopping sets show which bottleneck binds; different ones need different architectural fixes. \\
\bottomrule
\end{tabular}
\end{table}

\subsection{The model in one picture}
\label{sec:intro-one-picture}

The mathematical object is a sparse bipartite factor graph
\begin{equation}
\label{eq:intro-Gn}
   \Gn = (\Vset_n \cup \Cset_n,\, \Eset_n),
\end{equation}
drawn from a role-typed bounded-degree configuration model; the
resulting ensemble is locally tree-like in the limit $n \to \infty$
(\Cref{lem:tree-like}).  Two kinds of nodes live on it.

\emph{Variable agents} $i \in \Vset_n$ each hold one subclaim
$X_i \in \{0,1\}$.  A variable agent's role $r \in \RolesV$ says
what kind of agent it is, for example ``junior prover,''
``code-fragment proposer,'' or ``retrieval agent.''  Roles are how the
model captures systematic differences in reliability across agents.

\emph{Check agents} $a \in \Cset_n$ each evaluate a local Boolean
function $C_a : \{0,1\}^{|\partial a|} \to \{0,1\}$ on a small
neighborhood $\partial a$ of subclaims.  In the framework of this
paper, $C_a$ is an arbitrary bounded-arity Boolean function.  Two
analytically tractable specializations carry most of the work.  The
\emph{XOR specialization} sets $C_a(x_{i_1}, \ldots, x_{i_d}) =
x_{i_1} \oplus \cdots \oplus x_{i_d}$ and recovers the LDPC-clean
linear case.  The \emph{AND-monotone specialization} sets $C_a =
x_{i_1} \wedge \cdots \wedge x_{i_d}$ and is the realistic verifier
case.  These are the questions real verifiers ask: ``do these proof
steps type-check together?'' (a Lean kernel call); ``do these unit
tests all pass?'' (a CodeR test runner); ``does this configuration
satisfy every local consistency rule?''  Horn clauses, OR factors, and
implications fall under the same framework.

Throughout the paper, $r \in \RolesV$ denotes a variable-agent role
and $s \in \RolesC$ a check-agent role.

Three sources of noise enter the model, all at initialization $t = 0$.

\begin{enumerate}[leftmargin=2em]
\item \emph{Variable-side erasure $\epsV[r]$.}  A variable agent of role
$r$ abstains, returns ``$?$'', and otherwise reports the true value
$\Xstar_v$.  In practice, $\epsV[r]$ is the rate at which a step checker
times out, fails to commit, or reports ``not enough context.''
\item \emph{Verifier-side erasure $\epsC[s]$.}  A check agent of role
$s$ fails to produce a verdict on its constraint and otherwise reports
$T_a = C_a(\Xstar_{\partial a})$.  In practice, $\epsC[s]$ is the rate at
which a Lean kernel times out, a sandbox fails, a test cannot be
invoked, or a judge declines to rule.
\item \emph{Reasoning-channel erasure $1 - \eta_{r,s}$.}  When a
message produced by an agent of role $r$ is sent to an agent of role
$s$, the message is replaced by ``$?$'' (lost or unusable on receipt)
with probability $1 - \eta_{r,s}$.  In practice, this is the rate at
which a refactor between two roles fails, a chain-of-thought becomes
unparseable, or an artifact format mismatch makes the receiver unable
to use what the sender computed.
\end{enumerate}

After this single noise injection, message passing is deterministic.
We use \emph{extrinsic} edge-specific updates: the message from $u$
to $w$ is computed without using the previous message from $w$ to
$u$.  This extrinsic property is what keeps local neighborhoods
asymptotically tree-like, and is what makes density evolution
rigorous in the configuration-model ensemble.  Real agent systems
typically \emph{broadcast}: each round-$t$ output is sent to every
neighbor, including the agent that produced the incoming message.  The
extrinsic update drops this self-reuse; this is the standard
idealization in message-passing analysis \cite{richardson2008modern}.
A message can turn ``$?$'' into a definite value but never the
reverse, so reusing one's own earlier message changes nothing on a
tree.  The difference between broadcasting and the extrinsic rule
therefore appears only through short cycles.  This is the same kind of finite-cycle
effect that produces the error floor measured in \Cref{sec:numerical}
(\Cref{fig:det-errorfloor}); a cycle-aware version of the recursion is
left to future work.  The iteration index $T$
is how many rounds the system runs: how many times the agents exchange
and update their outputs before stopping, a finite budget in real
systems.  \Cref{sec:applications} traces this concretely for
formal-proof, code-generation, and debate systems.

\emph{The unifying check-to-variable update is the logical-forcing
rule.}  A check $a$ is joined to a few variables, one per edge, and each
variable holds one subclaim (each point where a variable attaches to
the check is a \emph{socket}).  Its \emph{template} $\theta$ fixes the
Boolean function $f_\theta$ that $a$ computes over those $d_\theta$
inputs.  We write $z \in \{0,1\}$ for the check's observed
output, for example $z = 1$ if a test passes and $z = 0$ if it fails.

Messages travel both ways along each edge.  From a variable to a
check, the message is a candidate set
$M_j \in \Mset \defeq \{\{0\}, \{1\}, \Uset\}$: a singleton $\{b\}$ if
the sender certifies that its subclaim equals $b$, the unresolved
$\Uset = \{0,1\}$ otherwise.  From a check back to the subclaim at a
\emph{target} socket $j$, the check combines its observed output $z$
with the incoming messages $M_{-j} = \{M_k\}_{k \neq j}$ from the other
sockets and returns the target values that remain consistent with what
it knows:
\begin{equation}
\label{eq:intro-forcing}
\force_{\theta, j}(z;\, M_{-j})
   \defeq
\big\{ b \in \{0, 1\} :\;
   \exists\, x_k \in M_k,\, k \neq j,\;
   f_\theta(x_1, \ldots, x_{j-1}, b, x_{j+1}, \ldots, x_{d_\theta}) = z\big\}.
\end{equation}
In words, a value $b$ survives if the other inputs can be filled in
from their candidate sets so that the function outputs the observed
$z$.  If exactly one value survives, the check has forced the target
and sends the singleton $\{b\}$; otherwise it sends $\Uset$
(\Cref{fig:forcing}).  This single rule specializes
correctly to every Boolean primitive of interest.  Under XOR, the
target is forced exactly when every other input is known, recovering
the classical BEC singleton-neighbor rule.  Under AND, a positive
verifier output ($z = 1$) certifies every input as $1$ at once even
when several were previously unresolved, while a negative verifier
output ($z = 0$) certifies a target as $0$ only when every other input
is already known to be $1$.  This positive--negative asymmetry is
absent from XOR but present in real verifier semantics, which is why
Boolean verifier nodes, not parity checks alone, are a faithful
abstraction for the certifying layer of agent systems.

\begin{figure}[t]
\centering
\begin{tikzpicture}[>=Latex, every node/.style={font=\small},
  vnode/.style={circle, draw, minimum size=0.9cm, inner sep=0pt},
  vtarget/.style={vnode, fill=gray!25},
  cnode/.style={draw, minimum size=0.9cm, inner sep=0pt}]
  \node[vnode]   (u1) at (-2.4, 1.8) {$u_1$};
  \node[vnode]   (u2) at (-0.6, 1.8) {$u_2$};
  \node[vtarget] (w)  at ( 2.4, 1.8) {$w$};
  \node[cnode]   (a)  at ( 0.0, 0.0) {$a$};
  \node (zsrc) at (0.0, -1.5) {observed output $z = f_\theta(\cdot)$};
  \draw[->] (zsrc) -- (a);
  \draw[->] (u1) -- node[above left=-3pt] {$M_1$} (a);
  \draw[->] (u2) -- node[left=-1pt] {$M_2$} (a);
  \draw[->] (a) -- node[below right=-2pt] {$\force_{\theta,j}(z;\,M_{-j})$} (w);
  \node[font=\footnotesize, text=black!70] at (-1.5, 2.5) {other sockets $k \neq j$};
  \node[font=\footnotesize, text=black!70] at ( 2.4, 2.5) {target socket $j$};
\end{tikzpicture}
\caption{The logical-forcing rule at a single check.  Check $a$ has
template $\theta$, so it computes the Boolean function $f_\theta$ over
the subclaims at its sockets, and its output is observed as
$z \in \{0,1\}$ (for instance a test reporting pass or fail).  To
update the target subclaim $w$ at socket $j$, the check combines $z$
with the incoming candidate sets $M_{-j} = \{M_k\}_{k \neq j}$ from the
other sockets and keeps the target values consistent with both
(\Cref{eq:intro-forcing}).  If exactly one value is consistent the
check forces $w$ to it and sends a singleton; otherwise the target
stays unresolved ($\Uset$).  XOR forces the target once every other
input is known; AND with $z = 1$ forces all inputs to $1$ at once.}
\label{fig:forcing}
\end{figure}
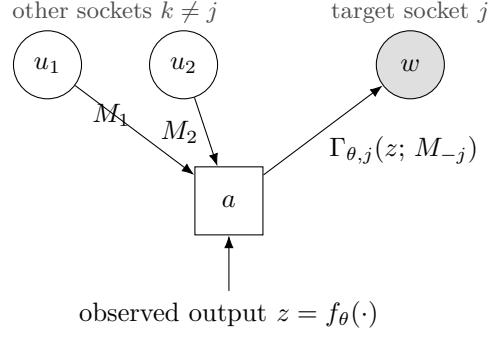

Two global objects anchor this picture, and it is worth naming them in
coding terms.  The hidden vector $\Xstar$ is the validator-relative
ground truth defined above, the analog of a transmitted codeword that
no agent observes directly.  The set-valued messages are the running
estimate of $\Xstar$: a singleton $\{b\}$ is a coordinate resolved to
$b$, and $\Uset$ is a coordinate still unresolved, the erasure
``$?$'' of a binary-erasure-channel decoder.  The message vector
therefore plays the role of a decoder state such as
$(0, \Uset, 1, \Uset, 1)$, and soundness guarantees that each resolved
coordinate equals the true $\Xstar_j$, so an iteration can only turn an
$\Uset$ into a singleton, never into a wrong value.

The correspondence to coding theory is summarized by the dictionary
in Table~\ref{tab:dictionary}.  Each row pairs an LDPC concept with
its agent-network counterpart in a form that lets the technical
content transfer in both directions.

\begin{table}[t]
\centering
\renewcommand{\arraystretch}{1.25}
\caption{The translation between LDPC and agent-network language used
throughout the paper.  The analogy is structural: parity-check
constraints are engineered into the code, whereas verifier functions
are inherited semantic constraints on the task
(see \Cref{sec:intro-ldpc-bridge}).}
\label{tab:dictionary}
\small
\begin{tabular}{p{0.40\linewidth}p{0.55\linewidth}}
\toprule
\textbf{LDPC / BEC concept} & \textbf{Agent-network meaning} \\
\midrule
Code bit / variable node
   & Hidden subclaim or task component held by a variable agent \\
Parity check / factor node
   & Local Boolean verifier node (XOR specialization is parity) \\
Received erasure
   & Variable agent lacks usable local evidence (abstains, times out) \\
Check evaluation
   & Verifier, test, proof-checker, or judge observation \\
Message-passing iteration
   & One round of inter-agent communication \\
Message erasure
   & Failed or unusable reasoning artifact (refactor failure) \\
Stopping set
   & Residual cluster the local agents cannot resolve \\
Two-edge-connected augmentation
   & Cross-verification layer that frees small clusters \\
Degree distribution / role mix
   & Agent architecture and communication pattern \\
BP / recovery threshold
   & Reliability-recovery threshold surface \\
\bottomrule
\end{tabular}
\end{table}

\rhead{Scope of the erasure-only model.}  The framework targets
the \emph{certifying layer} of verifier-driven agent systems: proof
kernels, test runners, validators, abstaining subclaim generators,
and parseable/unparseable role transitions.  In this layer a
non-erased certificate is assumed sound, while missing, timed-out,
or unusable outputs are treated as erasures.  Confident wrong
messages (hallucinated subclaims asserted with the same surface form
as correct ones) fall outside the present first-order theory and
require a hybrid erasure-and-flip extension that we discuss as
a follow-on direction in \Cref{sec:outlook-bsc}.  Our positioning
is that the erasure-only theory captures the failure modes that
dominate in certifying agent stacks, not that LLM agents never
produce wrong messages.

\rhead{What the message-passing model assumes.}  The graph and its
update rule describe how local certificates combine; they are not a
claim that the agents run this decoder inside themselves.  What we
follow is a reliability state: for each subclaim, whether it has been
resolved and to which value, and how that state changes as evidence
passes between agents.  Three points follow.  First, the round count
$T$ is a finite depth, not a run to convergence: a single forward
pass, in which each agent produces its output once, is the case
$T = 1$, and letting verifiers send their results back for another
round adds depth.  Second, the only structural assumption is that the
task-dependency graph is locally tree-like, so the certificates
reaching a check are nearly independent; this is a property of the
task graph, not of any agent's internal reasoning, and where it fails,
through hubs or short cycles, the cost is the error floor measured in
\Cref{sec:numerical}.  Third, logical forcing is the sound certifying
rule, and \Cref{thm:converse} shows it is the best such rule, so the
reliability limits we prove bound any sound local certification scheme,
whatever algorithm a deployed system runs.

\subsection{A toy example: a team of agents checking a 4-step proof}
\label{sec:intro-toy}

To make the model tangible, we walk through a tiny scenario that a
contemporary multi-agent proof-checking system could plausibly face.
The checks here are AND tests, which is what real verifiers compute:
each asks whether a group of steps composes into a valid
sub-derivation.  The linear special case, in which each check is
instead a parity constraint, is what reproduces the classical LDPC
erasure recursion exactly and connects the model to coding theory
(\Cref{sec:intro-ldpc-bridge}); the formal theorems cover the general
Boolean case (\Cref{sec:stopping}).

\emph{The task.}  A junior model has produced a candidate proof of the
elementary lemma
\begin{quote}
\textit{If $n$ and $m$ are even integers, then $n + m$ is even.}
\end{quote}
The proof attempt is broken into four steps:
\begin{enumerate}[leftmargin=2em]
\item[(1)] \emph{unfold ``even'' on $n$:} ``there exists an integer $k$
with $n = 2k$,''
\item[(2)] \emph{unfold ``even'' on $m$:} ``there exists an integer $j$
with $m = 2j$,''
\item[(3)] \emph{substitute and combine:} ``then $n + m = 2k + 2j =
2(k+j+1)$,''
\item[(4)] \emph{apply ``even'' to conclude:} ``since $k + j + 1$ is an
integer, $2(k + j + 1)$ is even, so $n + m$ is even.''
\end{enumerate}
Each step has a hidden truth label $\Xstar_i \in \{0,1\}$ that equals
$1$ if the step is logically valid and $0$ if it contains a flaw.
Here step~3 hides an arithmetic slip, $2k + 2j = 2(k+j+1)$ where the
correct value is $2(k+j)$, so the true labels are
$\Xstar = (1, 1, 0, 1)$.  The team of agents does not know this yet,
and its goal is not just a global ``proof correct/incorrect'' verdict
but a per-step verdict $\hat X_i \in \{0,1\}$ for every $i$.  A global
reject tells a downstream user nothing about which step to fix; per-step
verdicts let an editor or another agent target the actual flaw, and
this is what verifier-centered proof systems such as Hilbert
\cite{varambally2025hilbert} report.

\emph{The agents.}  Variable agents $v_1, v_2, v_3, v_4$ are LLM
\textbf{step checkers}, one per step.  Each $v_i$ is given step $i$ in
isolation, together with its hypotheses, and asked the single local
question ``is this one inference valid?''  For instance, $v_1$ checks
that ``there exists an integer $k$ with $n = 2k$'' correctly unfolds
``$n$ is even,'' $v_3$ checks the algebra in step~3, and $v_4$ checks
that ``$n + m$ is even'' follows from ``$n + m = 2(k+j+1)$.''  Its
private observation is
\begin{equation*}
   \Aobs_i = \begin{cases}
   \{\Xstar_i\}, & \text{if the step checker produced a confident
   verdict,} \\
   \Uset, & \text{abstention (timeout, ambiguous notation, missing
   context).}
   \end{cases}
\end{equation*}
A step checker may abstain, but by soundness it never reports a
confident verdict that disagrees with $\Xstar_i$.

Check agents $a_1, a_2, a_3$ are higher-level \textbf{consistency
checkers}.  Each $a_j$ takes a small group of steps and re-derives
whether they compose into a coherent sub-derivation, reporting
$\Zobs_j = 1$ if the group re-checks and $\Zobs_j = 0$ if it does not.
This is the AND of the group's step labels,
$\Zobs_j = C_{a_j}(\Xstar_{\partial a_j})$, which is what a proof
assistant computes when it re-runs a block of reasoning.  The
observation is erased, $\Zobs_j = {?}$, when the check times out.  We
take three consistency checks with the connectivity
\[
   a_1 \leftrightarrow \{v_1, v_2, v_3\}, \quad
   a_2 \leftrightarrow \{v_2, v_3, v_4\}, \quad
   a_3 \leftrightarrow \{v_1, v_4\},
\]
so $a_1$ and $a_2$ re-derive the front and back halves of the proof
(both touch the flawed step~3 and fail), while $a_3$ checks that the
opening and closing definitional moves agree (it does not touch
step~3 and passes).  A consistency check is informative in two
asymmetric ways.  A \emph{passing} block ($\Zobs_j = 1$) is a
\textbf{positive certificate}: the whole sub-derivation re-checks, so
every step it touches is valid at once.  A \emph{failing} block
($\Zobs_j = 0$) is a \textbf{negative certificate}: it shows at least
one step is broken, and pins down which one only when every other step
it touches is already known valid, in which case the lone remaining
step must be the flaw.  The bipartite factor graph is drawn in
\Cref{fig:toy-graph}; shaded nodes mark the channel realization
analyzed below.

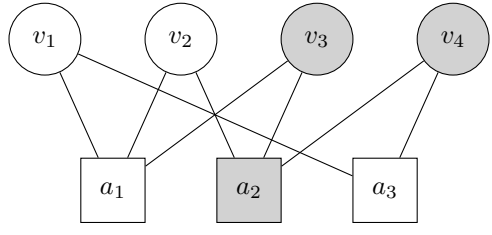
\begin{figure}[t]
\centering
\begin{tikzpicture}[scale=1.0,
  every node/.style={font=\small},
  vnode/.style={circle, draw, minimum size=0.95cm, inner sep=0pt},
  cnode/.style={draw, minimum size=0.85cm, inner sep=0pt},
  vshade/.style={vnode, fill=gray!35},
  cshade/.style={cnode, fill=gray!35}]
  \node[vnode] (v1) at (0, 1.6) {$v_1$};
  \node[vnode] (v2) at (1.8, 1.6) {$v_2$};
  \node[vshade] (v3) at (3.6, 1.6) {$v_3$};
  \node[vshade] (v4) at (5.4, 1.6) {$v_4$};
  \node[cnode] (a1) at (0.9, -0.4) {$a_1$};
  \node[cshade] (a2) at (2.7, -0.4) {$a_2$};
  \node[cnode] (a3) at (4.5, -0.4) {$a_3$};
  \draw (v1) -- (a1); \draw (v2) -- (a1); \draw (v3) -- (a1);
  \draw (v2) -- (a2); \draw (v3) -- (a2); \draw (v4) -- (a2);
  \draw (v1) -- (a3); \draw (v4) -- (a3);
\end{tikzpicture}
\caption{Bipartite factor graph for the 4-step proof scenario in
\Cref{sec:intro-toy}.  Variable agents $v_1, \ldots, v_4$ (circles, top
row) are step checkers holding the per-step verdicts; check agents
$a_1, a_2, a_3$ (squares, bottom row) are consistency checks that
re-derive whether a group of steps composes.  The hidden flaw is in
step~3, so the two checks touching it, $a_1$ and $a_2$, fail, while
$a_3$ passes.  The shaded nodes $v_3, v_4, a_2$ denote the channel
realization analyzed in the text: step checkers $v_3, v_4$ abstain
(variable-side erasure) and check $a_2$ times out (verifier-side
erasure); unshaded nodes deliver cleanly.}
\label{fig:toy-graph}
\end{figure}

\emph{Noise enters once, at $t = 0$.}  A representative realization
of all three noise tiers:
\begin{itemize}[leftmargin=2em]
\item \emph{Variable-side:} $v_1, v_2$ run cleanly and report
$\Aobs_1 = \Aobs_2 = \{1\}$.  $v_3, v_4$ abstain and report
$\Aobs_3 = \Aobs_4 = \Uset$; in particular the step checker for the
flawed step gives no verdict of its own.
\item \emph{Verifier-side:} $a_1$ re-derives the front half and reports
failure, $\Zobs_1 = 0$; $a_2$ exceeds its compute budget and reports
$\Zobs_2 = {?}$; $a_3$ re-derives the definitional bracket and reports
success, $\Zobs_3 = 1$.
\item \emph{Reasoning channels:} every directed edge delivers cleanly
in this first pass.
\end{itemize}

After this one-shot noise injection the agents exchange messages
deterministically.  In round 1, each step checker forwards its
observation to its consistency-check neighbors, the singleton
$\{1\}$ if it has one and the unresolved $\Uset$ otherwise.  In round
2, each consistency check tries to resolve the steps it touches, using
its own pass/fail observation and the messages just received from the
other steps in its group:
\begin{itemize}[leftmargin=2em]
\item $a_3 \to v_4$: $a_3$ passed, $\Zobs_3 = 1$, a positive
certificate, so it forces both of its steps to $\{1\}$ and recovers
the abstaining $v_4$.  \textbf{Success.}
\item $a_1 \to v_3$: $a_1$ failed, $\Zobs_1 = 0$, with inbound
singletons $\{1\}$ from $v_1$ and $\{1\}$ from $v_2$.  Every other step
it touches is confirmed valid, so the negative certificate forces the
lone remaining step, $\Xstar_3 = 0$, and $a_1$ sends this verdict to
$v_3$.  \textbf{Success: the flaw is localized.}
\item $a_2 \to v_3$: $a_2$ timed out, $\Zobs_2 = {?}$, so the forcing
operator returns $\Uset$ and $a_2$ contributes nothing.
\textbf{Fails (verifier-side erasure), but harmlessly: $a_1$ has
already localized step~3.}
\end{itemize}

\noindent
\fbox{\parbox{\dimexpr\linewidth-2\fboxsep-2\fboxrule\relax}{%
\textbf{Takeaway.} Even with two abstaining step checkers ($v_3, v_4$)
and one timed-out consistency check ($a_2$), the team produces
$\hat X = (1, 1, 0, 1)$: steps 1, 2, and 4 confirmed valid and step 3
correctly flagged as the flaw.  One passing block gave a positive
certificate and one failing block gave a negative certificate; the
communication structure recovered, and localized, what no single agent
could.}}\medskip

\emph{Now break the example three ways.}  Small changes to the
channel realization illustrate the three failure modes of the
stopping-set theorem (\Cref{thm:certstop}).

\textbf{(M1) The verifier fails.}  If $a_1$ also times out, both
checks touching step~3 ($a_1, a_2$) have erased verdicts and $a_3$
does not touch step~3, so $\hat X_3$ stays $\Uset$: every check that
could have caught the flaw gave up.

\textbf{(M2) Too many unknowns in one group.}  Restore $a_1$ but let
$v_2$ also abstain.  Then $a_1$ fails with inbound messages
$\{1\}, \Uset, \Uset$ from $v_1, v_2, v_3$: a failing check with two
unresolved steps cannot tell which one is broken, so the negative
certificate does not fire and both $v_2$ and $v_3$ stay $\Uset$.  A
check can localize at most one unknown step at a time.

\textbf{(M3) The reasoning channel fails.}  Restore the original
channel but drop the directed edge $a_1 \to v_3$ (a failed format
conversion between $a_1$'s output and $v_3$'s expected input).  Now
$a_1$ correctly determines $\Xstar_3 = 0$, but the verdict never
arrives in usable form and $v_3$ is left with $\hat X_3 = \Uset$: the
flaw was found but not communicated.

These three modes (verifier-erased, multi-step combinatorial,
channel-erased) are the three modes of the stopping-set theorem
(\Cref{thm:certstop}, \Cref{cor:and-stopping}), and each calls for a
different fix: replace or add a verifier (M1), restructure the
consistency-check connectivity (M2), or harden the communication
channel (M3).  Making these three knobs separately tunable, with
adjoint-derived shadow prices saying which one binds, is one of the
main contributions of the theory (\Cref{thm:optimization}).

\emph{The linear special case.}  If each check were a parity
constraint, $\bigoplus_{i \in \partial a_j} X_i = T_j$, rather than an
AND, the same machinery would apply but recovery would become
symmetric: a check missing exactly one of its inputs recovers it by
parity, whatever the values, and two missing inputs leave it
unresolved.  This linear case is the one that reproduces the classical
LDPC erasure-decoding recursion exactly, and it is the bridge to coding
theory developed next (\Cref{sec:intro-ldpc-bridge},
\Cref{cor:xor-stopping}).

\subsection{Why this is the LDPC story, and why it is not}
\label{sec:intro-ldpc-bridge}

\rhead{Where the algebraic constraint comes from.}
In a parity-check code every codeword satisfies $H\bm x = \bm 0$, a
constraint the designer engineers into the system, and that engineered
constraint is what drives iterative decoding.  In an agent network the
hidden values $X^*_i$ are facts about the world, not bits of a
codeword, so it is worth asking what plays the role of the parity
equation.  The answer is one level up.  Each verifier function
$C_a:\{0,1\}^{d_a}\to\{0,1\}$, together with its true output
$T_a^* \defeq C_a(X^*_{\partial a})$, asserts a Boolean equation
$C_a(X^*_{\partial a}) = T_a^*$ (literally a parity check under XOR;
``all inputs are $1$'' or ``some input is $0$'' under AND).  The
constraint set
$\mathcal{C} = \{X \in \{0,1\}^n : C_a(X_{\partial a}) = T_a^*
\;\forall a\}$
plays the role of the codebook, and iterative recovery resolves the
noisy observations $(\Aobs, \Zobs)$ to the element of $\mathcal{C}$ they
identify, with each verifier function locally restricting the candidate
set as in LDPC.

The one difference is where the constraints come from.  A code
designer engineers $H$ for minimum distance or threshold; an
agent-network designer inherits $\{C_a\}$ from the task domain:
logical inference rules in proof checking, module-correctness tests
in code generation, biological-consistency relations in medical
diagnosis.  These constraints are real features
of the task, present whether or not an agent network is deployed on
them.  The remaining design freedom is which constraints to deploy and how
to set the role and degree distributions; this is exactly what the
architecture-optimization theorem (\Cref{thm:optimization}) operates
on.  This sharpens the LDPC analogy rather than weakening it.

\rhead{Where the random-graph hypothesis comes from.}
Density evolution as usually presented relies on a randomly
constructed graph, and the concentration argument of \Cref{thm:de}
needs that randomness.  Four points justify or relax this in the
agent-network setting.  First, the random graph is a design variable,
not a constraint imposed by the world.  The architect chooses which
constraints to deploy, then samples the role-typed matching uniformly
over \emph{compatible} sockets (those sharing role, template, and
socket-type semantics), not over arbitrary attachments.  Second, when
deliberate randomization is unavailable, real task graphs at large $n$
often have configuration-model-like degree statistics; this is a
modeling assumption, the analog of ``real channels are approximately
memoryless.''  Third, the configuration-model hypothesis is sufficient
but not necessary for the DE prediction itself: any locally-tree-like
graph sequence drives the same recursion in expectation, since each
depth-$L$ message depends only on a local tree neighborhood.  The
random construction additionally supplies the concentration bound, via
a Doob martingale over the matching.  \Cref{thm:de-deterministic}
proves a deterministic-graph counterpart under Benjamini--Schramm
local convergence.  Fourth, the
framework breaks when the graph has heavy hubs or systematic short
cycles, where the locally-tree-like approximation fails on a
non-vanishing fraction of vertices; cluster-aware DE with local-MAP
correction is the natural follow-on.  The scope of the present theory
is locally-tree-like task graphs, deterministic or random.

In the XOR specialization the recovery structure is a standard LDPC
pattern: density evolution predicts the asymptotic erasure
probability \cite{luby1998analysis,richardson2001capacity,%
richardson2001design,richardson2008modern}, stopping sets characterize
finite-length BEC failures \cite{di2002finite}, and irregular and
unequal-error-protection LDPC design links the degree distribution to
threshold and protection tradeoffs \cite{pishro2007nonuniform,%
richardson2008modern}.  We use the same toolkit:
concentration on locally tree-like neighborhoods, spectral-radius
stability of the DE Jacobian, peeling-decoder fixed-point
characterization, and two-edge-connected augmentation.

Two things change in the agent-network setting.  First, three
independent erasure tiers (variable-side, verifier-side,
reasoning-channel) enter the density-evolution Jacobian in
structurally different positions.  The Non-Interchangeability
Proposition (\Cref{prop:noninterchange}) makes this precise: no change
of variables collapses the three into a single effective
$\tilde\epsilon$, so improving a verifier, adding proposer redundancy,
and improving a communication channel are separate design knobs whose
marginal values must be computed independently.  This is the
structural distinction from multi-edge-type LDPC
\cite{richardson2008modern} and the noisy-message-passing-decoder line
\cite{tarighati2015noisy,dupraz2021noisy}, neither of which carries
such a statement (\Cref{tab:related-comparison}).

Second, the check primitive need not be linear.  Every bounded-degree
Boolean factor fits the framework, and the value-conditioned recursion
of \Cref{thm:de} tracks separate per-value erasure probabilities
$p_{\ell,\tau}^{(0)}$ and $p_{\ell,\tau}^{(1)}$ when the factor is not
value-symmetric.  Under XOR these collapse and we recover standard
LDPC-BEC analysis (\Cref{cor:xor-de}); under AND the recursion exposes
a positive-versus-negative certificate asymmetry.  This asymmetry is
the formal version of a practitioner intuition: a passing test is more
informative than a failing one.  A passing Lean-kernel call certifies
every step in its scope at once, while a failing one localizes the
error only once the rest of the scope is known good.
\Cref{prop:and-de} quantifies this.

\subsection{Concrete applications}
\label{sec:intro-applications}

Each primitive of the framework maps onto an operational quantity
of a deployed multi-agent system, and each parameter is in principle
measurable from agent traces.  Concretely,
$\epsV[r]$ is the abstention or timeout rate of role-$r$ variable
agents; $\epsC[s]$ is the no-verdict rate of role-$s$ verifiers;
$\eta_{r,s}$ is the success probability of transmitting or
translating a usable artifact from role $r$ to role $s$; and the
template proportions and role-degree laws are read off the deployed
task graph.  After calibration, the DE recursion of \Cref{thm:de}
predicts the asymptotic residual erasure rate, the
stopping-set/augmentation theorems
(\Cref{thm:certstop}, \Cref{thm:augmentation}) identify which finite
patterns cause failures, and the adjoint sensitivities
(\Cref{thm:optimization}\,(d)) pinpoint which reliability tier is
most valuable to improve at the margin.

Five deployed systems illustrate this pipeline, and
\Cref{sec:applications} traces each one in detail.  In multi-agent
formal theorem proving, Hilbert \cite{varambally2025hilbert} reaches
$99.2\%$ on miniF2F with Lean-kernel checks of AND-monotone form.  In
multi-agent code generation, CodeR \cite{chen2024coder} and the
SWE-bench line \cite{jimenez2024swebench,yang2024sweagent} use test
runners, and the verified-code benchmarks CLEVER and FVAPPS
\cite{thakur2025clever,dougherty2025proving} use Lean type-checking,
as AND-monotone checks.  MAKER \cite{meyerson2025maker} runs a
one-million-step task by maximal decomposition with per-subtask
voting, a cross-verification structure of the kind
\Cref{thm:augmentation} formalizes.  LLM debate
\cite{du2024debate,khan2024debate,kenton2024oversight,choi2025debate}
maps each debate round to one message-passing iteration in
structured-output settings.  Sensor and drone classification fits
only after a non-erasure extension (\Cref{sec:outlook}), since soft
detections are not erasures.  For each system, \Cref{sec:applications}
maps the agent roles to variable and check nodes and ties $\epsV$,
$\epsC$, and $\eta$ to quantities already logged in agent traces.

\subsection{Main contributions}
\label{sec:intro-contributions}

The paper makes the following main contributions.

\begin{enumerate}[leftmargin=2.5em, label=\textbf{C\arabic*.},
                itemsep=0.5ex, ref={C\arabic*}]
\item \textbf{A role-typed Boolean-verifier-node model.}  We define
a sparse role-typed bipartite factor-graph ensemble in which variable
nodes are subclaims, check nodes are noisy Boolean verifiers of
bounded degree, and messages are set-valued certificates.  Three
erasure tiers (variable-side $\epsV[r]$, verifier-side $\epsC[s]$,
reasoning-channel $\eta_{r,s}$) appear as separate channel layers.
A single logical-forcing rule $\force_{\theta,j}$ on candidate sets
specializes to XOR, AND, OR, Horn, and other monotone Boolean
primitives.

\item \textbf{Density evolution and concentration (Theorems
\ref{thm:de} and \ref{thm:de-deterministic}).}  For every fixed $L$, the empirical value-conditioned
message erasure rates on a bounded-degree role-typed configuration
ensemble converge in probability to a deterministic recursion.  The
recursion is value-conditioned because general Boolean factors need
not be value-symmetric.  The XOR specialization (\Cref{cor:xor-de})
recovers the classical LDPC-BEC map; the AND specialization
(\Cref{prop:and-de}) exposes a positive-versus-negative-certificate
asymmetry of realistic verifiers.  The McDiarmid concentration
argument is channel-agnostic; the present recursion is its
specialization to the certifying-layer logical-forcing rule.  A
companion result (\Cref{thm:de-deterministic}) shows that the same DE
prediction holds on \emph{deterministic} locally-tree-like graphs,
with no appeal to graph randomness, so the theory also covers deployed
systems whose task graph is fixed rather than sampled.

\item \textbf{Threshold, stability, and non-interchangeability of the
three erasure tiers (Theorem \ref{thm:threshold}, Corollary
\ref{cor:zero-stability}, Propositions \ref{prop:noninterchange} and
\ref{prop:noninterchange-roles}).}  We
prove existence of a residual fixed point of the value-conditioned DE
recursion and a monotone target-reliability threshold along any
monotone reliability path.  In regimes where the zero-erasure state
is a fixed point (e.g., the noiseless-verifier noiseless-channel
limit recovering LDPC-BEC), local spectral stability holds with
Gelfand-rate contraction.  The Non-Interchangeability Proposition
shows that the three erasure tiers cannot be collapsed into a single
effective scalar: the parameter Jacobian has differential rank at
least two on a generic open parameter set.  With heterogeneous roles
the separation is sharper: a single check role serving two variable
roles gives rank at least three (\Cref{prop:noninterchange-roles}),
which distinguishes verifier-side erasure from the return reasoning
channel.  This non-interchangeability is the structural distinction
from multi-edge-type LDPC and the noisy-message-passing-decoder line
(\Cref{tab:related-comparison}).

\item \textbf{Certificate-stopping sets (Theorem \ref{thm:certstop}).}
We define certificate-stopping sets for general Boolean verifier
factors and prove that the terminal unresolved set of the deterministic
peeling decoder is the unique maximal certificate-stopping set.  The
XOR specialization (\Cref{cor:xor-stopping}) recovers the classical
stopping-set characterization; the AND specialization
(\Cref{cor:and-stopping}) gives a positive-versus-negative certificate
condition.

\item \textbf{Separating augmentation (Theorem \ref{thm:augmentation}).}
We define $k$-separating augmentations of the verifier graph and
prove that any $k$-separating augmentation eliminates all
certificate-stopping sets of size at most $k$.  A noisy-augmentation
corollary (\Cref{cor:noisy-aug}) quantifies fault tolerance via a
union bound; the XOR specialization recovers the classical
two-edge-connected freeing-set construction.

\item \textbf{Cost-constrained architecture optimization (Theorem
\ref{thm:optimization}).}  We formulate role-and-degree optimization
under a budget on agent counts, verifier invocations, and
communication edges.  We give existence (Weierstrass), asymptotic
optimality of DE-optimized designs, finite-to-infinite-round
consistency under uniform convergence, backward-mode adjoint
sensitivity equations for $\nabla_\lambda \Pdee^{(L)}(\lambda)$, and
KKT necessary conditions whose multiplier is the shadow price of
additional budget.  A budget-monotonicity corollary
(\Cref{cor:budget}) follows.  The adjoint/KKT framework is
channel-agnostic; the same shadow-price interpretation carries over to
non-erasure DE recursions when those are derived.

\item \textbf{Calibration protocol and applications.}  We give a
calibration protocol mapping each model parameter
($\epsV[r], \epsC[s], \eta_{r,s}$, template proportions, value
priors) to an operational quantity already logged by deployed agent
systems (\Cref{sec:calibration}).  The protocol is traced through
the Hilbert proof-agent architecture, multi-agent code generation
(CodeR, SWE-bench), and structured-output debate
(\Cref{sec:applications}).

\item \textbf{Local-soundness converse on the computation tree
(Theorem \ref{thm:converse}).}  Within the class of $T$-round
\emph{sound} (certifying) local message-passing protocols on the
role-typed configuration ensemble, no protocol (including
non-extrinsic, randomized, or soft-information sound variants) can
asymptotically leave fewer variables unresolved at its terminal
output than the value-conditioned logical-forcing decoder of
\Cref{thm:de}.  The argument localizes
each per-variable decision to the depth-$2T$ computation tree and
identifies the unique sound Bayes-optimal estimator there with the
logical-forcing decoder.  A stronger Fano-cut-set converse against
unbounded-alphabet local protocols without the soundness restriction
is sketched in \Cref{sec:outlook}.
\end{enumerate}

\subsection{Organization of the paper}
\label{sec:intro-organization}

\Cref{sec:related} reviews related work; \Cref{sec:notation} fixes
notation and the configuration model.  \Cref{sec:model} defines the
role-typed Boolean-verifier-node model.  \Cref{sec:de} states and
proves the density-evolution and concentration theorem
(\Cref{thm:de}) with its XOR and AND specializations.
\Cref{sec:threshold} treats residual fixed points, thresholds,
stability, and the Non-Interchangeability Proposition;
\Cref{sec:stopping} characterizes finite-length failures by
certificate-stopping sets; \Cref{sec:augmentation} proves the
separating-augmentation theorem.  \Cref{sec:optimization} formulates
cost-constrained architecture optimization, and \Cref{sec:converse}
proves the local-soundness converse.  \Cref{sec:calibration},
\Cref{sec:numerical}, and \Cref{sec:applications} give the calibration
protocol, the numerical validation, and the worked applications;
\Cref{sec:outlook} discusses extensions, including the
deterministic-graph density-evolution theorem
(\Cref{thm:de-deterministic}) for fixed task graphs and a
Fano-cut-set converse, and \Cref{sec:conclusion} concludes.  Appendices
collect auxiliary monotonicity facts and the detailed concentration
bounds.

\rhead{Reader map.}  Because the paper targets both the
information-theory and the AI-systems communities, two reading
paths are useful.  Readers primarily interested in the sparse-graph
theory may focus on \Cref{sec:notation,sec:model,sec:de,%
sec:threshold,sec:stopping,sec:augmentation,sec:optimization,%
sec:converse}.  Readers primarily interested in AI-agent
implications may first read \Cref{sec:intro},
\Cref{sec:calibration}, \Cref{sec:applications}, and
\Cref{sec:outlook}, then return to the formal model.  The toy
example in \Cref{sec:intro-toy} is intended as a common
entry point for both audiences.

\section{Related Work}\label{sec:related}

\subsection{Factor graphs, LDPC codes, and density evolution}

Factor graphs provide a general language for representing global
functions as products of local functions and for deriving local
message-passing algorithms \cite{kschischang2001factor}.  LDPC codes,
introduced by Gallager \cite{gallager1962ldpc}, are the canonical
sparse-graph coding example.  Density evolution
\cite{luby1998analysis,richardson2001capacity,richardson2001design,%
richardson2008modern} gives sharp asymptotic predictions for
message-passing decoding over memoryless channels and has been central
to the design of capacity-approaching irregular LDPC codes.  On the
binary erasure channel (BEC), the recursion is especially transparent
because messages are either known bits or erasures.

Finite-length BEC performance is governed by stopping sets
\cite{di2002finite}.  Related work has studied improved decoding and
puncturing for LDPC codes \cite{pishro2004decoding,pishro2007punctured},
unequal error protection \cite{pishro2007nonuniform}, generalized LDPC
(GLDPC) codes \cite{paolini2010gldpc}, noisy message-passing decoders
\cite{tarighati2015noisy,dupraz2021noisy}, and absorbing sets for
non-erasure channels \cite{dolecek2010absorbing}.  Multi-edge-type LDPC
codes \cite{richardson2008modern} extend the analysis to ensembles in
which different message types follow different update rules.

The present paper uses the same broad toolkit, but the object being
analyzed is not a communication code.  The graph describes an agent
architecture, the check nodes are noisy verifier agents, the Boolean
factor semantics need not be linear, and three independent erasure
tiers act on three structurally different positions in the message
flow.

Density evolution has also been extended beyond symmetric channels.
For asymmetric memoryless channels, such as the Z-channel, the
all-zero-codeword reduction is no longer valid, and the analysis must
track message distributions conditioned on the transmitted value, or
restore a usable symmetry through an equivalent coset construction
\cite{wang2005asymmetric}.  Our value-conditioned recursion is similar
in spirit, tracking separate states for target values $0$ and $1$, but
the source of the asymmetry is reversed.  In asymmetric-channel density
evolution the parity checks remain linear and value-symmetric while the
channel law is asymmetric; here the observation channel is a sound,
value-independent erasure channel and the asymmetry is induced by the
Boolean verifier functions themselves.  Moreover, the coset construction
that restores symmetry for linear codes does not apply to non-linear
factors such as AND.  \Cref{thm:de} should therefore be read as an
adaptation of standard density-evolution concentration methods to
role-typed Boolean verifier nodes, not as a claim that
value-conditioning is itself new.

\subsection{Multi-agent reasoning, verifiers, and agentic workflows}

Multi-agent language-model systems have been proposed for debate,
critique, self-refinement, search, and tool use
\cite{du2024debate,khan2024debate,kenton2024oversight,%
shinn2023reflexion,yao2023react,yao2023tree,wang2023selfconsistency}.
Verifiers and process supervision play an important role in
mathematical reasoning and step-by-step problem solving
\cite{cobbe2021verifiers,lightman2023step}.  Choi \emph{et al.}\
\cite{choi2025debate} analyze LLM debate as a martingale belief process
on structured tasks.

Formal theorem proving and verified code generation provide especially
clean testbeds because proof assistants and compilers can act as
strong local verifiers \cite{zheng2021minif2f,azerbayev2023proofnet,%
lin2024leanstar,varambally2025hilbert,thakur2025clever,%
dougherty2025proving}.  Software-engineering agents and benchmarks
such as SWE-bench and SWE-agent \cite{jimenez2024swebench,%
yang2024sweagent,chen2024coder} show that realistic code tasks already
have graph-like decompositions and tool-mediated verification.

Most of these works are empirical or algorithmic, asking which
protocol performs well on a benchmark; theoretical analyses include
the martingale framing of debate \cite{choi2025debate} and the
decision-theoretic Bayes-dominance bound for delegated multi-agent
DAGs \cite{ao2026reliability}.  This paper asks a complementary
information-theoretic question: after abstracting an agent system
into roles, local Boolean verifiers, and inter-agent communication
channels, what does the architecture imply about asymptotic
recoverability and finite-length failure patterns?

\subsection{What is new}

The framework is related to multi-edge-type LDPC, generalized LDPC,
and noisy-message-passing-decoder analyses, but differs in five
ways.

\emph{First}, roles are semantic objects: they represent proposer,
verifier, retriever, test-runner, or judge populations with different
reliabilities and costs and a built-in calibration path to logged
trace data.  \emph{Second}, verifier-side erasure is modeled separately
from variable-side erasure; this is essential for agent systems where
the verifier itself is an imperfect computational object.  \emph{Third},
communication fidelity is role-pair dependent and acts on reasoning
artifacts rather than channel symbols; the third tier $\eta_{r,s}$
is not separated as a distinct erasure tier in classical LDPC.  \emph{Fourth}, check nodes compute arbitrary Boolean verifier functions
and the recovery rule is value-conditioned logical forcing on
set-valued messages.  GLDPC handles non-parity local constraints,
but its component codes are linear, hence value-symmetric; the
AND-monotone case here is non-symmetric and exposes a
positive-versus-negative certificate asymmetry not central to the
GLDPC formulation.  \emph{Fifth}, the design payoff is the agent-operational
analog of multi-knob sparse-graph design.  LDPC design is itself
multi-knob (irregular degree distributions, MET edge types,
protograph parameters, UEP-LDPC reliability classes, GLDPC sub-code
choice); what is new here is that the knobs have agent-system
semantics (role mix, role-typed degrees, verifier reliability,
role-pair communication fidelity).  The Non-Interchangeability
Proposition formalizes that no change of variable absorbs them into
one effective scalar, and adjoint-derived KKT shadow prices identify
which resource constraint is binding at any operating point.

\Cref{tab:related-comparison} summarizes which of these axes are
present in adjacent literatures and which are new here.

\begin{table}[t]
\centering
\caption{Where the present framework sits relative to adjacent
sparse-graph literatures.  ``$\checkmark$'' indicates that the axis
is centrally treated in the cited line; ``$-$'' indicates that the
axis is not separated with the operational meaning used here, even
when related machinery exists; ``partial'' indicates that the axis
is treated but in a different form.  The aim is to locate the
distinctive combination of axes, not to claim invention of any one
of them in isolation.}
\label{tab:related-comparison}
\renewcommand{\arraystretch}{1.18}
\small
\begin{tabular}{@{}lcccc@{}}
\toprule
\textbf{Axis} & \textbf{MET-LDPC} & \textbf{GLDPC} &
\textbf{Noisy-MP} & \textbf{This paper} \\
& \cite{richardson2008modern} & \cite{paolini2010gldpc} &
\cite{tarighati2015noisy,dupraz2021noisy} & \\
\midrule
Role-typed channels & $\checkmark$ & $-$ & $-$ & $\checkmark$ \\
Boolean verifier nodes with logical-forcing semantics
   & $-$ & partial & $-$ & $\checkmark$ \\
Value-conditioned DE under non-symmetric checks & $-$ & $-$ & $-$ & $\checkmark$ \\
Verifier-side erasure as a separate tier & $-$ & $-$ & partial & $\checkmark$ \\
Role-pair reasoning channels as a separate tier & $-$ & $-$ & $-$ & $\checkmark$ \\
Three reliability tiers as separate design knobs & $-$ & $-$ & $-$ & $\checkmark$ \\
Protocol for estimating model parameters from agent traces & $-$ & $-$ & $-$ & $\checkmark$ \\
\bottomrule
\end{tabular}

\vspace{2pt}
\begin{minipage}{0.92\linewidth}
\footnotesize \emph{Note.} The ``value-conditioned DE'' row refers to
conditioning induced by non-symmetric Boolean verifier functions.
Value- or codeword-conditioned density evolution is itself known for
asymmetric memoryless channels \cite{wang2005asymmetric}; the
distinction is discussed in \Cref{sec:related} and
\Cref{rem:value-cond}.
\end{minipage}
\end{table}

\section{Notation and Preliminaries}\label{sec:notation}

\subsection{Sets, vectors, and indicators}

We use $\R, \N, \Z$ for the reals, naturals, and integers; $\PR$ and
$\E$ for probability and expectation under whatever measure is in
context.  For a positive integer $K$, $[K] \defeq \{1, \ldots, K\}$.
Boldface lowercase letters denote vectors; calligraphic uppercase
letters denote sets.  Indicator functions are written $\one\{\cdot\}$.
For a function $f$ of several arguments, we sometimes write $f(\cdot)$
to indicate a slot to be filled.

\subsection{Alphabets and the unresolved symbol}

Three alphabets recur throughout the paper.  Hidden subclaim values
lie in $\mathcal{X} \defeq \{0, 1\}$.  Set-valued messages and
variable-side observations take values in the set-valued message
alphabet
\begin{equation}
\label{eq:msg-alphabet}
   \Mset \defeq \big\{ \{0\}, \{1\}, \Uset \big\},
   \qquad
   \Uset \defeq \{0, 1\},
\end{equation}
where a singleton $\{b\}$ encodes a sender's certification of value
$b$ and $\Uset$ encodes ``unresolved.''  Verifier outputs are
scalar, taking values in
\begin{equation}
\label{eq:verifier-alphabet}
   \mathcal{Z} \defeq \{0, 1, *\},
\end{equation}
with $*$ denoting an erased verifier output.  The two erasure
markers are kept distinct: a message-side $\Uset$ encodes ``both
target values remain possible'' on a candidate-set update, whereas a
scalar verifier-side $*$ encodes ``the verifier produced no usable
Boolean observation.''  In informal narrative passages we
occasionally write ``$?$'' for ``erased'' when the type of erasure
is clear from context; the formal model uses $\Uset$ and $*$ as
defined above.  The reader more comfortable with scalar erasure
messages should mentally identify $\{0\} \leftrightarrow 0$,
$\{1\} \leftrightarrow 1$, $\Uset \leftrightarrow {?}$; the
set-valued representation is more natural for general Boolean
factors because the logical-forcing operator
(\Cref{eq:intro-forcing}) acts directly on candidate sets.

\rhead{Closure of $\Mset$ under intersection.}  The variable-side
update intersects incoming candidate sets.  In principle two
contradictory singletons $\{0\}$ and $\{1\}$ would intersect to
$\emptyset \notin \Mset$, but on \emph{sound} message configurations
this never occurs: \Cref{lem:soundness} establishes the invariant
that every singleton message contains the true value, so
contradictory singletons cannot co-occur and the update remains in
$\Mset$.

\subsection{Configuration-model background}

Given finite multisets of variable-node degrees and check-node
degrees with matching socket counts, a uniformly random pairing of
sockets defines a bipartite multigraph; this is the standard
configuration model \cite{richardson2008modern,mezard2009info}.
Under bounded-degree assumptions, conditioning on the graph being
simple does not change the bounded-radius local neighborhood law, and
hence does not change the density-evolution recursion for any fixed
number of iterations $L$.  This holds even though the two ensemble
laws are not, in general, globally close in total variation.  We accordingly state the asymptotic results in this paper
on the role-typed configuration-model multigraph ensemble; the
simple-graph-conditioned version has the same local limit and the
same DE prediction.  We work with a role-typed and template-typed
extension of the configuration model introduced formally in
\Cref{sec:model}.  The key structural fact is that for any bounded
radius $R$, the neighborhood of a uniformly chosen directed edge in
such a graph converges in distribution to the corresponding typed
Galton-Watson computation tree as $n \to \infty$.  This is the
\emph{locally tree-like} property; standard references are
\cite{richardson2008modern} and \cite{mezard2009info}.

\section{The Role-Typed Boolean-Verifier-Node Model}\label{sec:model}

\subsection{Graphs, roles, and Boolean verifier templates}
\label{sec:model-graph}

For each problem size $n \in \N$, let
\begin{equation}
\label{eq:model-Gn}
   \Gn = (\Vset_n \cup \Cset_n,\, \Eset_n)
\end{equation}
be a bipartite graph.  Variable nodes $i \in \Vset_n$ represent
subclaims; check nodes $a \in \Cset_n$ carry local Boolean
verifier functions.  We use \emph{check node}, \emph{verifier node},
and \emph{Boolean verifier factor} interchangeably for these nodes,
and \emph{verifier function} for the local map they compute.  We assume
$|\Vset_n| = n$ and $|\Cset_n| = \lfloor \alpha\, n \rfloor$ for a
fixed scaling parameter $\alpha > 0$.

There are finite role sets $\RolesV$ for variable agents and $\RolesC$
for check agents.  A variable role represents a population with
distinguishable observation-noise statistics, proposer, retriever,
lemma generator, code writer, sensor.  A check role represents a
population with distinguishable verifier statistics, proof checker,
test runner, judge, type checker, plausibility checker.  Let
$\rolemap(i) \in \RolesV$ be the role of variable $i$ and
$\rolemap(a) \in \RolesC$ be the role of check $a$.

The local verifier semantics are specified by a finite set $\Templates$
of \emph{templates}.  A template is a tuple
\begin{equation}
\label{eq:template}
   \theta = \big(s_\theta,\, d_\theta,\, r_{\theta,1}, \ldots,
   r_{\theta, d_\theta},\, f_\theta\big),
\end{equation}
in which $s_\theta \in \RolesC$ is a check role, $d_\theta \in \N$ is
the arity, $r_{\theta, j} \in \RolesV$ is the role of the variable
expected at socket $j$, and
\begin{equation}
\label{eq:template-fn}
   f_\theta : \{0, 1\}^{d_\theta} \to \{0, 1\}
\end{equation}
is the local Boolean verifier function.  A check node with template
$\theta$ has ordered sockets $j = 1, \ldots, d_\theta$, and socket $j$
must connect to a variable of role $r_{\theta, j}$.  We write
\begin{equation}
\label{eq:socket-type}
   \tau = (\theta, j) \in \Sockets, \qquad
   \socketrole{\tau} = r_{\theta, j}, \qquad
   \socketcheckrole{\tau} = s_\theta,
\end{equation}
for the socket type, the variable role at the socket, and the check
role of the host template.

\begin{example}[Common verifier semantics]
\label{ex:common-verifiers}
The same framework covers several common local verifiers as templates
$\theta$ with different choices of $f_\theta$.
\begin{itemize}[leftmargin=2em]
\item \emph{XOR / parity:} $f_\theta(x_1, \ldots, x_d) =
x_1 \oplus \cdots \oplus x_d$.  This is the LDPC-clean linear
baseline; a check is satisfied iff the parity of the inputs equals
the observed value.
\item \emph{AND / monotone conjunction:} $f_\theta(x_1, \ldots, x_d) =
x_1 \wedge \cdots \wedge x_d$.  This models a passing local test, a
proof-step type-check, or a composition check: the check is satisfied
iff every input is correct.
\item \emph{OR:} $f_\theta(x_1, \ldots, x_d) = x_1 \vee \cdots \vee
x_d$.  Symmetric to AND for negated inputs.
\item \emph{Implication:} $f_\theta(x_1, x_2) = \one\{x_1 \le x_2\}$,
or equivalently $x_1 \Rightarrow x_2$.  This models a local
dependency rule.
\item \emph{Horn clause:} $(x_1 \wedge \cdots \wedge x_m) \Rightarrow
y$.  This models proof-rule semantics in formal systems.
\end{itemize}
\end{example}

\subsection{Hidden values and erasure-only observations}
\label{sec:model-observations}

Each variable has a hidden truth value $\Xstar_i \in \{0, 1\}$.  In the
general (non-XOR) case we assume a role-dependent product prior with
hidden values drawn independently of the degree sequence conditional
on role,
\begin{equation}
\label{eq:prior}
   \PR\big\{\Xstar_i = 1 \,\big|\, \rolemap(i) = r\big\}
   = \Vbias_r, \qquad r \in \RolesV,
   \qquad
   \Xstar_i \perp\!\!\!\perp D_i \mid \rolemap(i).
\end{equation}
The independence between $\Xstar_i$ and the degree vector $D_i$ at
fixed role is required because density evolution samples variables
through sockets, not uniformly through nodes; without value-degree
independence, the socket-level value distribution at a role-$r$
neighbor would differ from the node-level prior $\Vbias_r$ and the
recursion would carry an additional bias.  The same density-evolution
limits also apply to deterministic sequences of hidden vectors
provided the \emph{socket-conditional} empirical value frequencies
converge to $\{\Vbias_r\}$ (a strictly stronger condition than
role-conditional convergence).  The prior is unnecessary for the XOR
specialization, because XOR erasure dynamics do not depend on the
actual hidden vector (one can analyze the all-zero codeword without
loss of generality; see \Cref{cor:xor-de}).  It is needed for
non-symmetric Boolean factors because the probability that a check
can logically force a target depends on the values of the
neighboring true subclaims.  The product form carries one further
consequence used in \Cref{prop:and-de}: it makes the true inputs
$(\Xstar_i)_{i \in \partial a}$ at a check mutually independent, which
is what yields the product forcing probability
$\forceprob_{\ell, \theta, j}^{(1)} = \prod_{k \neq j} \Vbias_{r_{\theta, k}}$
there.  Local tree-likeness makes the \emph{messages} arriving at a
check independent, but not the co-incident ground-truth values, so the
product prior is a mean-field, first-order approximation of the joint
task distribution: a correlated task law with the same
socket-conditional marginals $\{\Vbias_r\}$ preserves the
density-evolution marginals stated here, while replacing this product
by the corresponding conditional joint at each check.

A variable node $i$ of role $r$ has private observation
\begin{equation}
\label{eq:var-obs}
   \Aobs_i = \begin{cases}
   \{\Xstar_i\}, & \text{with probability } 1 - \epsV[r], \\
   \Uset, & \text{with probability } \epsV[r].
   \end{cases}
\end{equation}
A check node $a$ with template $\theta$ has true verifier output
\begin{equation}
\label{eq:check-true}
   T_a = f_\theta\big((\Xstar_i)_{i \in \partial a}\big),
\end{equation}
where the variables in $\partial a$ are ordered according to the
template sockets.  It observes
\begin{equation}
\label{eq:check-obs}
   \Zobs_a = \begin{cases}
   T_a, & \text{with probability } 1 - \epsC[s_\theta], \\
   *, & \text{with probability } \epsC[s_\theta],
   \end{cases}
\end{equation}
where $*$ denotes an erased verifier output.  All variable-side and
verifier-side erasures are independent conditional on the graph and
the hidden vector.  This is a first-order, conditionally memoryless
approximation: correlated agent failures (for example, two proposer
agents backed by the same foundation model failing on the same
subclaim) sit outside the present setting and are deferred to the
limitations and dependence-relaxation discussion in
\Cref{sec:outlook-limitations,sec:outlook-belief}.

\subsection{Reasoning channels}
\label{sec:model-reasoning-channel}

For each ordered role pair $(r, s) \in (\RolesV \cup \RolesC)^2$, a
\emph{reasoning-channel fidelity} $\eta_{r, s} \in [0, 1]$ specifies
the probability that a singleton message produced by an agent of role
$r$ is delivered intact to an agent of role $s$.  The role-pair
channel kernel is
\begin{equation}
\label{eq:channel}
   \Qchan_{r, s}(\widetilde M \mid M)
   = \begin{cases}
   \eta_{r, s}, & M = \{b\},\, \widetilde M = \{b\},\, b \in \{0,1\}, \\
   1 - \eta_{r, s}, & M = \{b\},\, \widetilde M = \Uset,\, b \in \{0,1\}, \\
   1, & M = \Uset,\, \widetilde M = \Uset, \\
   0, & \text{otherwise}.
   \end{cases}
\end{equation}
Thus the base model has erasures but no flips.  This is the analogue
of the BEC in classical coding theory.  Confident wrong messages would
lead to a different, non-erasure theory analogous to absorbing-set
analysis on the BSC \cite{dolecek2010absorbing}; we discuss this
extension briefly in \Cref{sec:outlook}.  Each directed edge has its
own independent channel-erasure variable drawn once at $t = 0$ in
the persistent-edge convention used throughout the paper; an
alternative i.i.d.-per-round channel would change the analysis (each
round draws fresh gates) and is outside the present setting.

\subsection{Logical forcing: the unifying check-to-variable update}
\label{sec:model-forcing}

Suppose check $a$ has template $\theta$ and target socket $j$.
Suppose $\Zobs_a = z \in \{0, 1\}$ is available (not erased) and the
incoming variable-to-check messages are candidate sets
$M_k \in \Mset$, $k \neq j$.  The \emph{logical-forcing operator} is
\begin{equation}
\label{eq:forcing}
\force_{\theta, j}(z;\, M_{-j})
   \defeq \big\{ b \in \{0, 1\} :\;
   \exists\, x_k \in M_k,\, k \neq j,\;
   f_\theta(x_1, \ldots, x_{j-1}, b, x_{j+1}, \ldots, x_{d_\theta}) = z \big\}.
\end{equation}
If $\Zobs_a = *$, we set $\force_{\theta, j} = \Uset$ by convention.
This operator is the produced check-to-variable message at edge
$(a, j)$ (the formal update is \Cref{eq:check-update} below): a
singleton when the observed verifier value and the unknown-input
candidate sets logically force the target, otherwise $\Uset$.  The
reasoning channel can further replace this message by $\Uset$ on the
directed edge $s_\theta \to r_{\theta, j}$ before it reaches the
variable.  This is local Boolean reasoning with unknown inputs
represented by candidate sets; it is the natural generalization of
the parity-inversion rule of LDPC-BEC decoding to arbitrary Boolean
verifier functions.

\emph{Concrete examples.}  For an AND check that reports
$\Zobs_a=1$, every input must be $1$ for the conjunction to hold,
so $\force_{\theta, j}(1; M_{-j}) = \{1\}$ regardless of $M_{-j}$.
For the same AND check reporting $\Zobs_a = 0$, the target is forced
to $0$ only if every other input has already been resolved to $1$:
$\force_{\theta, j}(0; M_{-j}) = \{0\}$ when $M_{i'} = \{1\}$ for all
$i' \neq j$, and $\force_{\theta, j}(0; M_{-j}) = \Uset$ otherwise.
This is the positive-versus-negative-certificate asymmetry that
\Cref{prop:and-de} quantifies in DE form.  For an XOR check, the
operator is value-symmetric: $\force_{\theta, j}(z; M_{-j})$ is a
singleton if and only if every $M_{i'}$ for $i' \neq j$ is itself a
singleton, in which case it returns the unique $b$ satisfying
$b \oplus \bigoplus_{i' \neq j} x_{i'} = z$.

Let $\widetilde V_{i \to a}^{(\ell)}$ be the variable-to-check message
\emph{received} by check $a$ after the role-pair channel acts, and let
$\widetilde C_{a \to i}^{(\ell)}$ be the check-to-variable message
received by variable $i$.  Extrinsic edge-specific updates are
\begin{align}
\label{eq:check-update}
   C_{a \to i}^{(\ell)}
   &= \force_{\theta, j}\Big( \Zobs_a;\,
   \big(\widetilde V_{k \to a}^{(\ell)}\big)_{k \in \partial a \setminus \{i\}}\Big), \\
\label{eq:var-update}
   V_{i \to a}^{(\ell+1)}
   &= \Aobs_i \cap \bigcap_{c \in \partial i \setminus \{a\}}
      \widetilde C_{c \to i}^{(\ell)},
\end{align}
where socket $j$ on $a$ corresponds to neighbor $i$.  The final
estimate after $L$ rounds uses all neighboring checks:
\begin{equation}
\label{eq:final-est}
   \widehat M_i^{(L)}
   = \Aobs_i \cap \bigcap_{a \in \partial i}
      \widetilde C_{a \to i}^{(L)}.
\end{equation}
The bit-erasure rate is
\begin{equation}
\label{eq:empirical-bit}
   \Pbit^{(L)}(\Gn)
   = \frac{1}{n} \sum_{i \in \Vset_n} \one\big\{\widehat M_i^{(L)} = \Uset\big\}.
\end{equation}

\begin{lemma}[Soundness invariant]
\label{lem:soundness}
On every finite graph and for every erasure-only realization of the
observations and reasoning channels, every message about variable $i$
at every iteration $\ell \ge 0$ contains the true value $\Xstar_i$.
Hence every singleton message about $i$ equals $\{\Xstar_i\}$, and no
message is ever empty.
\end{lemma}

\begin{proof}
Induction on $\ell$.  Private observations are $\{\Xstar_i\}$ or
$\Uset$; a non-erased check output $\Zobs_a = T_a$ forces only values
consistent with the true verifier output, and $\Xstar_i$ is consistent
by \eqref{eq:check-true}; reasoning channels either preserve a
singleton or replace it by $\Uset$; and the variable update
\eqref{eq:var-update} intersects sets that all contain $\Xstar_i$.
\end{proof}

\Cref{lem:soundness} is the basic soundness invariant used
throughout the analysis: messages can be erased, never wrong.  This
is the agent-network analog of erasure-channel soundness in
classical coding, and it is what makes the BEC-style analysis
possible in this setting.

\subsection{Role-typed configuration ensemble}
\label{sec:model-ensemble}

We use a bounded-degree role-typed configuration model.  Variable
nodes of role $r$ have a random socket-count vector
\begin{equation}
\label{eq:degree-law}
   D^{(r)} = \big(D_\tau^{(r)}\big)_{\tau \in \Sockets_r},
   \qquad
   \Sockets_r \defeq \{\tau \in \Sockets : \socketrole{\tau} = r\},
\end{equation}
with bounded support $D_\tau^{(r)} \le D_{\max}$ for some fixed
$D_{\max} \in \N$ independent of $n$, and minimum-degree assumption
$\sum_\tau D_\tau^{(r)} \ge 2$ almost surely.  The minimum-degree
assumption ensures that every variable is touched by at least two
checks; it
is needed for the threshold theorem (\Cref{thm:threshold}) and the
zero-erasure stability corollary, but the basic
density-evolution recursion (\Cref{thm:de}) holds without it.
Variable nodes are partitioned into roles with asymptotic
proportions $\{\pi_r^V\}_{r \in \RolesV}$, and check nodes are
assigned templates in $\Templates$ with asymptotic proportions
$\{\pi_\theta^C\}_{\theta \in \Templates}$.  The ensemble parameters
$\big(\{\pi_r^V\}, \{\mathsf{P}_D^{(r)}\}, \{\pi_\theta^C\}, \alpha\big)$
are \emph{admissible} if, for every socket type
$\tau = (\theta, j)$,
\begin{equation}
\label{eq:socket-balance}
   \pi^V_{\socketrole{\tau}} \cdot
   \E\!\big[D^{(\socketrole{\tau})}_\tau\big]
   = \alpha \cdot \pi^C_\theta,
\end{equation}
where $\alpha \defeq |\Cset_n| / |\Vset_n|$ is the asymptotic
check-to-variable ratio.  We assume admissibility throughout.  At
finite $n$, socket balances are enforced exactly by a
rounding-and-matching step (with $O(\sqrt{n})$ rounding
fluctuations), and sockets of each type are then paired uniformly
at random.

For a variable reached by following a uniformly chosen socket of type
$\tau$, let $D^{(r), \tau, \mathrm{ex}}$ be the size-biased excess
degree vector of the variable, after deleting the arrival socket.  For
a uniformly chosen variable node of role $r$, let $D^{(r), \mathrm{node}}$
be its full degree vector.

\begin{lemma}[Locally tree-like]
\label{lem:tree-like}
Let $R > 0$ be a fixed integer.  As $n \to \infty$, the depth-$R$
neighborhood of a uniformly chosen directed socket of type $\tau$ in
$\Gn$ converges in total variation to the corresponding typed
Galton-Watson computation tree.  At a variable node entered through
a socket of type $\tau'$, the remaining socket-degree vector has law
$D^{(\socketrole{\tau'}), \tau', \mathrm{ex}}$ (and the root variable
has full vector $D^{(\socketrole{\tau}), \mathrm{node}}$).  At a
check node entered through socket $(\theta, j)$, the template is
deterministically $\theta$ and the remaining socket types are
deterministic $\{(\theta, k) : k \neq j\}$; the template proportions
$\{\pi_\theta^C\}$ enter only through the law of which socket type a
uniformly chosen variable-side socket connects to, via the
socket-balance equation \eqref{eq:socket-balance}.
\end{lemma}

\Cref{lem:tree-like} is standard for bounded-degree configuration
models; we record it here only to fix notation.  The role-typed
extension is mechanical, with conditioning on type at every step.

\subsection{Summary of notation}
\label{sec:model-notation-summary}

The symbols used throughout the analysis are collected in
\Cref{tab:notation-summary}.  Two conventions are worth flagging.
First, $p$-prefixed symbols are erasure probabilities (larger
$p$ means worse reliability); $\eta$-prefixed symbols are
delivery probabilities (larger $\eta$ means better reliability).  Second, the message alphabet
$\Mset = \{\{0\},\{1\},\Uset\}$ uses $\Uset = \{0,1\}$ for set-valued
unresolvedness, while the scalar verifier alphabet
$\mathcal{Z} = \{0,1,*\}$ uses $*$ for an erased verifier output;
the two erasure markers encode operationally different states
(see \Cref{sec:notation}).

\begin{table}[t]
\centering
\caption{Summary of notation used in the model and analysis.}
\label{tab:notation-summary}
\renewcommand{\arraystretch}{1.18}
\begin{tabular}{@{}lp{0.62\linewidth}@{}}
\toprule
\textbf{Symbol} & \textbf{Meaning} \\
\midrule
$\Xstar_i \in \{0, 1\}$ & Hidden truth value of subclaim $i$. \\
$\Aobs_i \in \{\{\Xstar_i\}, \Uset\}$ & Variable-side observation
(soundness preserved). \\
$\Zobs_a \in \{f_\theta(\Xstar_{\partial a}), *\}$ & Verifier
observation (soundness preserved). \\
$f_\theta : \{0,1\}^{d_\theta} \to \{0, 1\}$ & Boolean verifier
template function. \\
$\force_{\theta, j}(z; M_{-j}) \subseteq \{0, 1\}$ & Logical-forcing
operator at socket $j$ given observed output $z$. \\
$\epsV[r] \in [0, 1)$ & Variable-side erasure probability (role $r$
abstention rate). \\
$\epsC[s] \in [0, 1)$ & Verifier-side erasure probability (check role
$s$ failure rate). \\
$\eta_{r, s} \in (0, 1]$ & Reasoning-channel delivery probability
(directed; $r \to s$). \\
$\Vbias_r \in (0, 1)$ & Role-$r$ value prior $\PR\{\Xstar_i = 1\}$. \\
$\pmsg_{\ell, \tau}^{(b)} \in [0, 1]$ & Variable-to-check erasure
probability at iteration $\ell$ at socket type $\tau$, conditional
on $X = b$. \\
$\hmsg_{\ell, \tau}^{(b)} \in [0, 1]$ & Received check-to-variable
erasure probability (post-channel). \\
$\forceprob_{\ell, \theta, j}^{(b)} \in [0, 1]$ & Forcing
probability at template $\theta$, socket $j$, value $b$. \\
$\DEmap_\lambda$ & Density-evolution map; $\bm p_{\ell+1} =
\DEmap_\lambda(\bm p_\ell)$. \\
$\DEjac_\lambda(\bm p)$ & Jacobian of $\DEmap_\lambda$ at $\bm p$
(state Jacobian). \\
$D_{\mathrm{par}}[\DEmap_\lambda]$ & Parameter Jacobian of
$\DEmap_\lambda$, with respect to the erasure-tier parameters. \\
$\Pdee^{(L)}(\lambda)$ & DE prediction of bit-erasure rate after $L$
rounds. \\
$\Pdee^{(\infty)}(\lambda)$ & Residual bit-erasure rate (limit of
$\Pdee^{(L)}$). \\
$\Pbit^{(L)}(\Gn)$ & Empirical bit-erasure rate on graph $\Gn$ after
$L$ rounds. \\
$\Sockets, \Templates$ & Socket-type set, template set. \\
$\RolesV, \RolesC$ & Variable-role set, check-role set. \\
$D_\tau^{\mathrm{ex}}$ & Excess-degree socket law. \\
\bottomrule
\end{tabular}
\end{table}

\section{Density Evolution and Concentration}\label{sec:de}

\subsection{The value-conditioned recursion}
\label{sec:de-recursion}

\emph{Scalar anchor.}  Before the formal recursion, here is the
calculation in scalar form.  The DE state at round $\ell$ is
$\pmsg_{\ell, \tau}^{(b)}$, the value-conditioned probability that a
variable-to-check message of socket type $\tau$ is unresolved, given
target value $b$.  The recursion has three composed steps: (i) a
\emph{persistent} role-pair channel (drawn once at $t = 0$, not
redrawn per round) may erase a message in transit; (ii) a check
fails to certify a target if the verifier output is erased, the
return channel is erased, or the Boolean forcing rule cannot
isolate the target value from the inbound messages; (iii) the target
variable remains unresolved at the next round iff its private
observation is erased and all extrinsic incoming check messages are
erased.  The persistent-channel convention models failures
that do not go away on retry, such as a stable format incompatibility
between two roles; an i.i.d.-per-round channel would give a different
DE map.  The equations below assemble these three steps with the
role-typed degree laws and value-dependent forcing probabilities of
general Boolean factors; symbols are tabulated in
\Cref{sec:model-notation-summary}.

The density-evolution state tracks both socket type and the
underlying hidden value.  For each socket type $\tau \in \Sockets$ and
value $b \in \{0, 1\}$, define
\begin{equation}
\label{eq:p-def}
   \pmsg_{\ell, \tau}^{(b)}
   \defeq \PR\Big\{
   V_{i \to a}^{(\ell)} = \Uset \,\Big|\,
   \Xstar_i = b,\, (i, a) \text{ has socket type } \tau\Big\},
\end{equation}
the extrinsic variable-to-check erasure probability conditional on the
underlying value.  Define $\hmsg_{\ell, \tau}^{(b)}$ analogously for
the received check-to-variable message:
\begin{equation}
\label{eq:h-def}
   \hmsg_{\ell, \tau}^{(b)}
   \defeq \PR\Big\{
   \widetilde C_{a \to i}^{(\ell)} = \Uset \,\Big|\,
   \Xstar_i = b,\, (a, i) \text{ has socket type } \tau\Big\}.
\end{equation}

Consider a socket $\tau = (\theta, j)$.  For another socket $k \neq j$
of the same template $\theta$, conditional on the true value $x_k$,
the incoming variable-to-check message at socket $k$ is the singleton
$\{x_k\}$ with probability
$\eta_{r_{\theta, k}, s_\theta}\big(1 - \pmsg_{\ell, (\theta, k)}^{(x_k)}\big)$
and is $\Uset$ otherwise.  Define
\begin{equation}
\label{eq:pbar}
   \pbar_{\ell, \theta, k}^{(x_k)}
   \defeq 1 - \eta_{r_{\theta, k}, s_\theta}
   \big(1 - \pmsg_{\ell, (\theta, k)}^{(x_k)}\big),
\end{equation}
the effective inbound erasure probability after the role-pair
channel.  The factorization in \eqref{eq:pbar} relies on the
extrinsic update: $V^{(\ell)}_{i \to a}$ is computed from the
depth-$2\ell+1$ computation tree rooted at $i$ \emph{with the
directed edge $i \to a$ removed}, hence is measurable with respect
to a sub-tree disjoint from the persistent channel variable on
edge $i \to a$.  The pre-channel message and the persistent channel
variable are therefore independent, even though the channel is
fixed at $t = 0$.  The same extrinsic argument applies to the
return channel on edge $a \to i$ used in \eqref{eq:hupdate} below.

Conditional on $X_j = b$, the \emph{forcing probability} is
\begin{equation}
\label{eq:forcing-prob}
   \forceprob_{\ell, \theta, j}^{(b)}
   \defeq \PR\Big\{
   \force_{\theta, j}\big(f_\theta(X_1, \ldots, X_{d_\theta});\,
   M_{-j}\big) = \{b\} \,\Big|\, X_j = b\Big\},
\end{equation}
where each $X_k$ ($k \neq j$) is drawn from the role-prior $\Vbias_{r_{\theta, k}}$
independently, and each $M_k$ is the singleton $\{X_k\}$ or $\Uset$
according to \eqref{eq:pbar}.  The received check-to-variable erasure
probability is
\begin{equation}
\label{eq:hupdate}
   \hmsg_{\ell, (\theta, j)}^{(b)}
   = 1 - \big(1 - \epsC[s_\theta]\big)\,
   \eta_{s_\theta, r_{\theta, j}}\,
   \forceprob_{\ell, \theta, j}^{(b)}.
\end{equation}
The variable update is
\begin{equation}
\label{eq:pupdate}
   \pmsg_{\ell+1, \tau}^{(b)}
   = \epsV[\socketrole{\tau}]\,
   \E\!\left[ \prod_{\tau' \in \Sockets_{\socketrole{\tau}}}
   \big(\hmsg_{\ell, \tau'}^{(b)}\big)^{D_{\tau'}^{(\socketrole{\tau}),
   \tau, \mathrm{ex}}} \right],
\end{equation}
the expectation taken over the size-biased excess degree law.  The
initial condition is
\begin{equation}
\label{eq:pinit}
   \pmsg_{0, \tau}^{(b)} = \epsV[\socketrole{\tau}],
\end{equation}
because at iteration zero a variable sends $\Uset$ iff its private
observation is erased.  After $L$ rounds, the density-evolution
prediction for the empirical bit-erasure rate is
\begin{equation}
\label{eq:terminal-de}
   \Pdee^{(L)}
   = \sum_{r \in \RolesV} \pi_r^{\mathrm V}
   \sum_{b \in \{0,1\}} \PR(X = b \mid r)\,
   \epsV[r]\,
   \E\!\left[ \prod_{\tau \in \Sockets_r}
   \big(\hmsg_{L, \tau}^{(b)}\big)^{D_\tau^{(r), \mathrm{node}}} \right],
\end{equation}
with the expectation taken over the full (non-size-biased) degree law.

\emph{In words.}  Equation \eqref{eq:pbar} is the per-edge effective
erasure rate seen by a check from one of its variable neighbors after
the role-pair channel acts.  Equation \eqref{eq:forcing-prob} is the
probability that the logical-forcing operator returns a singleton
$\{b\}$ given that the target value is $b$; this is the step that
depends on the Boolean primitive (XOR vs.\ AND vs.\ Horn).  Equation
\eqref{eq:hupdate} composes verifier erasure, return-channel erasure,
and forcing failure into the received check-to-variable erasure rate.
Equation \eqref{eq:pupdate} is the standard variable update: the
outgoing extrinsic erasure rate equals the variable's own
observation-erasure probability times the probability that every
other incoming
check message is erased, averaged over the excess-degree law.
Equation \eqref{eq:terminal-de} is the prediction for the bit error
rate one would measure on a held-out instance.

\subsection{Concentration}
\label{sec:de-concentration}

\rhead{Empirical statistics.}  Let $E_\tau$ denote the set of
directed variable-to-check edges of socket type $\tau$ in $\Gn$.  For
each socket type $\tau \in \Sockets$ and value $b \in \{0,1\}$,
define the empirical conditional erasure fraction
\begin{equation}
\label{eq:phat-def}
   \widehat\pmsg_{L, \tau}^{(b)}(\Gn)
   \defeq \frac{
   \sum_{(i,a) \in E_\tau} \one\{\Xstar_i = b\}\,
                          \one\{V^{(L)}_{i \to a} = \Uset\}
   }{
   \sum_{(i,a) \in E_\tau} \one\{\Xstar_i = b\}
   }
\end{equation}
on the event that the denominator is positive, and undefined
otherwise.  Let
\begin{equation}
\label{eq:qmass-def}
   q_{\tau, b}
   \defeq \lim_{n \to \infty} \frac{1}{n}\,
   \big|\{(i,a) \in E_\tau : \Xstar_i = b\}\big|,
\end{equation}
the limiting fraction of socket-type-$\tau$ edges incident to
value-$b$ variables.  Under the admissible ensemble parameters of
\Cref{sec:model-ensemble} and the value-degree independence
assumption \eqref{eq:prior}, $q_{\tau, b}$ exists and is positive
for every $(\tau, b)$ with $\Vbias_{\socketrole{\tau}} \in (0, 1)$
and $\E[D_\tau^{(\socketrole{\tau})}] > 0$.  We assume positive
limiting socket-value mass throughout the per-type statements below.

\rhead{Standing assumptions for \Cref{thm:de}.}  For ease of
auditing, the assumptions invoked by the theorem and its proof are:
\begin{enumerate}[label=(A\arabic*),leftmargin=2.2em]
  \item finite role sets $\RolesV, \RolesC$ and finite template set $\Templates$;
  \item bounded degrees $D_{\max} < \infty$ independent of $n$;
  \item exact or asymptotic socket-balance \eqref{eq:socket-balance};
  \item bounded-degree role-typed configuration ensemble
        (\Cref{sec:model-ensemble}), giving the locally tree-like
        property of \Cref{lem:tree-tv};
  \item value-degree independence within roles \eqref{eq:prior}
        (\Cref{thm:de-deterministic} draws hidden values from the
        same prior, independently of the fixed graph);
  \item independent variable-side erasures \eqref{eq:var-obs},
        verifier-side erasures \eqref{eq:check-obs}, and persistent
        role-pair channel erasures \eqref{eq:channel} (each drawn
        once at initialization);
  \item fixed iteration count $L$ as $n \to \infty$.
\end{enumerate}
Each assumption is stated in the corresponding paragraph of
\Cref{sec:model}; we collect them here for clarity, not as new
hypotheses.

\begin{theorem}[Boolean-verifier density evolution and concentration]
\label{thm:de}
Fix $L \in \N$.  Consider the bounded-degree role-typed configuration
ensemble of \Cref{sec:model-ensemble}, with independent
role-dependent hidden values \eqref{eq:prior}, erasure-only variable
observations \eqref{eq:var-obs}, erasure-only verifier observations
\eqref{eq:check-obs}, and persistent role-pair erasure reasoning
channels \eqref{eq:channel}.  Then for every socket type
$\tau \in \Sockets$ and value $b \in \{0, 1\}$ with $q_{\tau, b} > 0$,
\begin{equation}
\label{eq:phat-conv}
   \widehat\pmsg_{L, \tau}^{(b)}(\Gn)
   \xrightarrow[n \to \infty]{\PR}\, \pmsg_{L, \tau}^{(b)},
\end{equation}
where $\pmsg_{L, \tau}^{(b)}$ is defined by
\eqref{eq:pbar}--\eqref{eq:pinit}.  Moreover, for the global
bit-erasure rate,
\begin{equation}
\label{eq:bit-conv}
   \Pbit^{(L)}(\Gn) \xrightarrow[n \to \infty]{\PR}\, \Pdee^{(L)}.
\end{equation}
Quantitatively, the empirical fraction concentrates exponentially
around its own finite-$n$ mean, which sits within a vanishing bias
$b_{n, L}$ of the DE prediction.  Here
$b_{n, L} \defeq C_L / n + \Delta_{n, L}$ is the tree-approximation
bias of \Cref{lem:tree-tv} (with $C_L$ and $\Delta_{n, L}$ as in
\eqref{eq:Delta-n}).  Under the standard $O(\sqrt n)$-rounded
socket-balance construction $\Delta_{n, L} = O(\mathrm{poly}(N_L)/\sqrt n)$,
while $\Delta_{n, L} = 0$ under a deterministic socket-balance
construction, so $b_{n, L} = O(1/\sqrt n)$ in general and $O(1/n)$ in
the deterministic case.  There is then a constant $a_L > 0$ depending only on
$L$, the bounded-degree constant $D_{\max}$, and the fixed ensemble
parameters $(\{\pi_r^V\}, \{\mathsf{P}_D^{(r)}\}, \{\pi_\theta^C\},
\alpha, \{\Vbias_r\})$ such that, for every $t > 0$ and every $n$,
\begin{equation}
\label{eq:bit-tail}
   \PR\!\left\{
   \big| \Pbit^{(L)}(\Gn) - \Pdee^{(L)} \big| > b_{n, L} + t\right\}
   \le 2 \exp\!\left( -a_L\, n\, t^2 \right).
\end{equation}
Since $b_{n, L} \to 0$, the offset is eventually negligible against
any fixed resolution: for every fixed $\delta > 0$ there is a
constant $c_L > 0$ (depending on the same quantities, not on
$\delta$) such that
$\PR\{ | \Pbit^{(L)}(\Gn) - \Pdee^{(L)} | > \delta\}
\le 2 \exp(-c_L\, n\, \delta^2)$
for all $n$ sufficiently large.
\end{theorem}

\rhead{Proof outline.}
A message after $L$ iterations is a deterministic function of the
depth-$2L+1$ directed computation neighborhood around its edge,
together with the hidden values, initial observations, verifier
observations, and channel-erasure variables in that neighborhood.
Bounded degree implies this neighborhood contains at most a constant
number of items depending on $L$ and $D_{\max}$ but not on $n$.
\Cref{lem:tree-like} gives total-variation convergence of the
neighborhood to the typed computation tree; on the tree, incoming
messages from distinct descendants are conditionally independent given
the target value and roles, and \Cref{lem:soundness} characterizes
each message as singleton-true or unresolved.  Evaluating the check
node on the tree gives \eqref{eq:pbar}--\eqref{eq:hupdate}; evaluating
the variable node gives \eqref{eq:pupdate}.

The convergence \eqref{eq:bit-conv} is obtained by triangle
decomposition,
\begin{equation*}
   \big|\Pbit^{(L)}(\Gn) - \Pdee^{(L)}\big|
   \le \big|\Pbit^{(L)}(\Gn) - \E\,\Pbit^{(L)}(\Gn)\big|
   + \big|\E\,\Pbit^{(L)}(\Gn) - \Pdee^{(L)}\big|.
\end{equation*}
The second term is $o_n(1)$ by the quantitative tree-convergence
of \Cref{lem:tree-tv} (and its variable-rooted form
\Cref{cor:tree-tv-node}): the expected per-variable bit-erasure
rate equals its tree-limit value up to a $C_L/n + \Delta_{n,L}$
correction, where the $C_L/n$ term comes from the breadth-first
collision bound and $\Delta_{n,L}$ measures the discrepancy between
the finite-$n$ empirical typed degree/socket law and the limiting
law (\eqref{eq:Delta-n}).  Under the standard
$O(\sqrt n)$-rounded socket-balance construction $\Delta_{n,L} =
O(\mathrm{poly}(N_L)/\sqrt n)$, so the empirical-law term dominates
the collision term in the typical regime; under a deterministic
socket-balance construction $\Delta_{n,L} = 0$ and the rate is the
sharper $C_L/n$.  For the first
term: changing one socket pairing, one hidden value, or one local
erasure variable can affect only messages whose depth-$2L+1$
computation neighborhoods intersect the changed item.  Bounded degree
caps the number of such messages by a constant $K_L$ depending on $L$
and $D_{\max}$ alone.  The empirical fraction is therefore a
bounded-differences function with Lipschitz constant $K_L / n$ per
independent random ingredient.  McDiarmid's inequality
\cite{mcdiarmid1989} gives the exponential tail
$\PR\{|\Pbit^{(L)}(\Gn) - \E\,\Pbit^{(L)}(\Gn)| > t\} \le
2\exp(-a_L n t^2)$ around the finite-$n$ mean; combining this with
the $b_{n, L}$ bound on the second (bias) term through the triangle
decomposition above yields \eqref{eq:bit-tail}, and the fixed-$\delta$
asymptotic form follows once $n$ is large enough that
$b_{n, L} \le \delta / 2$.  The full proof,
including the typed-Galton-Watson computation in
\Cref{lem:tree-like} and the bounded-difference constant tracking
through the value-conditioned recursion, is in \Cref{app:de-proof}.

The per-type ratio-statistic convergence \eqref{eq:phat-conv} is a
ratio of two empirical sums; both numerator and denominator are
bounded-differences functionals of the independent socket-pairing,
hidden-value, and erasure-variable ingredients, with Lipschitz
constants $K_L/n$ and $K_0/n$ respectively (the denominator depends
only on hidden values, $K_0 = 1$).  McDiarmid + the tree-convergence
expectation calculation give
$n^{-1} \sum_{(i,a) \in E_\tau} \one\{\Xstar_i = b\}
\xrightarrow{\PR} q_{\tau, b}$ and
$n^{-1} \sum_{(i,a) \in E_\tau} \one\{\Xstar_i = b\}
\one\{V^{(L)}_{i \to a} = \Uset\} \xrightarrow{\PR}
q_{\tau, b}\, \pmsg_{L, \tau}^{(b)}$.  Since $q_{\tau, b} > 0$ by
hypothesis, the continuous-mapping theorem applied to the ratio
function $(x, y) \mapsto x/y$ at $(q_{\tau, b}\,
\pmsg_{L, \tau}^{(b)},\, q_{\tau, b})$ gives \eqref{eq:phat-conv}.

\begin{remark}[Fixed-$L$ scope]
\label{rem:fixed-L-scope}
\Cref{thm:de} is a fixed-$L$ statement: the iteration count $L$ is
fixed before taking $n \to \infty$, and all constants $a_L, c_L, C_L,
K_L, \Delta_{n, L}$ in the bound depend on $L$.  Statements about
the $L \to \infty$ residual erasure rate or threshold behaviour on
finite-$n$ graphs do not follow directly from \Cref{thm:de}; they
require either a double-limit convention (taking $n \to \infty$
first for each $L$, then $L \to \infty$, as the residual-fixed-point
discussion below \Cref{thm:threshold} adopts) or an additional
uniform-in-$L$ concentration argument outside the scope of the
present finite-iteration result.
\end{remark}

\begin{remark}[Conservativeness of the concentration constants]
\label{rem:conservative-constants}
The constant $a_L$ in \eqref{eq:bit-tail} produced by the
bounded-differences method is highly conservative.  Tracking the
Lipschitz constants through the value-conditioned recursion gives
$a_L = 1 / (C\, D_{\max}^{4L+7})$ (\Cref{app:de-proof}): the
\emph{rate} it certifies is the stated $O(1/\sqrt n)$, but the sample
sizes for which the tail \eqref{eq:bit-tail} is numerically meaningful
grow rapidly in $L$ and $D_{\max}$.  This reflects the worst-case
per-ingredient bounded-differences accounting rather than the typical
fluctuations: the finite-graph Monte-Carlo overlays of
\Cref{sec:numerical} (\Cref{fig:and-concentration}) concentrate on the
DE prediction at far smaller $n$ than these constants would require.
We make no claim that $a_L$ is order-optimal.
\end{remark}

\begin{remark}[Channel-agnostic concentration step]
\label{rem:channel-agnostic-de}
The McDiarmid bounded-difference argument in the proof depends only
on the bounded message alphabet and bounded local degree, not on the
channel being erasure-only; it concentrates any fixed-round,
bounded-degree, finite-alphabet local message-passing scheme on a
sparse role-typed graph around its mean.  Deriving the DE recursion is
algorithm- and channel-specific, however: \Cref{thm:de} carries this
out only for the erasure-only logical-forcing decoder.  An
LLR-quantization DE under sum-product would inherit the concentration
step, but its recursion form and fixed-point setup are separate work.
\end{remark}

\begin{remark}[Why the recursion is value-conditioned]
\label{rem:value-cond}
For XOR checks the erasure process is value-symmetric: one may
analyze the all-zero codeword, and \eqref{eq:p-def} collapses to a
single recursion $\pmsg_{\ell, \tau}$.  For AND, OR, implication,
Horn, and other non-symmetric Boolean factors the symmetry is absent,
since a positive and a negative check carry different certificate
information, so the recursion must track $\pmsg_{\ell, \tau}^{(0)}$
and $\pmsg_{\ell, \tau}^{(1)}$ separately.  Averaging under the
role-prior \eqref{eq:prior} recovers a marginal recursion
$\pmsg_{\ell, \tau} = (1 - \Vbias_{\socketrole{\tau}})\, \pmsg_{\ell, \tau}^{(0)}
+ \Vbias_{\socketrole{\tau}}\, \pmsg_{\ell, \tau}^{(1)}$, but this is not
faithful to the dynamics; the value-conditioned recursion is the
primitive object.  Conditioning DE on the true value is familiar from
LDPC analysis over asymmetric memoryless channels
\cite{wang2005asymmetric}; what is specific here is that the asymmetry
comes from the verifier function $f_\theta$, not from the observation
channel, which stays a sound, value-independent erasure channel
(\Cref{sec:related}).
\end{remark}

\subsection{XOR specialization: recovering the LDPC-BEC baseline}
\label{sec:de-xor}

\emph{What collapses in the notation.}  Under XOR the value
superscript $(b)$ disappears entirely: the per-socket erasure
probability $p^{(b)}_{\ell, \tau}$ becomes a single quantity
$p_{\ell, \tau}$ independent of $b\in\{0,1\}$, and the forcing
probability $\varphi^{(b)}_{\ell, \theta, j}$ likewise loses its
value index.  The double-superscripted, template-indexed general
recursion of \Cref{thm:de} then collapses into a familiar
single-state recursion.  In the additional single-role case it
reduces to the textbook scalar recursion of LDPC-BEC density
evolution (Example~\ref{ex:regular-ldpc}).

When every check template is an XOR factor,
\begin{equation}
\label{eq:xor-fn}
   f_\theta(x_1, \ldots, x_d) = x_1 \oplus \cdots \oplus x_d,
\end{equation}
the observed check value determines the target if and only if every
other input is known.  The forcing probability is therefore
value-independent:
\begin{equation}
\label{eq:xor-forcing}
   \forceprob_{\ell, \theta, j}^{\mathrm{XOR}}
   = \prod_{k \neq j} \eta_{r_{\theta, k}, s_\theta}
   \big(1 - \pmsg_{\ell, (\theta, k)}\big),
\end{equation}
and the recursion collapses to a single (non-value-conditioned) state.

\begin{corollary}[XOR / parity-factor density evolution]
\label{cor:xor-de}
For XOR templates, \Cref{thm:de} reduces to the value-independent
recursion
\begin{align}
\label{eq:xor-h}
   \hmsg_{\ell, (\theta, j)}
   &= 1 - \big(1 - \epsC[s_\theta]\big)\,
   \eta_{s_\theta, r_{\theta, j}}
   \prod_{k \neq j} \eta_{r_{\theta, k}, s_\theta}
   \big(1 - \pmsg_{\ell, (\theta, k)}\big), \\
\label{eq:xor-p}
   \pmsg_{\ell+1, \tau}
   &= \epsV[\socketrole{\tau}]\,
   \E\!\left[ \prod_{\tau' \in \Sockets_{\socketrole{\tau}}}
   \hmsg_{\ell, \tau'}^{D_{\tau'}^{(\socketrole{\tau}), \tau, \mathrm{ex}}}
   \right].
\end{align}
In the single-role, noiseless-verifier ($\epsC = 0$),
noiseless-channel ($\eta \equiv 1$) case, this is exactly the standard
LDPC-BEC density-evolution recursion of \cite{luby1998analysis,
richardson2008modern}.
\end{corollary}

\begin{example}[Regular LDPC reduction]
\label{ex:regular-ldpc}
In a single-role $(d_v, d_c)$-regular ensemble with $\epsC = 0$ and
$\eta \equiv 1$, \eqref{eq:xor-h}--\eqref{eq:xor-p} become
\[
   \hmsg_\ell = 1 - (1 - \pmsg_\ell)^{d_c - 1},
   \qquad
   \pmsg_{\ell+1} = \epsV\, \hmsg_\ell^{d_v - 1},
\]
which is the textbook BEC edge-erasure recursion
\cite{richardson2008modern}.
\end{example}

\subsection{AND specialization: certificates and verifier asymmetry}
\label{sec:de-and}

\emph{What does \emph{not} collapse in the notation.}  In contrast to
the XOR case, under AND the value superscript $(b)$ stays.  The
recursion tracks two separate per-socket erasure probabilities,
$p^{(1)}_{\ell, \tau}$ for sockets at variables holding the true
value $1$ and $p^{(0)}_{\ell, \tau}$ for those at variables holding
$0$, and these obey two different update rules.  Behind the
heavier notation is the statement that a passing AND test and a
failing AND test carry quantitatively different information: that
content is exactly what \eqref{eq:and-positive}--\eqref{eq:and-negative}
encode below.

When every check template is an AND factor,
\begin{equation}
\label{eq:and-fn}
   f_\theta(x_1, \ldots, x_d) = x_1 \wedge \cdots \wedge x_d,
\end{equation}
the recursion stays value-conditioned and exposes the
positive--negative certificate asymmetry of realistic verifiers.

\begin{proposition}[AND density evolution]
\label{prop:and-de}
For AND templates with target socket $j$, the forcing probabilities in
\eqref{eq:forcing-prob} are
\begin{align}
\label{eq:and-positive}
   \forceprob_{\ell, \theta, j}^{(1)}
   &= \prod_{k \neq j} \Vbias_{r_{\theta, k}}, \\
\label{eq:and-negative}
   \forceprob_{\ell, \theta, j}^{(0)}
   &= \prod_{k \neq j} \Vbias_{r_{\theta, k}}\,
   \eta_{r_{\theta, k}, s_\theta}\,
   \big(1 - \pmsg_{\ell, (\theta, k)}^{(1)}\big).
\end{align}
The received check-to-variable erasure $\hmsg_{\ell, (\theta, j)}^{(b)}$
is then given by \eqref{eq:hupdate} with \eqref{eq:and-positive} for
$b = 1$ and \eqref{eq:and-negative} for $b = 0$, and the variable
update is \eqref{eq:pupdate}.
\end{proposition}

\begin{proof}
Condition on $X_j = 1$.  For an AND template, the verifier output is
$1$ iff every other input is also $1$, and only then.  When
the output is $1$, no assignment with $x_j = 0$ can satisfy the check
because $\one\{0 \cdot x_2 \cdots x_d = 1\} = 0$.  The target feasible
set $\force_{\theta, j}(1; M_{-j})$ is therefore $\{1\}$ regardless of
whether the other input messages were singletons or unresolved.
The probability that $X_j = 1$ and every other true value is $1$ is
$\Vbias_{r_{\theta, j}}\, \prod_{k \neq j} \Vbias_{r_{\theta, k}}$;
conditioning on $X_j = 1$ removes the leading factor and gives
\eqref{eq:and-positive}.

Now condition on $X_j = 0$.  The AND output is then $0$.  The target
is forced to $0$ iff the assignment $x_j = 1$ is infeasible given
$M_{-j}$; that is, every other input message is the singleton $\{1\}$
(otherwise some other variable could be $0$, making the AND output
$0$ regardless of $x_j$).  For each $k \neq j$, the probability that
the true value $X_k = 1$ and the check receives the singleton
$\{1\}$ is $\Vbias_{r_{\theta, k}}\, \eta_{r_{\theta, k}, s_\theta}\,
(1 - \pmsg_{\ell, (\theta, k)}^{(1)})$.  Multiplying over $k \neq j$
gives \eqref{eq:and-negative}.  Verifier erasure and return-channel
erasure then enter through \eqref{eq:hupdate}.
\end{proof}

\begin{remark}[Practical reading of the AND rule]
\label{rem:and-practical}
A passing local test, in this framework, is a \emph{strong} certificate:
all pieces required by the test must be valid, and the test certifies
all of them at once.  A failing local test is a \emph{weaker}
certificate: it identifies that some piece is invalid, but to pin
which one, the remaining pieces must already be known valid.  This
asymmetry is absent from XOR and is one of the structural reasons
Boolean verifier nodes are a better abstraction for agent systems
than parity-only LDPC analysis.  In the running Hilbert example
(\Cref{tab:hilbert-mapping}), this is exactly what a Lean kernel
does: a successful kernel call certifies the joint validity of every
proposed lemma in the call's scope, while a kernel failure says only
that something in the scope was off.
\end{remark}

\section{Threshold, Stability, and Non-Interchangeability}\label{sec:threshold}

Let $\bm p_\ell$ denote the vector $(\pmsg_{\ell, \tau}^{(b)})_{\tau, b}
\in [0, 1]^{2|\Sockets|}$ collecting the value-conditioned per-socket
erasure probabilities at iteration $\ell$.  Equations
\eqref{eq:pbar}--\eqref{eq:pupdate} define a continuous map
\begin{equation}
\label{eq:DE-map}
   \bm p_{\ell+1} = \DEmap_\lambda(\bm p_\ell),
\end{equation}
where $\lambda$ collects all parameters, role proportions, value
priors, degree laws, template proportions, erasure probabilities, and
channel fidelities.

\subsection{Residual fixed point and target-reliability threshold}

\rhead{What this theorem says.}  The DE iteration is monotone
from its natural initialization $\bm p_0$, so it converges to a
trajectory-selected residual fixed point
$\bm p_\infty(\lambda) \defeq \lim_{\ell \to \infty}
\DEmap_\lambda^\ell(\bm p_0)$ (other fixed points may exist but are
unreachable from $\bm p_0$; no uniqueness claim).  Along any
one-parameter worsening path, the residual $\Pdee^{(\infty)}(\lambda)$
is monotone in the path parameter and crosses any target reliability
$\delta$ at a well-defined threshold $a^*_\delta$.  The convergence
and threshold arguments use the standard monotone-density-evolution
technique; what is specific to this model is the state space on which
it runs, the value-conditioned socket-typed erasure vector of
\Cref{thm:de}, the operational reliability paths along which the
threshold is read (variable abstention, verifier erasure,
reasoning-channel loss), and the two path-monotonicity conditions of
part~(b), which hold automatically only for single-coordinate
worsening paths.

\begin{theorem}[Residual fixed point and target threshold]
\label{thm:threshold}
Consider the density-evolution iteration \eqref{eq:DE-map} with initial
condition \eqref{eq:pinit}.
\begin{enumerate}[label=(\alph*),leftmargin=2em]
\item \emph{Monotone convergence.}  $\DEmap_\lambda$ is continuous and
coordinatewise monotone increasing in its argument.  The sequence
$\{\bm p_\ell\}_{\ell \ge 0}$ is coordinatewise nonincreasing and
converges to a fixed point $\bm p_\infty(\lambda)$.  The residual
bit-erasure rate
\[
   \Pdee^{(\infty)}(\lambda) \defeq \lim_{L \to \infty} \Pdee^{(L)}(\lambda)
\]
is well-defined.
\item \emph{Target-reliability threshold.}  Let $a \in I = [a_{\min},
a_{\max}]$ and let $a \mapsto \lambda(a)$ be a one-parameter
worsening path along which (i) the initial condition
$\bm p_0(\lambda(a))$ is coordinatewise nondecreasing in $a$, and
(ii) $\DEmap_{\lambda(a)}(\bm p)$ is coordinatewise nondecreasing
in $a$ for every $\bm p$ (e.g.\ $a$ is a single $\epsV[r]$,
$\epsC[s]$, or $1 - \eta_{r,s}$ entry, in which case both
conditions hold automatically).  Then $\Pdee^{(\infty)}(\lambda(a))$
is nondecreasing in $a$.  For each target residual level
$\delta \in [0, 1]$, set $S_\delta \defeq \{a \in I :
\Pdee^{(\infty)}(\lambda(a)) \le \delta\}$ and
\[
   a_\delta^*(\lambda) \defeq
   \begin{cases}
     \sup S_\delta, & S_\delta \neq \varnothing, \\
     a_{\min}^{-}, & S_\delta = \varnothing
       \text{ (target unreachable on the path)}.
   \end{cases}
\]
\end{enumerate}
\end{theorem}

\begin{remark}[Why the zero-erasure state is treated separately]
\label{rem:zero-fp-caveat}
The reader familiar with classical LDPC-BEC density evolution
\cite{richardson2008modern} may expect a third clause stating that
the zero-erasure state $\bm p =
\zero$ is a fixed point with local stability governed by the spectral
radius $\rho(\DEjac_\lambda(\zero))$.  In the present three-tier
setting, $\bm p = \zero$ is in general not a fixed point: a
check-to-variable message can remain $\Uset$ even when all incoming
variable messages are singletons, because the verifier output may be
erased ($\epsC[s] > 0$) or the return channel may have erased
($\eta_{s, r} < 1$).  Spectral-radius stability of the zero-erasure
state therefore needs an explicit condition.  We separate this fact
into \Cref{cor:zero-stability} below.
\end{remark}

\begin{definition}[Coordinate-forcing template]
\label{def:coord-forcing}
A Boolean verifier template $\theta$ is \emph{coordinate-forcing}
if, for every target socket $j$, every assignment
$(x_k)_{k \neq j} \in \{0,1\}^{d_\theta - 1}$, and the consistent
verifier output $z = f_\theta(x_1, \ldots, x_{d_\theta})$, the
target value $x_j$ is uniquely determined by $z$ together with
$(x_k)_{k \neq j}$.
\end{definition}

XOR/parity is coordinate-forcing: knowing the parity of all inputs
and any $d - 1$ of them determines the remaining input.  AND, OR,
implication, and Horn factors are not coordinate-forcing in
general: e.g., AND with $z = 0$ and one other input already $0$
admits both $x_j = 0$ and $x_j = 1$ as feasible.
In fact the property pins the admissible templates down exactly.
Writing $x_{-j} = (x_k)_{k \neq j}$, forcing at socket $j$ for every
$x_{-j}$ means the map $x_j \mapsto f_\theta(x)$ is a bijection of
$\{0, 1\}$ at each fixed $x_{-j}$, which holds iff $f_\theta$ is affine
in that coordinate, $f_\theta(x) = x_j \oplus g_j(x_{-j})$.  Imposing
this at every socket forces $f_\theta(x) = c \oplus \bigoplus_k x_k$,
a parity check up to a constant.  The coordinate-forcing templates are
therefore exactly the affine (parity-type) checks over
$\mathrm{GF}(2)$; this is the structural reason classical LDPC codes
use parity checks exclusively, since they are the only Boolean
constraints that force every coordinate.
Coordinate-forcing is a sufficient structural condition guaranteeing
that $\bm p = \zero$ is a fixed point of $\DEmap_\lambda$ at
noiseless verifier and channel.  Under full-support value priors
$\Vbias_r \in (0, 1)$ and nondegenerate socket use, failure of
coordinate-forcing produces a positive residual offset at
$\bm p = \zero$ even with noiseless verifier and channel
(\Cref{rem:coord-forcing-needed}); degenerate cases such as
$\Vbias_r = 1$ for AND can still attain the zero state, but they sit
outside the regime our threshold analysis targets.

\begin{corollary}[Local stability of the zero-erasure fixed point,
coordinate-forcing case]
\label{cor:zero-stability}
Suppose every check template is coordinate-forcing
(\Cref{def:coord-forcing}), $\epsC[s] = 0$ for all
$s \in \RolesC$, and $\eta_{r, s} = 1$ for all role pairs.  Then
$\bm p = \zero$ is a fixed point of $\DEmap_\lambda$.  For finite
template/role/degree laws, $\DEmap_\lambda$ is polynomial, hence
$C^1$ on $[0, 1]^{2|\Sockets|}$.  Let $\bm p^*$ be any fixed point of $\DEmap_\lambda$ at which
$\DEmap_\lambda$ is $C^1$ on a neighborhood of $\bm p^*$ and
$\rho\big(\DEjac_\lambda(\bm p^*)\big) < 1$.  Then the iteration is
locally exponentially attracted to $\bm p^*$: for every
$c \in (\rho(\DEjac_\lambda(\bm p^*)), 1)$ there exists a
neighborhood $U$ of $\bm p^*$ such that
\[
   \| \bm p_\ell - \bm p^* \|
   \;\le\; K c^\ell\, \| \bm p_0 - \bm p^* \|
   \qquad \text{for all } \bm p_0 \in U,
\]
for some constant $K = K(c, \lambda)$ in any fixed norm; continuity
of $\DEmap_\lambda, \DEmap_\lambda^2, \ldots, \DEmap_\lambda^{k_0-1}$
controls the finitely many intermediate iterates.  In particular,
under the coordinate-forcing, noiseless-verifier, noiseless-channel
hypotheses of the corollary, the zero-erasure fixed point
$\bm p^* = \zero$ has $\| \bm p_\ell \| \le K c^\ell \| \bm p_0 \|$
whenever $\rho(\DEjac_\lambda(\zero)) < 1$.
\end{corollary}

\begin{remark}[Why coordinate-forcing is needed]
\label{rem:coord-forcing-needed}
Without coordinate-forcing, $\bm p = \zero$ is in general
not a fixed point even at noiseless verifier and channel.
For AND, \Cref{prop:and-de} gives forcing probabilities
$\forceprob^{(1)} = \prod_{k \neq j} \Vbias_{r_{\theta, k}}$ and
$\forceprob^{(0)} = \prod_{k \neq j} \Vbias_{r_{\theta, k}}\,
\eta_{r_{\theta, k}, s_\theta}\,(1 - \pmsg^{(1)}_{\ell, (\theta, k)})$
which equal $1$ at $\bm p = \zero$ only when every $\Vbias = 1$;
generically $\hmsg^{(b)}\big|_{\bm p = \zero} > 0$, so $\bm p = \zero$
is not a fixed point.  The corollary therefore restricts to
coordinate-forcing primitives such as XOR.
\end{remark}

\begin{remark}[XOR/BEC specialization]
\label{rem:xor-bec-stability}
\Cref{cor:zero-stability} applies in particular to the XOR/BEC
specialization (\Cref{cor:xor-de}) with noiseless verifier and
noiseless channel.  In that regime the spectral-radius condition
$\rho(\DEjac_\lambda(\zero)) < 1$ recovers the classical LDPC-BEC
\emph{local stability} condition at $\bm p = \zero$ (see, e.g.,
\cite{richardson2008modern}).  We do not claim this recovers the
full LDPC-BEC threshold: for $d_v \geq 3$ regular ensembles the
threshold is determined by a global tangency of $\DEmap_\lambda$,
not by the linearization at $\bm p = \zero$, and for $d_v = 2$
the threshold coincides with the local-stability condition.
\end{remark}

\begin{proof}
\emph{Monotonicity (a).}  $\DEmap_\lambda$ is built from finitely many
sums and products of continuous functions, hence continuous.  For
monotonicity, observe that increasing any incoming variable-to-check
erasure probability $\pmsg_{\ell, \tau}^{(b)}$ replaces some singleton
inbound messages by $\Uset$.  By
\Cref{lem:gamma-monotone} (\Cref{app:monotonicity}), this enlarges the
feasible set $\force_{\theta, j}$, which can only convert a singleton
output to $\Uset$, never the reverse.  Hence outgoing
check-to-variable erasure probabilities cannot decrease.  The variable
update \eqref{eq:pupdate} is also monotone in its arguments.
Composing, $\DEmap_\lambda$ is coordinatewise monotone increasing.

The initial vector is $\bm p_0 = (\epsV[\socketrole{\tau}])_{\tau, b}$.
Since the variable update is multiplied by $\epsV[\socketrole{\tau}]$
in \eqref{eq:pupdate} and the inner expectation is at most $1$, we
have $\pmsg_{1, \tau}^{(b)} \le \epsV[\socketrole{\tau}] = \pmsg_{0,
\tau}^{(b)}$.  By monotonicity, $\bm p_{\ell+1} \le \bm p_\ell$
coordinatewise for all $\ell$.  The sequence is bounded below by zero,
so it converges coordinatewise.  Continuity gives that the limit is a
fixed point.  Monotone convergence in \eqref{eq:terminal-de} gives the
residual bit-erasure prediction.

\emph{Target threshold (b).}  Along a monotone reliability path, the
limiting residual is nondecreasing in $a$ by repeated application of
clause (a).  Hence the sublevel set $\{a : \Pdee^{(\infty)}(\lambda(a))
\le \delta\}$ is an interval, and the displayed supremum is
well-defined.
\end{proof}

\begin{proof}[Proof of \Cref{cor:zero-stability}]
Under $\epsC[s] = 0$, $\eta_{r,s} = 1$ everywhere, and every check
template coordinate-forcing (\Cref{def:coord-forcing}), the
check-to-variable update \eqref{eq:check-update} produces a singleton
output whenever the forcing operator does, which at $\bm p_\ell =
\zero$ occurs almost surely: every other incoming variable message
is the (true-value) singleton, and the coordinate-forcing condition
asserts that the target value is then uniquely determined by the
verifier output and the other inputs.  Hence
$\hmsg_{\ell, \tau}^{(b)} = 0$ for all $\tau, b$ at $\bm p_\ell = \zero$,
so $\bm p_{\ell+1} = \DEmap_\lambda(\zero) = \zero$ and $\zero$ is a
fixed point.

For local stability, take any fixed point $\bm p^*$ at which
$\DEmap_\lambda$ is differentiable with $\rho(\DEjac_\lambda(\bm p^*))
< 1$ (in particular $\bm p^* = \zero$ in the noiseless regime).  This
is the standard spectral-radius stability criterion: a fixed point of
a $C^1$ map whose Jacobian has spectral radius below $1$ is locally
geometrically attracting, because some iterate $\DEmap_\lambda^{k_0}$
is then a contraction on a neighborhood of $\bm p^*$.  Density
evolution started anywhere near $\bm p^*$ therefore converges to it at
a geometric rate.
\end{proof}

\subsection{The Non-Interchangeability Proposition}

The statement of \Cref{thm:threshold}~(b) hides a fact that turns out
to be one of the main structural distinctions of the framework: the
three erasure tiers $\epsV[r], \epsC[s], \eta_{r, s}$ enter the DE
Jacobian in functionally different positions and cannot be collapsed
into a single effective scalar.

\begin{proposition}[No scalar effective-noise reduction of the
DE map]
\label{prop:noninterchange}
There is no smooth change of variables $\widetilde\epsilon =
\Psi(\epsV, \epsC, \eta_{V\to C}, \eta_{C\to V})$ such that the
density-evolution map $\DEmap_\lambda$ depends on the four tier
parameters only through $\widetilde\epsilon$, even on a single-role
single-template AND ensemble.  Quantitatively, the
\emph{parameter Jacobian}
\[
   D_{\mathrm{par}}\big[\DEmap_\lambda(\zero)\big]
   \defeq \Big[\partial_{\epsV}\DEmap_\lambda,\,
   \partial_{\epsC}\DEmap_\lambda,\,
   \partial_{\eta_{V\to C}}\DEmap_\lambda,\,
   \partial_{\eta_{C\to V}}\DEmap_\lambda\Big]
   \bigg|_{\bm p = \zero}
\]
(distinct from the state Jacobian $\DEjac_\lambda(\bm p)$ used in
stability analysis) has rank at least two on a generic open subset
of parameter space.  The stronger rank-$\ge 3$ separation of the
three operational tiers ($\epsC$ versus $\eta_{C \to V}$), under
nondegenerate role structure, is established in
\Cref{prop:noninterchange-roles}.
\end{proposition}

\begin{proof}
We exhibit the rank-$\ge 2$ statement on a single-role
single-template AND ensemble with separated channel directions; the
general case follows by including this submodel as a slice of the
parameter space.

\emph{Setup.}  Take $|\RolesV| = |\RolesC| = 1$, one AND template of
arity $d \ge 2$ and excess variable degree $m = d_v - 1 \ge 1$, value
prior $\Vbias \in (0, 1)$, variable-side erasure $\epsV \in [0, 1)$,
verifier-side erasure $\epsC \in [0, 1)$, and the two directional
channel fidelities $\eta_{V \to C}, \eta_{C \to V} \in (0, 1]$ kept
separate (this is the structural distinction that the proof
needs to expose; collapsing $\eta_{V \to C} = \eta_{C \to V}$ from
the start would obscure exactly the rank we want to read off).

\emph{The DE map at the zero-erasure state.}  Combining
\Cref{prop:and-de} with \eqref{eq:hupdate} and \eqref{eq:pupdate} and
evaluating at $\bm p = \zero$ (so $\pmsg^{(1)} = \pmsg^{(0)} = 0$),
\begin{align}
\label{eq:noninterchange-h1}
   \hmsg^{(1)}\big|_{\bm p = \zero}
   &= 1 - (1 - \epsC)\, \eta_{C \to V}\, \Vbias^{d - 1}, \\
\label{eq:noninterchange-h0}
   \hmsg^{(0)}\big|_{\bm p = \zero}
   &= 1 - (1 - \epsC)\, \eta_{C \to V}\, \Vbias^{d - 1}\,
   \eta_{V \to C}^{\,d - 1},
\end{align}
and the value-conditioned variable updates at $\bm p = \zero$ are
\begin{equation}
\label{eq:noninterchange-Phi}
   \DEmap_\lambda^{(b)}(\zero)
   = \epsV \cdot \big(\hmsg^{(b)}\big|_{\bm p = \zero}\big)^m,
   \qquad b \in \{0, 1\}.
\end{equation}

\emph{Two linearly independent parametric directions.}  We compute
the parametric derivatives of \eqref{eq:noninterchange-Phi} at the
zero-erasure state with respect to two of the four parameters
$(\epsV, \epsC, \eta_{V \to C}, \eta_{C \to V})$ and verify that
the resulting two columns are linearly independent.

\emph{Column 1 ($\partial / \partial \epsV$).}  By
\eqref{eq:noninterchange-Phi},
\begin{equation}
\label{eq:noninterchange-col1}
   \partial_{\epsV}\, \DEmap_\lambda^{(b)}(\zero)
   = \big(\hmsg^{(b)}\big|_{\bm p = \zero}\big)^m,
   \qquad b \in \{0, 1\}.
\end{equation}
Both entries are strictly positive on the open set
$\{\epsC \in [0, 1), \,\Vbias \in (0, 1), \,\eta_{V \to C},
\eta_{C \to V} \in (0, 1]\}$ because each $\hmsg^{(b)}|_{\bm p =
\zero} \in (0, 1)$ on that set.

\emph{Column 2 ($\partial / \partial \eta_{V \to C}$).}  Only the
$b = 0$ branch depends on $\eta_{V \to C}$:
\begin{equation}
\label{eq:noninterchange-col2}
   \partial_{\eta_{V \to C}}\, \DEmap_\lambda^{(1)}(\zero) = 0,
   \qquad
   \partial_{\eta_{V \to C}}\, \DEmap_\lambda^{(0)}(\zero)
   = -\, m\, \epsV\, \big(\hmsg^{(0)}\big|_{\bm p = \zero}\big)^{m - 1}\,
   (1 - \epsC)\, \eta_{C \to V}\, \Vbias^{d - 1}\,
   (d - 1)\, \eta_{V \to C}^{\,d - 2}.
\end{equation}
The $b = 0$ entry is strictly nonzero on the open set
$\{\epsV > 0, \,\epsC < 1, \,\Vbias > 0, \,\eta_{V \to C}, \eta_{C \to V}
> 0\}$, while the $b = 1$ entry is identically zero.

\emph{Linear independence.}  Suppose, on a neighborhood of any
parameter point in the above open set, there are scalars
$\alpha, \beta$ (not both zero) with
$\alpha \cdot \eqref{eq:noninterchange-col1} + \beta \cdot
\eqref{eq:noninterchange-col2} = 0$
componentwise.  The $b = 1$ component gives $\alpha (\hmsg^{(1)})^m =
0$, hence $\alpha = 0$ since $\hmsg^{(1)} > 0$.  The $b = 0$ component
then gives $\beta \cdot \partial_{\eta_{V \to C}}
\DEmap_\lambda^{(0)}(\zero) = 0$, hence $\beta = 0$ by
\eqref{eq:noninterchange-col2}.  The two columns are therefore
linearly independent on the open set, and the parameter Jacobian
$D_{\mathrm{par}}\big[\DEmap_\lambda(\zero)\big]$ has rank at least
$2$ on that open subregion.

\emph{Smooth-change-of-variables corollary.}  On any neighborhood
$U$ contained in the above open subregion, suppose
$\DEmap_\lambda(\bm p)$ depended on the four tier parameters only
through some smooth scalar $\widetilde\epsilon = \Psi(\epsV, \epsC,
\eta_{V \to C}, \eta_{C \to V})$.  Then the parameter Jacobian
on $U$ would factor as
$\partial \DEmap / \partial \widetilde\epsilon \cdot \nabla \Psi$,
which has rank at most $1$.  This contradicts the rank-$\ge 2$
statement above; no such $\Psi$ exists on $U$.
\end{proof}

\begin{remark}[What the rank-$\ge 2$ statement does and does not give]
\label{rem:rank-scope}
\Cref{prop:noninterchange} establishes that no smooth scalar
function of $(\epsV, \epsC, \eta_{V\to C}, \eta_{C\to V})$ summarizes
the DE map: the design problem is multi-knob, not single-knob.
This is strictly weaker than three-way mutual independence of
the tiers.  In the single-role single-template slice used in the
proof, $\epsC$ and $\eta_{C\to V}$ enter the
check-to-variable update through the product
$(1 - \epsC)\eta_{C\to V}$ in
\eqref{eq:noninterchange-h1}--\eqref{eq:noninterchange-h0}, so
locally these two knobs are confounded at the level of
$\DEmap_\lambda$; rank-$\ge 2$ is achieved by $\epsV$ versus
$\eta_{V\to C}$, not by all four parameters separately.  A stronger
rank-$\ge 3$ statement that distinguishes verifier-side erasure
from the return channel requires nondegenerate role structure;
\Cref{prop:noninterchange-roles} establishes it with two variable
roles, where $\epsC[s]$ acts as a gate shared across roles while the
return-channel fidelities act on different output blocks.
\end{remark}

\begin{proposition}[Rank-three tier separation under heterogeneous
roles]
\label{prop:noninterchange-roles}
Let one check role $s$ serve two variable roles $r_1, r_2$ through a
single arity-two AND template.  Then on a nonempty open subset of
parameter space the parameter Jacobian
$D_{\mathrm{par}}\big[\DEmap_\lambda(\zero)\big]$ has rank at least
three, with three independent directions given by the variable-side
erasure $\epsV[r_1]$, the verifier-side erasure $\epsC[s]$, and a
return-channel fidelity $\eta_{s, r_1}$.  In particular $\epsC[s]$
and $\eta_{s, r_1}$ are no longer confounded as in the single-role
slice of \Cref{prop:noninterchange}.
\end{proposition}

\begin{proof}
Use sockets $j_1, j_2$ of roles $r_1, r_2$; role-$r$ variables have
excess degree $m_r \ge 1$ on their socket type
$\tau_r = (\theta, j_r)$, with priors $\Vbias_r \in (0, 1)$ and
parameters $\epsV[r] \in (0, 1)$, $\epsC[s] \in [0, 1)$, forward
fidelities $\eta_{r, s} \in (0, 1)$, return fidelities
$\eta_{s, r} \in (0, 1]$.  As in the proof of
\Cref{prop:noninterchange}, evaluating \eqref{eq:hupdate} with the AND
forcing probabilities \eqref{eq:and-positive}--\eqref{eq:and-negative}
at $\bm p = \zero$ gives, for a role-$r$ target,
$\hmsg^{(b)}_{\tau_r}\big|_{\zero} = 1 - g_r^{(b)}$ with
$g_r^{(b)} = (1 - \epsC[s])\,\eta_{s, r}\,A_r^{(b)}$, and
$\DEmap^{(b)}_{\tau_r}(\zero) = \epsV[r]\,(1 - g_r^{(b)})^{m_r}$ by
\eqref{eq:pupdate}; here $A_{r_1}^{(1)} = \Vbias_{r_2}$ and
$A_{r_1}^{(0)} = \Vbias_{r_2}\,\eta_{r_2, s}$.

Order the four output coordinates as
$(\tau_1^{(1)}, \tau_1^{(0)}, \tau_2^{(1)}, \tau_2^{(0)})$ and take the
columns $\partial_{\epsV[r_1]}$, $\partial_{\eta_{s, r_1}}$,
$\partial_{\epsC[s]}$.  The first two vanish on the $\tau_2$ block,
whereas
$\partial_{\epsC[s]}\DEmap^{(b)}_{\tau_2}(\zero)
= \epsV[r_2]\, m_2 (1 - g_2^{(b)})^{m_2 - 1}\,\eta_{s, r_2}\,
A_{r_2}^{(b)} > 0$ there, because the gate $(1 - \epsC[s])$ is shared
across both roles; hence any vanishing combination of the three
columns has zero coefficient on $\partial_{\epsC[s]}$.  Restricted to
the $\tau_1$ block, the remaining columns $\partial_{\epsV[r_1]}$ and
$\partial_{\eta_{s, r_1}}$ have $2 \times 2$ determinant
\[
   -\,\epsV[r_1]\, m_1 (1 - \epsC[s])\,
   (1 - g_1^{(1)})^{m_1 - 1}(1 - g_1^{(0)})^{m_1 - 1}\,
   \big(A_{r_1}^{(0)} - A_{r_1}^{(1)}\big),
\]
the $g_1$ cross-terms cancelling exactly as in
\Cref{prop:noninterchange}.  Since
$A_{r_1}^{(0)} - A_{r_1}^{(1)} = \Vbias_{r_2}(\eta_{r_2, s} - 1)
\ne 0$ whenever the forward fidelity $\eta_{r_2, s} < 1$, this
determinant is nonzero on a nonempty open set, the three columns are
linearly independent, and the rank is at least three.  The gate
$(1 - \epsC[s])$ acting across both role blocks is exactly what
breaks the single-role confound $(1 - \epsC)\eta_{C \to V}$ of
\Cref{rem:rank-scope}.
\end{proof}

\begin{remark}[Operational reading]
\label{rem:operational}
\Cref{prop:noninterchange} says that improving a verifier
($\downarrow \epsC$), adding more proposer redundancy
($\downarrow \epsV$), and improving an inter-agent communication
channel ($\uparrow \eta$) are three separate design knobs whose
marginal values cannot be folded into a single ``effective noise''
parameter.  The shadow-price KKT corollary in
\Cref{sec:optimization} (\Cref{thm:optimization}~(e)) makes this
operational: at any architecture optimum, the three tier dual
variables are typically non-degenerate and can be read directly off
the adjoint solution.  Equivalently, in contrapositive form: adding redundant LLM
proposers cannot fix a broken Lean verifier, and neither can
fix a formatting mismatch between them.  The three tiers
attack disjoint failure modes, and any single-knob policy that
treats them as substitutes leaves Pareto-optimal designs on the
table.
\end{remark}

\begin{remark}[Differentiator from MET-LDPC and noisy-MP-decoder]
\label{rem:metldpc}
Multi-edge-type LDPC \cite{richardson2008modern} supports multiple
edge-type kernels under a homogeneous parity update, and the
noisy-message-passing-decoder line
\cite{tarighati2015noisy,dupraz2021noisy} distinguishes channel
noise from decoder-side message noise.  What is new here is the
agent-operational interpretation and joint presence of three
tiers, variable-side abstention, verifier-output erasure, and
directed role-pair artifact erasure, together with value-conditioned
Boolean logical forcing on a sparse role-typed factor graph; neither
line carries a non-interchangeability statement of the kind
formalized by the rank-$\ge 2$ argument above.  This is the
structural distinction recorded in \Cref{tab:related-comparison}.
\end{remark}

\section{Certificate-Stopping Sets}\label{sec:stopping}

Density evolution predicts typical asymptotic behavior.  At finite
length, the decoder can fail because the realized graph and noise
pattern contain a local obstruction.  For XOR on the BEC, the
obstruction is a stopping set \cite{di2002finite}.  For general
Boolean verifier functions, the right obstruction is a
\emph{certificate-stopping set}.

In this section we condition on the realized 5-tuple of all the
randomness, which we call the \emph{transcript}:
\begin{equation}
\label{eq:transcript}
   \mathcal{T} \defeq \big(\Gn,\, \Xstar,\, \{\Aobs_i\}_{i \in \Vset_n},\,
   \{\Zobs_a\}_{a \in \Cset_n},\, \{B_{u \to w}\}\big),
\end{equation}
where $B_{u \to w} \in \{0, 1\}$ is the persistent directed-channel
availability indicator on each ordered edge.  ``Persistent'' means
each directed edge is either available or erased throughout the
deterministic peeling transcript, the clean finite-transcript
analogue of the stochastic density-evolution model.  Once
$\mathcal{T}$ is fixed, no randomness remains (graph, hidden vector,
observations, verifier outputs, and channel gates are all realized),
so the peeling decoder analyzed below is deterministic; this section
has thus moved from the probabilistic density evolution of
\Cref{sec:de} to a conditioned finite-instance analysis.

\rhead{Set-level peeling and the message-passing decoder
coincide.}  Let $S_t \defeq \{i \in \Vset_n : \widehat M_i^{(t)} =
\Uset\}$ be the unresolved variable set after $t$ rounds of the
set-valued message-passing decoder (messages carry candidate sets in
$\Mset = \{\{0\}, \{1\}, \Uset\}$, with $\Uset = \{0, 1\}$ marking an
unresolved variable) of \Cref{sec:model-forcing}.
Under the soundness invariant (\Cref{lem:soundness}), every
non-erased private observation is the true singleton and every
non-erased verifier output equals $T_a$, so the variable update
\eqref{eq:var-update} produces a singleton at $i$ at round $t+1$
exactly when (i) $\Aobs_i$ is non-erased, or (ii) some adjacent
check certifies $i$ relative to $S_t$ in the sense made precise
below.  Thus $S_{t+1}$ is obtained from $S_t$ by removing exactly
those variables.  The set-level peeling process below is the
unresolved-set evolution of the actual decoder, not a separate
abstraction.

\subsection{Definition and main theorem}

\rhead{Intuition.}  A certificate-stopping set is a residual group
of unresolved variables such that every adjacent verifier is unable
to certify any one of them from the information outside the group.
The obstruction is not merely graph-theoretic: it depends on the
Boolean semantics of the verifier templates and on the realized
verifier outputs, so a set $S$ may be certificate-stopping under
one verifier transcript and not under another on the same graph.
This generalizes the classical LDPC-BEC stopping-set notion, which
is recovered as the value-symmetric, exact-verifier specialization
(\Cref{cor:xor-stopping}).

For a candidate unresolved set $S \subseteq \Vset_n$, define the
message available from a variable $k$ to a neighboring check $a$
\emph{relative to $S$} by
\begin{equation}
\label{eq:relative-message}
   M_{k \to a}(S)
   = \begin{cases}
   \{\Xstar_k\}, & k \notin S \text{ and the directed channel } k \to a
   \text{ is available}, \\
   \Uset, & \text{otherwise}.
   \end{cases}
\end{equation}
A check $a$ of template $\theta$ \emph{certifies} a target variable
$i \in S$ at socket $j$ relative to $S$ if
\begin{enumerate}[label=(\roman*),leftmargin=2.5em]
\item $\Zobs_a \in \{0, 1\}$ (the verifier output is not erased),
\item the directed channel $a \to i$ is available, and
\item $\force_{\theta, j}\big(\Zobs_a;\, M_{\partial a \setminus i}(S)\big)
   = \{\Xstar_i\}$ (the forcing operator outputs the singleton at the
   true value).
\end{enumerate}

\begin{definition}[Certificate-stopping set for a realized transcript]
\label{def:certstop}
For a fixed transcript $\mathcal{T}$ as in \eqref{eq:transcript}, a
nonempty set $S \subseteq \Vset_n$ is a \emph{certificate-stopping
set for $\mathcal{T}$} if every $i \in S$ has erased private
observation $\Aobs_i = \Uset$ and no adjacent check certifies $i$
relative to $S$.
\end{definition}

\begin{remark}[Three levels of stopping-set objects]
\label{rem:stopping-levels}
The certificate-stopping property is a property of a realized
transcript, not of the graph alone: stoppedness depends on
$\Xstar$ (via soundness of the verifier outputs and forcing
operator), on $\{\Zobs_a\}$, and on the channel realization
$\{B_{u \to w}\}$.  Three levels of obstruction object recur in the
sequel: (i) realized-transcript certificate-stopping sets, the
finite-length object analyzed here; (ii) graph-only stopping sets,
recovered in the XOR / noiseless-verifier / noiseless-channel
specialization (\Cref{cor:xor-stopping}), where stoppedness depends
only on the bipartite incidence structure; (iii) worst-case or
high-probability structural stopping sets over an ensemble of
transcripts, the natural object for augmentation design
(\Cref{thm:augmentation}).
\end{remark}

\begin{theorem}[Certificate-stopping obstruction]
\label{thm:certstop}
For the deterministic set-valued peeling decoder on a fixed finite
transcript, the terminal unresolved set, when nonempty, is the
unique maximal certificate-stopping set contained in the initially
unresolved variables.  Equivalently, the decoder recovers every
variable iff there is no nonempty certificate-stopping set within
the initially unresolved variables.
\end{theorem}

\begin{proof}
Let $S_t$ be the set of unresolved variables after $t$ peeling steps,
where one step removes every variable that is privately observed or
certified by at least one adjacent check relative to the current set.
The sets $S_t$ are decreasing and the graph is finite, so the process
reaches a fixed point $S_\infty$.

At the fixed point, every variable in $S_\infty$ has erased private
observation; otherwise it would have been removed.  No adjacent check
certifies a member of $S_\infty$ relative to $S_\infty$; otherwise that
member would also have been removed.  Thus $S_\infty$ is a
certificate-stopping set if it is nonempty.

Conversely, let $S$ be any certificate-stopping set contained in the
initially unresolved variables.  We prove by induction that
$S \subseteq S_t$ for all $t \ge 0$.  This is true at $t = 0$ because
each $i \in S$ has erased private observation.  Suppose $S \subseteq
S_t$.  Going from $M_{-i}(S)$ to $M_{-i}(S_t)$ only enlarges incoming
candidate sets at neighbors $k \in S_t \setminus S$, replacing some
singletons by $\Uset$.  Under the soundness invariant
(\Cref{lem:soundness}) the forcing operator $\force_{\theta, j}$
satisfies $\Xstar_i \in \force(z; M_{-j})$ for every input,
so its output is either the certifying singleton $\{\Xstar_i\}$ or
the unresolved set $\Uset$ (never the contradictory singleton
$\{1 - \Xstar_i\}$ or $\emptyset$).  By
\Cref{lem:gamma-monotone}, enlarging $M_{-j}$ enlarges
$\force(z; M_{-j})$, so an output of $\Uset$ at $M_{-j}(S)$ remains
$\Uset$ at $M_{-j}(S_t)$.  Hence no check that failed to certify
$i \in S$ relative to $S$ certifies $i$ relative to $S_t$, no
variable in $S$ is removed at the next step, and
$S \subseteq S_{t+1}$.  Therefore every certificate-stopping set is
contained in $S_\infty$.  Since $S_\infty$ itself is a
certificate-stopping set when nonempty, it is the unique maximal one.
\end{proof}

\subsection{XOR specialization: the classical stopping-set condition}

\begin{corollary}[XOR / parity stopping sets]
\label{cor:xor-stopping}
For XOR templates, a residual set $S$ is certificate-stopping iff
every $i \in S$ has $\Aobs_i = \Uset$ and, for every adjacent check
$a \in \partial i$, the pair $(a, i)$ is blocked in at least one of
the following ways:
\begin{enumerate}[label=(M\arabic*),leftmargin=2.5em]
\item \emph{Verifier-erased.}  $\Zobs_a = *$.
\item \emph{Multi-input combinatorial.}  $|\partial a \cap S| \ge 2$:
the check touches the unresolved set at $i$ and at some other
$k \in \partial a \setminus \{i\}$, so the parity equation has at
least two unknowns and cannot disambiguate.
\item \emph{Reasoning-channel-erased.}  The return channel
$B_{a \to i} = 0$, or some forward channel $B_{k \to a} = 0$ for
$k \in \partial a \setminus \{i\}$.
\end{enumerate}
In the noiseless-verifier ($\epsC = 0$), noiseless-channel ($\eta
\equiv 1$) special case, the per-edge condition reduces to the
classical BEC stopping-set condition: every check that touches $S$
touches $S$ at least twice.
\end{corollary}

\begin{proof}
For XOR, the forcing operator $\force_{\theta, j}(z; M_{-j})$ returns
a singleton iff every other input message is a singleton; relative to
a residual set $S$, this means the verifier output is available,
every other neighbor lies outside $S$, and the directed channels into
and out of the check are available.  Each of (M1) (verifier output
erased), (M2) (a second neighbor in $S$, hence unresolved), and (M3)
(a needed channel down) breaks exactly this condition, and when none
of them holds the check certifies $i$.  An adjacent check therefore
certifies $i$ relative to $S$ iff none of (M1)--(M3) applies, and
\Cref{thm:certstop} converts this into the stated set-level
equivalence.  In the noiseless special case (M1) and (M3) are
vacuous, leaving (M2): every check that touches $S$ touches it at
least twice, the classical stopping-set condition.
\end{proof}

\subsection{AND specialization: positive and negative certificates}

\begin{corollary}[AND certificate-stopping sets]
\label{cor:and-stopping}
For AND templates, an adjacent check $a$ with $\Zobs_a \in \{0, 1\}$
certifies a target variable $i \in \partial a \cap S$ relative to a
residual set $S$ in two ways:
\begin{enumerate}[label=(\roman*),leftmargin=2.5em]
\item \emph{Positive certificate} ($\Zobs_a = 1$): if $\Zobs_a = 1$
and the return channel $B_{a \to i} = 1$, then $a$ certifies
$\Xstar_i = 1$.  Consequently, when $\Zobs_a = 1$ and every return
channel $\{B_{a \to k} : k \in \partial a \cap S\}$ is available,
the single positive verifier output certifies every member of
$\partial a \cap S$ at once.
\item \emph{Negative singleton certificate} ($\Zobs_a = 0$): $a$
certifies $\Xstar_i = 0$ iff $\Zobs_a = 0$, $B_{a \to i} = 1$, and
every other input message $M_{k \to a}(S)$ is the singleton $\{1\}$
(equivalently, every $k \in \partial a \setminus \{i\}$ is outside
$S$, has true value $1$, and has available forward channel
$B_{k \to a} = 1$).
\end{enumerate}
$S$ is AND-certificate-stopping iff every $i \in S$ has $\Aobs_i =
\Uset$ and, for every adjacent check $a \in \partial i$, $a$ fails
both (i) and (ii).
\end{corollary}

\begin{proof}
For AND, the forcing operator $\force_{\theta, j}(z; M_{-j})$ at
$z = 1$ is $\{1\}$ regardless of $M_{-j}$, because no assignment with
$x_j = 0$ produces an AND output of $1$.  This gives (i).  At $z = 0$,
the operator outputs the singleton $\{0\}$ iff setting $x_j = 1$ is
infeasible, i.e., iff every other input message is the singleton $\{1\}$
(otherwise some other variable could be $0$, satisfying $z = 0$
without constraining $x_j$).  This gives (ii).  Combining via
\Cref{thm:certstop} and \Cref{def:certstop} gives the claim.
\end{proof}

\begin{remark}[Two-layer obstruction structure under AND]
\label{rem:two-layer}
Under AND templates the certificate-stopping condition has a richer
combinatorial structure than under XOR.  The positive-certificate
mode is a one-shot strong-recovery primitive: a single positive
verdict can free all unresolved variables in a check's neighborhood
at once.  The negative-certificate mode is closer to the XOR
singleton-neighbor rule but conditional on the boundary's known
values.  Where positive AND outputs are common, AND factors can clear
residual clusters that would be stopping sets under XOR on the same
graph; where negative outputs dominate and boundary variables are
unresolved, the negative-singleton certificate is strictly weaker
than the XOR rule.  A distributional comparison of AND and XOR
certificate-stopping-set sizes therefore depends on
$\{\Vbias_r\}$, $\epsC$, $\eta$, and the ensemble, and we claim no
universal ordering; the full combinatorial study is left to
follow-on work.  We state the two-layer structure here because the
augmentation theorem of the next section applies to both
specializations.
\end{remark}

The toy proof-checking example of \Cref{sec:intro-toy} is this AND
specialization: its successful recovery uses a positive certificate
(i) and a negative singleton certificate (ii), and its three failure
modes (M1)--(M3) are exactly the verifier-erased, combinatorial, and
channel-erased obstructions of \Cref{cor:and-stopping}.

\section{Separating Augmentation}\label{sec:augmentation}

Stopping-set characterizations are useful because they suggest
interventions.  If a small residual cluster is the dominant failure
mode, an architect can add targeted verifier nodes, route the
cluster through stronger roles, or improve the communication links
that would free the cluster.  The augmentation theorem of this section
makes that intuition precise: any augmentation that separates every
small residual pattern eliminates all certificate-stopping sets up to
the corresponding size.

\rhead{Conditioning convention.}  All statements in this section
are deterministic statements conditional on a realized baseline
transcript $\mathcal{T}_{\mathrm{base}}$ as in
\eqref{eq:transcript}.  The class $\mathcal{S}_k = \mathcal{S}_k(
\mathcal{T}_{\mathrm{base}})$ of baseline certificate-stopping sets
of size at most $k$ is a function of that transcript: for general
Boolean factors, certification depends on the realized hidden values,
verifier outputs, and channel-availability, so $k$-separation is not
a purely graph-combinatorial property in the general case.  The
results below distinguish two design regimes: \emph{adaptive
augmentation}, where the architect observes the baseline residual
pattern and then chooses or samples augmenting checks
(\Cref{thm:augmentation,cor:noisy-aug}); and \emph{non-adaptive
augmentation}, where a sampling distribution $\mu$ is fixed before
the realization of erasures, hidden values, and the residual set
(\Cref{thm:random-aug}).  In the non-adaptive case the same bounds
apply after conditioning on the realized transcript, with
system-level guarantees obtained by averaging over the
baseline-transcript
distribution.

\subsection{$k$-separating augmentations}

\begin{definition}[$k$-separating augmentation]
\label{def:ksep}
Fix a baseline transcript $\mathcal{T}_{\mathrm{base}}$ on $\Gn$ and
a class $\mathcal{S}_k \subseteq \mathcal{S}_k(\mathcal{T}_{
\mathrm{base}})$ of baseline certificate-stopping sets of size at
most $k$.  An \emph{augmentation} by additional Boolean verifier
factors is \emph{$k$-separating for $\mathcal{S}_k$} if, for every
nonempty $S \in \mathcal{S}_k$, there exists a variable $i \in S$
and an added check $a^+$ such that, whenever the added verifier
output and the required directed channels are available, $a^+$
certifies $i$ relative to $S$ (in the sense of \Cref{def:certstop}).
\end{definition}

\subsection{Main theorem and noisy-augmentation corollary}

\begin{theorem}[Stopping-set elimination by separating augmentation]
\label{thm:augmentation}
Fix a baseline transcript $\mathcal{T}_{\rm base}$ as in
\eqref{eq:transcript}.  Suppose the added verifier nodes are
noiseless and their directed channels are available.  If the
augmentation is $k$-separating for the family
$\mathcal{S}_k(\mathcal{T}_{\rm base})$ of all baseline
certificate-stopping sets of size at most $k$, then the augmented
transcript has no certificate-stopping set of size at most $k$.
Consequently, if the baseline terminal residual set has size at most
$k$, the augmented peeling decoder recovers every variable in that
residual set.
\end{theorem}

\begin{proof}
Suppose, for contradiction, that the augmented transcript has a
nonempty certificate-stopping set $S^+$ with $|S^+| \le k$.  Every
baseline check is also present in the augmented transcript, so the
fact that no augmented adjacent check certifies any $i \in S^+$
implies, a fortiori, that no baseline adjacent check certifies any
$i \in S^+$ relative to $S^+$; combined with $\Aobs_i = \Uset$ for
every $i \in S^+$, this means $S^+$ is itself a baseline
certificate-stopping set of size at most $k$, hence
$S^+ \in \mathcal{S}_k(\mathcal{T}_{\mathrm{base}})$.  By
$k$-separation, there exists an added check $a^+$ that certifies
some $i \in S^+$ relative to $S^+$.  Because the added verifier
output and required directed channels are available, $a^+$
certifies $i$ in the augmented transcript, contradicting the
augmented certificate-stopping property of $S^+$.  Therefore no such
set exists.  Moreover, by \Cref{lem:soundness} every message in the
augmented decoder (baseline or added) at variable $j$ contains
$\Xstar_j$; the augmented variable update at $j$
(\eqref{eq:var-update}) is therefore an intersection of sets each
containing $\Xstar_j$, so any singleton recovered by the baseline
decoder remains that same singleton in the augmented decoder.  Hence
$S_\infty^+ \subseteq S_\infty^{\mathrm{base}}$, and this augmented
residual set contains no certificate-stopping subset of size at
most $k$; \Cref{thm:certstop} then forces $S_\infty^+ = \varnothing$
whenever $|S_\infty^{\mathrm{base}}| \le k$.
\end{proof}

\begin{corollary}[Noisy augmentation union bound]
\label{cor:noisy-aug}
Fix a baseline transcript $\mathcal{T}_{\mathrm{base}}$ and let
$\mathcal{S}_k = \mathcal{S}_k(\mathcal{T}_{\mathrm{base}})$.  For
each $S \in \mathcal{S}_k$, suppose there are $m_S$ added certifiers
whose success events are conditionally independent given
$\mathcal{T}_{\mathrm{base}}$, and each succeeds (i.e., is both
certifying and available) with probability at least $1 - \zeta_S$.
Then
\begin{equation}
\label{eq:noisy-aug-bound}
   \PR\!\left\{\,\exists S \in \mathcal{S}_k\text{ that remains
   certificate-stopped} \,\,\Big|\,\, \mathcal{T}_{\mathrm{base}}
   \right\}
   \;\le\; \sum_{S \in \mathcal{S}_k} \zeta_S^{m_S}.
\end{equation}
\end{corollary}

\begin{proof}
For a fixed residual pattern $S$, all $m_S$ added certifiers must
fail for $S$ to remain stopped; conditional independence given
$\mathcal{T}_{\mathrm{base}}$ gives a failure probability at most
$\zeta_S^{m_S}$.  Union bound over $\mathcal{S}_k$.
\end{proof}

\rhead{Why an additional random-augmentation bound is useful.}
\Cref{thm:augmentation} is conditional on producing a $k$-separating
augmentation, and \Cref{cor:noisy-aug} is conditional on assigning
certifiers per pattern.  Both leave open the algorithmic question of
\emph{how} to find such an augmentation.  The next theorem answers
this question for the random-augmentation regime: with
$m = O\!\big( \log(|\mathcal{S}_k|/\delta) /
[ q_*(1-\zeta) ] \big)$ independent samples drawn from any
distribution that puts non-trivial mass on separators of each small
residual pattern, the survival probability of any small
certificate-stopping set is at most $\delta$.

\begin{theorem}[Random-augmentation survival bound]
\label{thm:random-aug}
Let $\mu$ be any probability distribution on candidate added Boolean
verifier nodes and their attached channel directions.  For each
nonempty residual pattern $S \in \mathcal{S}_k$, let
\begin{equation}
\label{eq:random-aug-q}
   q_S \,\defeq\, \mu\!\left(\big\{ a^+ : a^+ \text{ certifies some
   } i \in S \text{ relative to } S \big\}\right) \,\in\, [0, 1]
\end{equation}
denote the probability that a single $\mu$-sample separates $S$.
Suppose $m$ added checks are sampled i.i.d.\ from $\mu$, and that
the availability events of the $m$ sampled checks are conditionally
independent given the sampled identities, each with conditional
probability at least $1 - \zeta$ (verifier output and required
directed channels delivered).  Then the probability that the
augmented transcript on $\Gn$ has any certificate-stopping set in
$\mathcal{S}_k$ is at most
\begin{equation}
\label{eq:random-aug-bound}
   \sum_{S \in \mathcal{S}_k}
   \big(1 - q_S (1 - \zeta)\big)^m.
\end{equation}
In particular, if $q_* = \min_{S \in \mathcal{S}_k} q_S > 0$ and
$m \ge \log(|\mathcal{S}_k| / \delta) /
\big[ q_* (1 - \zeta) \big]$, the survival probability is at most
$\delta$.
\end{theorem}

\begin{proof}
Fix a residual pattern $S \in \mathcal{S}_k$.  By \eqref{eq:random-aug-q},
a single $\mu$-sample certifies some $i \in S$ relative to $S$ with
probability at least $q_S$, and is then available (with both verifier
output and required directed channels delivered) with probability at
least $1 - \zeta$.  By the conditional independence of availability
events across the $m$ samples (hypothesis of the theorem) and the
i.i.d.\ sampling from $\mu$,
$\PR\{\text{$S$ remains stopping after $m$ samples}\}
\le (1 - q_S(1 - \zeta))^m$.
Union bound over $\mathcal{S}_k$ gives \eqref{eq:random-aug-bound}.
The displayed sample-complexity statement follows from $1 - x \le e^{-x}$:
$(1 - q_*(1 - \zeta))^m \le e^{-m q_* (1 - \zeta)}$, and the choice of
$m$ ensures the right-hand side is at most $\delta / |\mathcal{S}_k|$.
As observed in the proof of \Cref{thm:augmentation}, every baseline
check is present in the augmented transcript, so any
certificate-stopping set of size at most $k$ in the augmented
transcript is also a baseline certificate-stopping set of size at
most $k$ (no augmented check certifying $i \in S^+$ implies, a
fortiori, no baseline check certifies $i$ either).  Hence when
$\mathcal{S}_k$ is taken as the family of all baseline
certificate-stopping sets of size at most $k$, the bound
\eqref{eq:random-aug-bound} controls every small augmented stopping
set, not only the preselected baseline family.
\end{proof}

\rhead{Operational reading.}  \Cref{thm:random-aug} provides a
constructive design recipe: any sampling distribution whose minimal
separation probability $q_*$ over $\mathcal{S}_k$ is bounded away
from zero requires only $O(\log |\mathcal{S}_k| / q_*)$ auxiliary
checks to drive the survival probability of any small residual
pattern to vanish.  For fixed $k$ the simple counting bound
$|\mathcal{S}_k| \le \sum_{j=1}^k \binom{n}{j} = O(n^k)$ holds in
the role-typed configuration ensemble; if additionally $q_* \ge q_0
> 0$ and $\zeta \le \zeta_0 < 1$ uniformly in $n$, then $m = O(\log
n)$ samples suffice to drive the survival probability below any
prescribed $\delta$.  The uniformity hypothesis on $q_*$ is the
essential design constraint: in non-adaptive schemes that sample
candidate checks almost uniformly over many possible variable
subsets, the probability of hitting a separator for one particular
small residual pattern can scale like a negative power of $n$; in
such regimes the $O(\log n)$ count fails.  The theorem also
disentangles the augmentation argument from the property of being a
$k$-separator: the random-sampling distribution is the design knob,
and the analyst chooses $\mu$ to make $q_*$ as large as possible for
the relevant family $\mathcal{S}_k$.

\subsection{XOR specialization: two-edge-connected freeing-set
augmentation}

\rhead{XOR specialization.}  Assume added XOR checks, noiseless
added verifier outputs, and perfect added channels.  Then an added
check $a^+$ separates a residual set $S$ relative to itself iff
\begin{equation}
\label{eq:xor-separator}
   |\partial a^+ \cap S| = 1.
\end{equation}
Indeed, by \Cref{cor:xor-stopping}, the XOR forcing operator
$\force_{\theta, j}(z; M_{-j})$ is a singleton iff every entry of
$M_{-j}$ is a singleton; relative to a residual set $S$, this means
exactly one input of $a^+$ lies in $S$ and every other input is
outside $S$ with available forward channel.  Consequently, an added
XOR layer is $k$-separating for $\mathcal{S}_k$ iff for every
nonempty $S \in \mathcal{S}_k$ there exists an added check $a^+$
satisfying \eqref{eq:xor-separator}.

\begin{remark}[Architecture interpretation by primitive]
\label{rem:aug-arch}
The separating-augmentation principle reads differently across
primitives.  For \emph{XOR} checks it creates singleton parity
witnesses via \eqref{eq:xor-separator}, reducing in the single-role
$(d_v, d_c)$-regular case to the classical two-edge-connected
freeing-set construction \cite{pishro2007punctured}.  For \emph{AND}
checks, augmentation adds tests that turn a positive local condition
into a one-shot certificate or isolate a negative cause, and the
required graph property is not two-edge-connectedness but a condition
involving the boundary's known values.  For \emph{Horn} checks,
augmentation adds proof obligations that force a missing premise or
conclusion.  The common principle is semantic: add local verifier
factors that force at least one variable in every small residual
pattern.
\end{remark}

\section{Cost-Constrained Architecture Optimization}\label{sec:optimization}

The density-evolution recursion converts architecture design into an
optimization problem.  A design parameter $\lambda$ collects role
proportions $\{\pi_r^V\}$, degree-law probabilities, template
proportions $\{\pi_\theta^C\}$, verifier reliabilities
$\{\epsV[r], \epsC[s]\}$, and communication fidelities
$\{\eta_{r,s}\}$ (we treat augmentation choices in
\Cref{rem:augmentation-in-lambda} below).  The design space
$\Designspace$ is a finite union of bounded-support strata; each
stratum fixes the supports of all degree laws and the template set,
and optimizes only over simplex coordinates and continuous
reliability parameters.  On each stratum the feasible set under
budget $B$ is
\begin{equation}
\label{eq:budget-set}
   \Designspace_B
   = \big\{ \lambda \in \Designspace :
       \Cost(\lambda) \le B,\,
       g_\tau(\lambda) = 0\;\forall \tau \in \Sockets,\,
       \lambda \in \Lambda\big\},
\end{equation}
where the equality constraints
\begin{equation}
\label{eq:socket-balance-opt}
   g_\tau(\lambda) \defeq
   \pi^V_{\socketrole{\tau}} \cdot
   \E_\lambda\!\big[D^{(\socketrole{\tau})}_\tau\big]
   - \alpha\, \pi^C_\theta = 0,
   \qquad \tau = (\theta, j),
\end{equation}
encode the socket-balance condition \eqref{eq:socket-balance} of
\Cref{sec:model-ensemble}, and $\Lambda$ collects simplex sums
($\sum_r \pi_r^V = 1$, $\sum_\theta \pi_\theta^C = 1$,
$\sum_d \mathsf{P}_D^{(r)}(d) = 1$) and box constraints
($0 \le \epsV[r], \epsC[s] \le 1$, $0 \le \eta_{r,s} \le 1$).  The
cost may include role costs, verifier-invocation costs,
communication-edge costs, and augmentation costs.  For a fixed
round budget $L \in \N$, define the objective
\begin{equation}
\label{eq:JL}
   J_L(\lambda) = \Pdee^{(L)}(\lambda),
\end{equation}
or any continuous monotone function thereof (e.g., a weighted
residual across role groups).

\begin{theorem}[Cost-constrained Boolean-verifier architecture
optimization]
\label{thm:optimization}
Assume $\Designspace_B$ in \eqref{eq:budget-set} is nonempty, that
each stratum is compact under the socket-balance, simplex, box, and
budget constraints, and that $J_L$ is continuous on each stratum.
Then:
\begin{enumerate}[label=(\alph*),leftmargin=2em]
\item \emph{Existence (Weierstrass).}  For every finite $L$, there
exists a design $\lambda_L^* \in \Designspace_B$ minimizing $J_L$.
\item \emph{Asymptotic optimality (finite design class).}  For each
fixed design $\lambda \in \Designspace_B$, the empirical $L$-round
residual of the random graph ensemble converges in probability to
$J_L(\lambda)$ by \Cref{thm:de}.  We interpret
$\Pbit^{(L)}(\Gn; \lambda)$ as follows: each design $\lambda$ has its
own role-typed configuration-model ensemble (since $\lambda$ may
include role proportions, template proportions, and degree
distributions), and $\Gn(\lambda)$ denotes one draw from that
$\lambda$-parameterized ensemble; the empirical residual is
$\Pbit^{(L)}(\Gn(\lambda); \lambda)$, with draws across different
$\lambda$ values mutually independent.  Let
\[
   \widehat\lambda_n
   \,\in\, \arg\min_{\lambda \in \Designspace_B}\;
   \Pbit^{(L)}(\Gn(\lambda); \lambda)
\]
denote the resulting empirical minimizer.  When
$\Designspace_B$ is finite, a union bound over \Cref{thm:de} gives
the uniform convergence
$\sup_{\lambda \in \Designspace_B} |\Pbit^{(L)}(\Gn(\lambda); \lambda) -
J_L(\lambda)| \xrightarrow{\PR} 0$, hence the empirical-minimizer
suboptimality vanishes:
$J_L(\widehat\lambda_n) -
\inf_{\lambda \in \Designspace_B} J_L(\lambda)
\xrightarrow{\PR} 0$.  For compact-continuum strata, the same
conclusion holds under a uniform-DE-concentration assumption (uniform
in $\lambda$ over the stratum) that we do not establish here; this
extension requires finite-$n$ Lipschitz or covering-number bounds on
the empirical-residual functional beyond \Cref{thm:de}.  In the
degenerate case where $\lambda$ enters only through reliability
parameters and the graph ensemble itself is held fixed across
designs, $\Gn(\lambda) \equiv \Gn$ and only channel randomness varies
with $\lambda$; the same finite-vs-continuum dichotomy then applies
with potentially simpler uniform-concentration arguments, but this
restricted setting is not the framing pursued here.
\item \emph{Limit-points.}  If $J_L \to J_\infty$ uniformly on
$\Designspace_B$, then every limit point of finite-round minimizers
$\{\lambda_{L_m}^*\}$ is an infinite-round minimizer.  Sufficient
conditions for the uniform limit include monotone pointwise
convergence $J_L \downarrow J_\infty$ together with continuity of
$J_\infty$ on $\Designspace_B$ (Dini's theorem); the uniform limit
may fail at threshold surfaces, where $J_\infty$ is discontinuous.
\item \emph{Adjoint sensitivity (backward mode).}  On any smooth
stratum of $\Designspace$, with recursion $\bm p_{\ell+1} =
\DEmap_\lambda(\bm p_\ell)$ and objective $J_L = \psi(\bm p_L,
\lambda)$, the gradient is given by the adjoint equations
\begin{align}
\label{eq:adjoint}
   \xi_L &= \nabla_{\bm p}\, \psi(\bm p_L, \lambda), \\
   \xi_\ell &= \big(D_{\bm p}\, \DEmap_\lambda(\bm p_\ell)\big)^T\,
   \xi_{\ell+1},
   \qquad \ell = L - 1, L - 2, \ldots, 0,
\end{align}
and
\begin{equation}
\label{eq:adjoint-grad}
   \nabla_\lambda\, J_L
   = \nabla_\lambda\, \psi(\bm p_L, \lambda)
   + \sum_{\ell=0}^{L-1} \big(D_\lambda\, \DEmap_\lambda(\bm p_\ell)\big)^T\,
   \xi_{\ell+1}
   + \big(D_\lambda\, \bm p_0(\lambda)\big)^T\, \xi_0.
\end{equation}
The last term captures the dependence of the initial state on the
design parameters: when $\lambda$ includes variable-side erasure
rates, $\bm p_0$ is initialized from those rates and
$D_\lambda\, \bm p_0(\lambda) \ne 0$.  When $\lambda$ enters only
through the recursion (e.g.\ degree distributions or verifier-side
parameters that leave $\bm p_0$ fixed), the boundary term vanishes
and \eqref{eq:adjoint-grad} reduces to the standard sum-only form.
Variable-side erasure $\epsV$ enters $\nabla_\lambda J_L$ through
both terms: the boundary term (since $\bm p_0$ is initialized
from $\epsV$) and the direct sum (since $\epsV$ also appears
multiplicatively inside $\DEmap_\lambda$ in \eqref{eq:pupdate});
implementations of the gradient must include both contributions.
\item \emph{KKT conditions and shadow prices.}  At a regular local
optimum on a smooth stratum, the Karush-Kuhn-Tucker conditions hold
for the Lagrangian
\begin{equation}
\label{eq:lagrangian}
   \mathcal{L}(\lambda, \mu, \nu, \alpha, \gamma)
   = J_L(\lambda)
   + \mu\big(\Cost(\lambda) - B\big)
   + \sum_\tau \nu_\tau\, g_\tau(\lambda)
   + \sum_j \alpha_j\, a_j(\lambda)
   + \sum_k \gamma_k\, b_k(\lambda),
\end{equation}
with $g_\tau$ the socket-balance equalities, $a_j$ the simplex-sum
equalities, and $b_k(\lambda) \le 0$ the box and nonnegativity
inequalities; stationarity $\nabla_\lambda \mathcal{L} = 0$ holds
along with $\mu \ge 0$, $\mu(\Cost(\lambda) - B) = 0$, and
analogous complementary slackness on the box constraints.  The
multiplier $\mu$ is the shadow price of the cost budget.  The
partial derivatives $\partial J_L / \partial \epsV[r]$,
$\partial J_L / \partial \epsC[s]$, $\partial J_L / \partial
\eta_{r,s}$ obtained from \eqref{eq:adjoint-grad} are the
sensitivity gradients with respect to each erasure-tier parameter,
not Lagrange multipliers themselves; in formulations where
reliabilities are generated by explicit investment variables
(\Cref{rem:investment-variables}) these partials translate directly
into shadow prices on the corresponding investments.
\end{enumerate}
\end{theorem}

\begin{remark}[Nonconvexity and local-optimality scope]
\label{rem:nonconvex}
The design objective $J_L(\lambda) = \Pdee^{(L)}(\lambda)$ is in
general a nonconvex function of $\lambda$, built from products and
sums of degree-law moments and reliability parameters through the
multi-round recursion.  The KKT conditions of
\Cref{thm:optimization}(e) are therefore \emph{necessary local}
optimality conditions, not a global solution method; the
shadow-price interpretation of $\mu$ and of the investment-variable
duals in \Cref{rem:investment-variables} is likewise local unless
additional convexity or monotonicity structure is imposed.  Gradient
methods on $\nabla_\lambda J_L$ converge to local minima; global
guarantees require further assumptions such as convex relaxations or
enumeration over discrete template choices.
\end{remark}

\begin{proof}
\emph{(a)} $\Designspace_B$ is compact by assumption and $J_L$ is
continuous because the finite-round recursion is built from finitely
many sums and products of continuous functions on each finite-support
stratum, with continuous extension across mixtures.  Weierstrass's
theorem applies.

\emph{(b)} For a fixed design, the empirical residual
$\Pbit^{(L)}(\Gn; \lambda)$ converges in probability to $J_L(\lambda)$
by \Cref{thm:de}.  When $\Designspace_B$ is finite, a union bound
over \Cref{thm:de} gives the uniform convergence
$\sup_{\lambda \in \Designspace_B} |\Pbit^{(L)}(\Gn; \lambda) -
J_L(\lambda)| \xrightarrow{\PR} 0$.  Standard
empirical-risk-minimization argument: for the empirical minimizer
$\widehat\lambda_n$
and any DE optimum $\lambda^*$,
$J_L(\widehat\lambda_n) - J_L(\lambda^*) \le \big[J_L(\widehat\lambda_n)
- \Pbit^{(L)}(\Gn;\widehat\lambda_n)\big] +
\big[\Pbit^{(L)}(\Gn;\widehat\lambda_n) - \Pbit^{(L)}(\Gn;\lambda^*)\big]
+ \big[\Pbit^{(L)}(\Gn;\lambda^*) - J_L(\lambda^*)\big]$, where the
middle bracket is $\le 0$ by the empirical-minimizer property and the
two outer brackets vanish in probability by uniform convergence.  For
a compact-continuum stratum, the same conclusion holds under the
stated uniform-DE-concentration assumption; we do not establish that
uniform limit here, since it would require finite-$n$ Lipschitz /
covering-number bounds on the empirical-residual functional that go
beyond \Cref{thm:de}.

\emph{(c)} Standard uniform-convergence argument.  If $\lambda_{L_m}^*
\to \bar\lambda$, then for any feasible $\lambda$,
\[
   J_\infty(\bar\lambda)
   = \lim_m J_{L_m}(\lambda_{L_m}^*)
   \le \lim_m J_{L_m}(\lambda)
   = J_\infty(\lambda),
\]
where uniform convergence justifies the interchange of limits.  The
sufficient condition $J_L \downarrow J_\infty$ with $J_\infty$
continuous follows from Dini's theorem on a compact stratum.

\emph{(d)} Reverse-mode differentiation of the finite recursion.
Differentiating $\bm p_{\ell+1} = \DEmap_\lambda(\bm p_\ell)$ with
respect to $\lambda$ and applying the chain rule backwards gives the
sum over $\ell = 0, \ldots, L-1$ in \eqref{eq:adjoint-grad}; the
$(D_\lambda\, \bm p_0(\lambda))^T\, \xi_0$ boundary term arises from
the dependence of the initial state on $\lambda$ at $\ell = 0$.  This
is the standard adjoint-state formula in finite-horizon optimal
control with parameter-dependent initial condition, or equivalently
backpropagation through a recurrent map.  Full derivation including
the explicit block structure of $D_{\bm p}\, \DEmap_\lambda$ and
$D_\lambda\, \DEmap_\lambda$ and the $\bm p_0$-dependence on the
variable-side tier is in \Cref{app:adjoint-detail}.

\emph{(e)} Standard KKT first-order necessary conditions for a
regular local optimum in a finite-dimensional smooth program with
the equality and inequality constraints displayed in
\eqref{eq:lagrangian}; complementary slackness for the budget
inequality gives the budget shadow price $\mu \ge 0$.  The
sensitivity-gradient interpretation of the partial derivatives with
respect to reliability parameters follows from \eqref{eq:adjoint-grad}
directly.
\end{proof}

\begin{remark}[Investment variables and operational shadow prices]
\label{rem:investment-variables}
Practical architecture-design problems often expose a budget through
explicit investment variables $u_r^V, u_s^C, u_{r,s}^\eta \ge 0$
with monotone reliability response curves $\epsV[r] = \epsV[r](u_r^V)$
(decreasing in $u_r^V$), $\epsC[s] = \epsC[s](u_s^C)$ (decreasing in
$u_s^C$), and $\eta_{r,s} = \eta_{r,s}(u_{r,s}^\eta)$ (increasing in
$u_{r,s}^\eta$), under a budget
$\sum_r c_r^V u_r^V + \sum_s c_s^C u_s^C + \sum_{r,s} c_{r,s}^\eta
u_{r,s}^\eta + \cdots \le B$.  The stationarity condition for an
interior active investment then reads, e.g.\ for variable-side,
$-(\partial J_L / \partial \epsV[r]) (d\epsV[r]/du_r^V) =
\mu\, c_r^V$, with analogous equations for verifier-side and
reasoning-channel investments (signs adjusted for the
residual-minimization objective).  Each such equation is a
per-unit-cost
shadow-price rule: at the optimum, the marginal residual reduction
per unit cost is equal across all active investments.  The
non-interchangeability proposition (\Cref{prop:noninterchange}) says
these shadow-price ratios depend on the design itself; an architect
cannot trade verifier reliability against proposer redundancy
one-for-one at a fixed exchange rate independent of the operating
point.
\end{remark}

\begin{remark}[Augmentation choices in $\lambda$]
\label{rem:augmentation-in-lambda}
\Cref{thm:optimization} treats $\lambda$ as an ex ante ensemble
parameter (role mix, degree laws, template proportions, reliability
tiers).  Augmentation can enter in two distinct ways: (i) as an
ex ante added-template layer, in which case the augmentation
template proportions and per-template costs are part of $\lambda$
and the framework above applies directly; or (ii) as an adaptive
or random repair policy whose distribution is realized only after
the baseline transcript is observed, in which case the optimization
problem must include the distribution of residual transcripts and
the expected repair cost.  We treat (i) within
\Cref{thm:optimization}; (ii), which is closer to a Markov decision
problem on residual transcripts, is left to follow-on work.
\end{remark}

\begin{corollary}[Budget monotonicity]
\label{cor:budget}
Increasing the available budget $B$, lowering a role cost, lowering a
verifier-template cost, or lowering a communication cost cannot worsen
the optimized residual objective.  For threshold objectives formulated
as maximization of an admissible noise level, the optimized threshold
cannot decrease when the feasible design set is enlarged.
\end{corollary}

\begin{proof}
Each such change enlarges $\Designspace_B$.  The minimum of a residual
objective over a larger feasible set cannot be larger; the maximum of
a threshold objective over a larger feasible set cannot be smaller.
\end{proof}

\begin{remark}[Channel-agnostic adjoint-sensitivity framework]
\label{rem:channel-agnostic-opt}
The Weierstrass existence in part~(a), the adjoint sensitivity
equations of part~(d), and the KKT shadow-price machinery of
part~(e) depend only on a continuously-differentiable DE
fixed-point map and a smooth design parametrization, not on the
message alphabet being $\Mset = \{\{0\}, \{1\}, \Uset\}$.  The same
architecture-optimization framework therefore applies, with
appropriate operator-theoretic setup, to non-erasure DE recursions
such as LLR-density DE on memoryless symmetric channels; the
role-typed shadow prices on the three erasure tiers are the
erasure-specific reading of a channel-agnostic multiplier framework.
\end{remark}

\begin{remark}[XOR/UEP convex specializations]
\label{rem:xor-water-filling}
In certain XOR/UEP specializations, under fixed code-ensemble
hypotheses and an edge-perspective parametrization, known
log-convexity results \cite{pishro2007nonuniform} yield convex or
water-filling-like allocation rules.  We do not rely on this in the
general Boolean-verifier theory: the finite-$L$ recursion
\eqref{eq:pbar}--\eqref{eq:pupdate} involves nested nonlinear products
of degree-distribution and reliability variables, so the program is in
general a non-convex NLP, for which the adjoint equations of part~(d)
and the KKT system of part~(e) give gradients and local necessary
conditions.  Closed-form characterizations are available only where
additional convexity or separability holds.
\end{remark}

\section{Converse: A Local-Soundness Bound on the Computation Tree}\label{sec:converse}

The achievability theorems \Cref{thm:de}--\Cref{thm:optimization}
specify when a particular protocol, extrinsic edge-specific message
passing on the role-typed Boolean-verifier-node ensemble, recovers
the hidden subclaim vector.  This section pairs them with a converse
direction.  We prove that, in the erasure model and within the class
of $T$-round \emph{sound} (certifying) local message-passing
protocols, no protocol can asymptotically leave fewer variables
unresolved at its terminal output than the value-conditioned
logical-forcing decoder of \Cref{thm:de}.
Throughout this section, $T$ denotes the round budget of the local
protocol class; this is distinct from the true verifier output $T_a$
introduced in \eqref{eq:check-true}; context always distinguishes the
two.  A
sound protocol is one whose non-erased outputs are correct almost
surely; soundness is the natural restriction for an erasure theory
in which messages may be missing but never wrong.  Without it, the
bound need not hold: a protocol that guesses the likeliest prior
value on undetermined transcripts can attain a smaller failure rate
than abstention, at the cost of producing wrong outputs.  That is the
regime of a BSC/absorbing-set converse, not the present erasure
converse.

\rhead{Scope of the converse.}  This is a
\emph{local-soundness} converse: a per-variable certifiability bound
within the certifying class, not an information-theoretic limit
against an unbounded-alphabet unrestricted class.  A stronger
Fano-cut-set converse, in which the per-edge alphabet is bounded by
$Y$ and the matching $T \to \infty$, $Y \to \infty$ joint limit is
taken, is sketched as a future direction in \Cref{sec:outlook}.

\rhead{Round and radius convention.}  Throughout this section,
one \emph{round} of message passing consists of one
variable-to-check update followed by one check-to-variable update.
After $T$ such rounds the terminal estimator at variable $i$ is, by
unrolling the recursion, a measurable function of the channel-gated
depth-$R_{\mathrm{node}}(T) = 2T + 2$ rooted variable neighborhood,
where depth is measured in graph distance (consistent with the
radius convention of \Cref{app:de-proof}: edge messages
$V^{(T)}_{i \to a}$ use the depth-$R_{\mathrm{edge}}(T) = 2T+1$
edge-rooted neighborhood, but the terminal node estimate at $i$
gains one additional hop in order to read all incoming round-$T$
check-to-variable messages).  The round index $T$ on the protocol
side is aligned with the iteration index $T$ in the achievability
decoder (\Cref{thm:de}), so the converse compares a $T$-round
protocol's per-variable estimate to $\Pdee^{(T)}(\lambda)$ at the
matching node depth $2T + 2$.

\subsection{The \texorpdfstring{$T$}{T}-round local protocol class}
\label{sec:converse-class}

\begin{definition}[Sound local message-passing protocol class
$\mathcal{P}_T^{\mathrm{snd}}$]
\label{def:protocol-class}
A \emph{$T$-round sound local message-passing protocol} on the
role-typed configuration ensemble of \Cref{sec:model-ensemble}
consists of:
\begin{enumerate}[label=(\alph*),leftmargin=2em]
\item For each directed edge $(i \to a)$ and each round
$t = 0, 1, \ldots, T$, a measurable update function
\[
   \Phi^{(t)}_{i \to a} : \mathcal{V}^{(t-1)}_i \times \Omega_i
   \to \mathcal{Y},
\]
where the round-$(t-1)$ local view $\mathcal{V}^{(t-1)}_i$ at
variable $i$ consists of the private observation $\Aobs_i$, the
role label $r_i = \rolemap(i)$, the incident socket types and
template labels of all checks $c \in \partial i$, and the messages
$\widetilde Y_{c \to i}^{(s)}$ received from each check neighbor in
rounds $s \le t - 1$ in the delivered alphabet defined below; at
the round-$0$ initialisation step the index set $\{s : s \le -1\}$
is empty, so $\mathcal{V}^{(-1)}_i$ contains only the observation,
role, and graph-structure information and $\Phi^{(0)}_{i \to a}$
plays the role of the achievability decoder's initial variable
message $V^{(0)}_{i \to a} = \Aobs_i$ of \eqref{eq:var-update}.
$\Omega_i$ is a private-randomness space and $\mathcal{Y}$
is an arbitrary standard Borel measurable message space (distinct
from the logical-forcing alphabet $\Mset = \{\{0\},\{1\},\Uset\}$
of \Cref{sec:model-forcing}; the converse holds against protocols
with arbitrary standard-Borel message alphabets, not only those
whose messages take values in $\Mset$, the standard-Borel hypothesis
guaranteeing existence of regular conditional probabilities used in
the proof).  The persistent role-pair channels of
\Cref{sec:model-reasoning-channel} extend to the alphabet
$\mathcal{Y}$ as follows: the delivered alphabet is
$\widetilde{\mathcal{Y}} \defeq \mathcal{Y} \cup \{\bot\}$, and for
each directed edge $u \to v$ a persistent gate $B_{u \to v}
\sim \mathrm{Bernoulli}(\eta_{\rolemap(u), \rolemap(v)})$ is drawn
once at $t = 0$; the received message at round $t$ is
$\widetilde Y_{u \to v}^{(t)} = Y_{u \to v}^{(t)}$ when $B_{u \to v}
= 1$ and $\bot$ otherwise.  This reduces to the logical-forcing
channel \eqref{eq:channel} when $\mathcal{Y} = \Mset$ and $\bot
\equiv \Uset$.
\item For each directed edge $(a \to i)$ and each round
$t = 0, 1, \ldots, T$, a measurable update function
\[
   \Psi^{(t)}_{a \to i} : \mathcal{V}^{(t)}_a \times \Omega_a
   \to \mathcal{Y},
\]
where the round-$t$ local view $\mathcal{V}^{(t)}_a$ at check
$a$ consists of the verifier output $\Zobs_a$, the template label
$\theta_a$ and its socket ordering $(j_1, \ldots, j_{d_a})$, the
role labels and incident socket types of all neighbors
$k \in \partial a$, and the messages $\widetilde Y_{k \to a}^{(s)}$
received from each variable neighbor $k \in \partial a$ in rounds
$s \le t$ in the delivered alphabet $\widetilde{\mathcal{Y}}$
(\emph{within-round ordering}: the round-$t$ check update may
read the round-$t$ output of clause~(a) of the same round, in
addition to all earlier variable outputs; see the alignment
paragraph following this definition).
The persistent-gate channel of clause (a) also gates the outgoing
$\mathcal{Y}$-valued message from $a$: the gate $B_{a \to i}$
replaces $Y_{a \to i}^{(t)}$ by $\bot$ with probability
$1 - \eta_{\rolemap(a), \rolemap(i)}$.
\item A measurable terminal estimator
$\widehat X_i^{\Pi} : \mathcal{V}^{(T)}_i \times \Omega_i \to
\{0, 1, ?\}$ for each $i \in \Vset_n$ satisfying the
\emph{soundness} (certifiability) constraint
\begin{equation}
\label{eq:soundness}
   \PR\!\left\{\widehat X_i^{\Pi} \notin \{?,\, \Xstar_i\}\right\}
   = 0
   \qquad \text{for every $i \in \Vset_n$,}
\end{equation}
that is, every non-erased terminal output must equal the true value
almost surely under the joint distribution of the graph,
observations, hidden values, and private randomness.
\end{enumerate}
The class of all such sound protocols, for fixed $T$, is denoted
$\mathcal{P}_T^{\mathrm{snd}}$.
\end{definition}

\rhead{Within-round ordering and alignment with the achievability
decoder.}
Within each round $t \in \{0, 1, \ldots, T\}$ the variable update
of clause~(a) fires first and the check update of clause~(b) fires
second, so the round-$t$ check update $\Psi^{(t)}_{a \to i}$ may
use the round-$t$ variable messages produced by its neighbors in
addition to all earlier variable outputs.  Under this ordering the
protocol-class messages $\Phi^{(t)}_{i \to a}$ and
$\Psi^{(t)}_{a \to i}$ play the roles of the achievability
extrinsic messages $V^{(t)}_{i \to a}$ and $C^{(t)}_{a \to i}$ of
\eqref{eq:var-update}--\eqref{eq:check-update} respectively for
$t = 0, 1, \ldots, T$: round~$0$ is the initialisation step in
which $\Phi^{(0)}_{i \to a}$ sees only the private observation
$\Aobs_i$ (no check has spoken yet) and $\Psi^{(0)}_{a \to i}$
then reads those round-$0$ variable initialisations from
$\partial a \setminus \{i\}$, paralleling the achievability pair
$V^{(0)}_{i \to a} = \Aobs_i$ and $C^{(0)}_{a \to i} =
\force_{\theta, j}(\Zobs_a; (\widetilde V^{(0)}_{k \to a})_{k \neq i})$
of \eqref{eq:check-update}.  Iterating for $t = 1, \ldots, T$, the
round-$T$ local view $\mathcal{V}^{(T)}_i$ that the terminal
estimator $\widehat X_i^{\Pi}$ may consult is, by unrolling the
recursion, a measurable function of the channel-gated rooted
variable neighborhood of $i$ at graph depth $R_{\mathrm{node}}(T)
= 2T+2$ (the radius convention of \Cref{app:de-proof}), exactly
matching the achievability terminal estimator $\widehat M_i^{(T)}$
of \eqref{eq:final-est}.  The converse comparison to
$\Pdee^{(T)}(\lambda)$ stated in \Cref{thm:converse} below is
therefore against the achievability $T$-iterate at the matching
node depth, and the round count $T$ on the two sides refers to the
same number of message-passing iterations counted from initialisation.

The class $\mathcal{P}_T^{\mathrm{snd}}$ is much larger than the
specific extrinsic decoder analyzed in \Cref{thm:de}: it includes
any local update that is a measurable function of the local view,
including non-extrinsic updates, randomized updates,
soft-information updates, and any sound decoder that exploits
messages from earlier rounds in non-trivial ways.  The constraints
are locality (messages pass along graph edges), the round budget
$T$, and the certifiability requirement \eqref{eq:soundness}.
Soundness is the natural matching constraint for an erasure model
\emph{(\Cref{lem:soundness})}: the achievability decoder of
\Cref{thm:de} is itself sound, and the converse measures whether
any sound local protocol can asymptotically leave fewer variables
unresolved at its terminal output than that specific sound decoder.

\subsection{The local-soundness converse}
\label{sec:converse-theorem}

\begin{theorem}[Local-soundness converse on the computation tree]
\label{thm:converse}
Fix any $T \in \N$.  Assume full-support value priors
$\Vbias_r \in (0, 1)$ for every $r \in \RolesV$.  Under the
erasure-only observation model (\Cref{sec:model-observations}) on
the bounded-degree role-typed configuration ensemble of
\Cref{sec:model-ensemble} with persistent role-pair channel gates,
every sound protocol $\Pi \in \mathcal{P}_T^{\mathrm{snd}}$
satisfies the ensemble-average lower bound
\begin{equation}
\label{eq:local-map-converse}
   \liminf_{n \to \infty}\, \E_{\Gn}\!\left[\frac{1}{n}
   \sum_{i \in \Vset_n}
   \PR\!\left\{\widehat X_i^{\Pi} = \,?\,\right\}\right]
   \;\ge\; \Pdee^{(T)}(\lambda),
\end{equation}
where $\Pdee^{(T)}(\lambda)$ is the per-variable terminal
unresolved fraction of the value-conditioned logical-forcing
decoder of \Cref{thm:de}, evaluated at iteration $T$.  The bound is
on the joint expectation over the graph ensemble and the protocol's
randomness, not a high-probability statement for almost every
graph.
\end{theorem}

\rhead{In words.}  Within the certifying class
$\mathcal{P}_T^{\mathrm{snd}}$, no $T$-round local protocol can
asymptotically leave fewer variables unresolved at its terminal
output than the value-conditioned logical-forcing decoder of
\Cref{thm:de}.  Equivalently, the logical-forcing decoder is
asymptotically optimal in the \emph{certifiability} sense: any sound
protocol that leaves fewer variables unresolved on a positive
fraction of instances would, by soundness,
already commit to the true value on those instances; but the proof
below shows there is no measurable function of the depth-$(2T+2)$
local view that does so almost surely beyond what logical forcing
already achieves.

\rhead{Terminal abstention vs.\ primitive erasure.}  The quantity
bounded in \eqref{eq:local-map-converse} is the protocol's
\emph{terminal abstention}: the event $\widehat X_i^{\Pi} = \,?$
records that variable $i$ is still uncertified \emph{after} all
observations, verifier outputs, channel gates, and $T$ rounds of
local computation have been used.  It is distinct from the model's
primitive erasure events, the variable-side abstention
$\Aobs_i = \Uset$, the verifier-side erasure $\Zobs_a = *$, and the
role-pair channel erasures $B_{u \to v} = 0$, which generate the
uncertainty rather than constitute the protocol's verdict.  We
accordingly read $\tfrac{1}{n}\sum_{i} \PR\{\widehat X_i^{\Pi} = \,?\}$
as the average \emph{terminal abstention probability}.

\begin{proof}
The proof has three steps: localization, tree representation, and
optimal sound estimator on the tree.

\emph{Step 1 (observable transcript dominance).}  By
\Cref{def:protocol-class}, every variable update
$\Phi^{(t)}_{i \to a}$ at round $t \ge 0$ is computed from the
round-$(t-1)$ local view at $i$ (with $\mathcal{V}^{(-1)}_i$
containing only the observation, role, and graph-structure
information at the round-$0$ initialisation step), every check
update $\Psi^{(t)}_{a \to i}$ at round $t \ge 0$ is computed from
the round-$t$ local view at $a$ (which includes the round-$t$
variable outputs of clause~(a) under the within-round ordering),
and each directed $\mathcal{Y}$-valued message is gated by its
persistent Bernoulli edge variable $B_{u \to v}$.  Iterating, the
round-$T$
local view at $i$ is a measurable function of the channel-gated
observable transcript
\begin{multline}
\label{eq:obs-transcript}
   \mathsf{Obs}_i^{(T)}
   \defeq \Big\{
   \big(\Aobs_j, r_j\big)_{j},\;
   \big(\Zobs_a, \theta_a, (j_1, \ldots, j_{d_a})\big)_{a},\;
   \big(B_{u \to v}\big)_{u \to v}\, :\\
   (j, a, u \to v)\text{ in the channel-gated depth-}(2T+2)\text{
   rooted variable neighborhood of }i,\\
   \text{accessible through open channels at the relevant round}
   \Big\},
\end{multline}
i.e., the channel-gated rooted typed factor-graph structure to
depth $2T+2$ together with the private observations $\Aobs_j$,
verifier outputs $\Zobs_a$, and forward/return channel gates
$B_{u \to v}$ encountered along open channels.  We pair this with
the protocol's private randomness $\Omega_i$.  Verifier outputs
behind erased return channels and incoming messages on erased
forward channels do not enter $\mathsf{Obs}_i^{(T)}$.  The terminal
estimator $\widehat X_i^{\Pi}$ is therefore a measurable function
of $\mathsf{Obs}_i^{(T)}$ and the private randomness.  Because $\mathcal{Y}$ and $\Omega_i$ are standard Borel, the
regular conditional probability
$\PR\{\Xstar_i \mid \mathsf{Obs}_i^{(T)}\}$ exists and the posterior
support is well-defined for almost every realization of
$\mathsf{Obs}_i^{(T)}$.  Soundness \eqref{eq:soundness} forces
$\widehat X_i^{\Pi} = \,?$ almost surely on every
$\mathsf{Obs}_i^{(T)}$ for which the posterior support
$\operatorname{supp}\, \PR\{\Xstar_i \mid \mathsf{Obs}_i^{(T)}\}$
has cardinality $\geq 2$: any positive conditional probability of
outputting a value $b \in \{0, 1\}$ on such a transcript would,
under full-support priors $\Vbias_r \in (0, 1)$, give a positive
unconditional probability of $\widehat X_i^{\Pi} = b \neq \Xstar_i$
on the complementary value branch, contradicting soundness; consequently, for every
sound $\Pi$,
\begin{equation}
\label{eq:bayes-bound}
   \PR\!\left\{\widehat X_i^{\Pi} = \,?\,\right\}
   \;\ge\; \PR\!\left\{\widehat X_i^{\,\mathrm{snd}}
   = \,?\,\right\},
\end{equation}
where $\widehat X_i^{\,\mathrm{snd}}$ is the sound Bayes-optimal
estimator on $\mathsf{Obs}_i^{(T)}$: declare $b$ when the posterior
support is the singleton $\{b\}$ and declare $?$ otherwise.

\emph{Step 2 (tree representation).}  By the variable-rooted
locally-tree-like result \Cref{cor:tree-tv-node} (a corollary of
\Cref{lem:tree-tv}), the depth-$(2T+2)$ neighborhood of $i$
converges in total variation, as $n \to \infty$, to the typed
Galton-Watson tree $\widetilde{\mathcal{N}}^{r}_{2T+2}$ at rate
$C_T / n + \Delta_{n, T}$, where $r$ is the role of $i$ and the
offspring law at the root is the full node degree law
$D^{(r),\mathrm{node}}$.
Total-variation convergence implies the corresponding convergence
of expected per-variable Bayes risk, with an
$O(C_T / n + \Delta_{n, T})$ correction; since both terms vanish as
$n \to \infty$, the correction vanishes in the
$\liminf_{n \to \infty}$ and the converse conclusion is unaffected.

\emph{Step 3 (support-BP equals logical forcing on the tree).}  On
the typed Galton-Watson tree $\widetilde{\mathcal{N}}^{r}_{2T+2}$,
exact posterior-support propagation under the channel-gated
erasure-only observation model satisfies the recursion
\begin{equation}
\label{eq:support-recursion}
   S_{a \to i}
   = \begin{cases}
   \Uset, & \Zobs_a = *\;\text{or}\;B_{a \to i} = 0, \\
   \force_{\theta, j}\!\big(\Zobs_a;\, (\widetilde S_{k \to a})_{k
   \in \partial a \setminus i}\big), & \text{otherwise},
   \end{cases}
   \qquad
   S_{i \to a} = \Aobs_i \cap
   \!\!\bigcap_{c \in \partial i \setminus a}\!\!
   S_{c \to i},
\end{equation}
where the channel-gated variable-to-check input is
\begin{equation}
\label{eq:gated-input}
   \widetilde S_{k \to a}
   \defeq
   \begin{cases}
   S_{k \to a}, & B_{k \to a} = 1, \\
   \Uset, & B_{k \to a} = 0.
   \end{cases}
\end{equation}
The forward gating mirrors the channel-gated observable transcript
$\mathsf{Obs}_i^{(T)}$ in \eqref{eq:obs-transcript}: incoming
messages on erased forward channels deliver $\Uset$ regardless of
the sender's support, matching the way the achievability decoder of
\Cref{thm:de} applies the forcing operator to the received messages
$\widetilde V_{k \to a}$ in \eqref{eq:check-update}.
With terminal node-support $S_i = \Aobs_i \cap \bigcap_{a \in
\partial i} S_{a \to i}$, this is a fact about the posterior
support of $\Xstar_i$ given $\mathsf{Obs}_i^{(T)}$ on a tree, not
about any specific protocol's message alphabet: the branches of the
tree rooted at distinct neighbors of any node are conditionally
independent given that node's hidden value, because the role-typed
Galton-Watson construction draws subtree degrees, observations, and
channel gates independently across branches conditional on the root's
role and value.  The posterior support therefore factorizes exactly
as the logical-forcing recursion of \Cref{sec:model-forcing}.  Under full-support priors
$\Vbias_r \in (0, 1)$, $S_i = \{\Xstar_i\}$ iff some chain of
singleton evidence forces $\Xstar_i$ from the tree's leaves to the
root, and $S_i = \Uset$ otherwise.  The terminal node estimate
$\widehat M_i^{(T)} = \Aobs_i \cap \bigcap_{a \in \partial i}
\widetilde C_{a \to i}^{(T)}$ produced by the value-conditioned
logical-forcing decoder of \Cref{thm:de} computes precisely $S_i$
in $T$ rounds; therefore
\begin{equation}
\label{eq:bayes-equals-de}
   \PR\!\left\{\widehat X_i^{\,\mathrm{snd}} = \,?
   \;\Big|\; \widetilde{\mathcal{N}}^{r}_{2T+2}\right\}
   \;=\; \PR\!\left\{\widehat M_i^{(T)} = \Uset
   \;\Big|\; \widetilde{\mathcal{N}}^{r}_{2T+2}\right\}.
\end{equation}

\emph{Step 4 (averaging to $\Pdee^{(T)}$).}  Averaging the
right-hand side of \eqref{eq:bayes-equals-de} over
$\widetilde{\mathcal{N}}^{r}_{2T+2}$, hidden values, observations,
and channel gates, with the root variable drawn uniformly under the
full node degree law $D^{(r), \mathrm{node}}$ as in
\eqref{eq:terminal-de}, gives the per-variable terminal
unresolved fraction $\Pdee^{(T)}(\lambda)$ of the logical-forcing
decoder.

Combining \eqref{eq:bayes-bound}, the convergence in Step 2, the
support-BP identity \eqref{eq:bayes-equals-de}, and the averaging
in Step 4 yields \eqref{eq:local-map-converse}.
\end{proof}

\subsection{Operational reading}
\label{sec:converse-operational}

\Cref{thm:converse} certifies the value-conditioned logical-forcing
decoder as asymptotically optimal among $T$-round
\emph{sound} (certifying) local protocols on the role-typed
configuration ensemble in the erasure model: no non-extrinsic,
randomized, or soft-information sound local protocol can
asymptotically leave fewer variables unresolved at its terminal
output than $\Pdee^{(T)}$.  Combined
with \Cref{thm:de} and \Cref{thm:threshold}, this gives a clean
operational picture: an architect committed to certifiable outputs,
the natural posture for verifier-style agent systems, cannot
outperform the iterate of the value-conditioned density-evolution
map at the corresponding round budget.  Dropping soundness allows a
protocol to guess the likeliest prior value on undetermined transcripts and
attain a smaller failure rate than abstention; that regime is the
natural domain of a non-erasure (BSC/absorbing-set) converse and is
outside the scope of the present erasure theory.

A separate, stronger converse direction, a Fano-cut-set lower
bound that allows the per-edge alphabet to be unbounded and that
matches \Cref{thm:threshold}'s asymptotic threshold surface in the
joint limit $T, Y \to \infty$, requires book-keeping of role-typed
cut-set capacities and matching with the spectral-radius condition.
We sketch this direction in \Cref{sec:outlook} and pursue it in
follow-on work.

\rhead{Relation to decision-theoretic converses on multi-agent
DAGs.}  A complementary converse direction is developed by Ao, Gao,
and Simchi-Levi \cite{ao2026reliability}, who treat LLM-based
multi-agent planning as a delegated decision problem on a finite
acyclic decision network with finite-capacity language interfaces.
They prove that any delegated DAG is decision-theoretically
dominated by a centralized Bayes decision-maker observing the same
evidence and characterize the gap as an expected posterior
divergence, which reduces to conditional mutual information under
logarithmic loss.  The two converses are complementary rather than
overlapping: \cite{ao2026reliability} bounds the
\emph{decentralization loss} of an arbitrary delegated DAG against a
centralized oracle under general decision losses, while
\Cref{thm:converse} bounds the \emph{message-passing reach} of the
certifying-protocol class on the role-typed configuration ensemble
against the value-conditioned density-evolution decoder under the
abstention loss.  The mutual-information characterization in
\cite{ao2026reliability} is closer in spirit to the
Fano-cut-set direction sketched in \Cref{sec:outlook-fano} and is a
natural reference point for the unbounded-alphabet limit.

\section{Calibration of Reliability Parameters}\label{sec:calibration}

The framework is useful only if its parameters can be related to real
agent traces.  We give a calibration protocol that maps each model
parameter to an operational quantity logged or readily instrumented
in deployed multi-agent systems.

\begin{itemize}[leftmargin=2em]
\item \emph{Architecture statistics $\{\pi_r^V\}, \{\pi_\theta^C\},
\alpha, \{\mathsf{P}_D^{(r)}\}$.}  Read off directly from the
deployed system's role labels, template taxonomy, and incidence
graph: $\widehat\pi_r^V = \#\{i : \rolemap(i) = r\} / |\Vset|$,
$\widehat\pi_\theta^C = \#\{a : \theta(a) = \theta\} / |\Cset|$,
$\widehat\alpha = |\Cset| / |\Vset|$, and the empirical socket-count
distribution $\widehat{\mathsf P}_D^{(r)}(d)$.  Verify the empirical
socket-balance $\widehat\pi_r^V \cdot
\widehat{\E}[D_\tau^{(r)}] \approx
\widehat\alpha \widehat\pi_\theta^C$ (cf.\
\eqref{eq:socket-balance-opt}) before running density evolution.
\item \emph{Variable-side erasure $\epsV[r]$.}  Estimate as the rate
at which role-$r$ variable agents abstain, time out, or produce
artifacts that cannot be used as local evidence on a held-out subset
of tasks.  In Hilbert this is the per-step type-checker abstention
rate; in CodeR it is the rate at which a code-fragment proposer fails
to commit; in a sensor network it is the per-modality
``signal-not-confident'' rate.
\item \emph{Verifier-side erasure $\epsC[s]$.}  Estimate as the rate
of \emph{no-verdict} events (verifier timeouts, sandbox crashes,
missing dependencies, malformed invocations, or ``cannot
determine'' / abstention outputs) among role-$s$ verifier calls
attempted.  Crucially, distinguish a no-verdict event ($\Zobs_a =
*$) from a definite negative verdict ($\Zobs_a = 0$): a failing
unit test or a Lean kernel rejection that runs to completion is a
non-erased negative output, not an erasure.  Confusing the two
inflates $\widehat\epsC[s]$ and erases the AND
positive-versus-negative certificate asymmetry the framework relies
on.
\item \emph{Reasoning-channel fidelity $\eta_{r, s}$.}  Estimate by
a controlled artifact-usability test: sample artifacts whose
source-side status is known to be non-erased and correct, transmit
them from role $r$ to role $s$, and measure the fraction the
receiver can parse and map to the intended local input object.  This
isolates the channel: format/translation success, not downstream
verifier success, not receiver competence.  In Hilbert this is the
refactor success rate (proposer's draft format $\to$ Lean kernel
input format).  Mixing $\eta_{r,s}$ with downstream verifier success
double-counts $\epsC[s]$ and is a common pitfall.
\item \emph{Value priors $\{\Vbias_r\}$.}  Estimate from held-out
traces or from the empirical fraction of valid subclaims in a
role-specific decomposition.  The DE recursion samples through
sockets, so if higher-degree subclaims are systematically harder
(value-degree dependence), the socket-level prior $\widehat\beta_{r,
\tau} = \big(\sum_{i : \rolemap(i)=r} D_{i,\tau} \one\{\Xstar_i =
1\}\big) / \big(\sum_{i : \rolemap(i)=r} D_{i,\tau}\big)$ differs
from $\widehat\beta_r$ and the socket-level prior should be used (cf.\
\Cref{sec:model-observations}).  When the role-level prior suffices,
report the diagnostic $\widehat\beta_{r,\tau} \approx
\widehat\beta_r$ for all relevant $\tau$.
\end{itemize}

\rhead{Soundness diagnostics.}  The erasure-only theory assumes
that non-erased outputs are correct: $\Aobs_i \neq \Uset \Rightarrow
\Aobs_i = \{\Xstar_i\}$ and $\Zobs_a \neq * \Rightarrow \Zobs_a =
f_\theta(\Xstar_{\partial a})$.  Empirically check this with
audit-labelled subsets:
\begin{equation}
\label{eq:wrong-non-erased}
\begin{aligned}
   \widehat\delta_r^V &\defeq
   \frac{\#\{\text{non-erased role-}r\text{ variable outputs that are
   wrong}\}}
   {\#\{\text{non-erased role-}r\text{ variable outputs}\}},
   \\[6pt]
   \widehat\delta_s^C &\defeq
   \frac{\#\{\text{non-erased role-}s\text{ verifier outputs that
   are wrong}\}}
   {\#\{\text{non-erased role-}s\text{ verifier outputs}\}}.
\end{aligned}
\end{equation}
The erasure-only theory is appropriate only when
$\widehat\delta_r^V$ and $\widehat\delta_s^C$ are negligible
relative to the target reliability level.  When they are not, the
calibrated system belongs to the hybrid erasure-and-flip extension
(\Cref{sec:outlook-bsc}) rather than the present BEC-style theory.

\rhead{Empirical fault-injection benchmarks as calibration
inputs.}  Recent fault-injection benchmarks for LLM-based multi-agent
systems, notably MAS-FIRE \cite{jia2026masfire}, can supply empirical
inputs for the calibration protocol above.  MAS-FIRE defines a
taxonomy of fifteen fault types covering intra-agent cognitive errors
and inter-agent coordination failures, injected through prompt
modification, response rewriting, and message-routing manipulation.
Its fault categories can be mapped, as a modeling step, onto the
M1 (variable-side, $\epsV[r]$) / M2 (verifier-side, $\epsC[s]$) /
M3 (reasoning-channel, $\eta_{r,s}$) decomposition that drives the
certificate-stopping-set theorem (\Cref{thm:certstop}); this mapping
is a calibration layer rather than MAS-FIRE's own taxonomy.  Such
empirical grounding supports the view that the three erasure tiers
correspond to distinct failure modes that engineers can observe and
intervene on separately.

After calibration, one runs the density-evolution recursion
(\Cref{thm:de}) and compares the predicted residual $\Pdee^{(L)}$
with held-out executions of finite agent networks.  The same
calibrated model can then be used for design: change a role mixture,
add cross-verification checks, improve a communication channel, or
allocate more budget to a verifier role, and evaluate the predicted
reliability change \emph{before} running a large benchmark.

The boundaries within which this calibration is meaningful, and the
failure modes the erasure-only model does not capture, are discussed
in \Cref{sec:outlook-limitations,sec:outlook-bsc}.

\section{Numerical Validation and DE Illustrations}\label{sec:numerical}

This section reports finite-graph Monte-Carlo concentration tests of
\Cref{thm:de} in both the canonical XOR specialization
(\Cref{sec:numerical-xor}) and the value-conditioned AND
specialization (\Cref{sec:numerical-and}); a deterministic-graph
validation of \Cref{thm:de-deterministic} on fixed graphs of
controlled local-cycle density (\Cref{sec:numerical-deterministic});
and a deterministic DE illustration of the non-interchangeability of
the three reliability tiers (\Cref{sec:numerical-3tier}).  The
simulator and the recursions used in these experiments are available
from the authors upon request.

\rhead{Reproducibility conventions.}  All Monte-Carlo
experiments use multigraph configuration-model instances (no
simple-graph conditioning, per \Cref{sec:notation}); each trial
independently resamples the matching, the hidden vector, the
variable- and verifier-side erasures, and the persistent
directed-channel gates of \Cref{sec:model-reasoning-channel};
random seeds are deterministic functions of the trial index and
$(\epsV, n)$ for the XOR sweep (full seed sequences available from
the authors); finite-$n$ socket-balance enforcement
uses the standard rounding-and-matching step of
\Cref{sec:model-ensemble} with $O(\sqrt{n})$ rounding fluctuations.
For the XOR ensemble the all-zero hidden vector is used without
loss of generality (\Cref{cor:xor-de}).

\subsection{XOR DE vs.\ Monte-Carlo}\label{sec:numerical-xor}

\rhead{Setup.}
Variable role and check role are both singletons.  Templates are XOR
factors of arity $d_c$.  Variable-side erasure is $\epsV[r] = \epsV$
for the unique variable role; verifier-side erasure is $\epsC[s] =
\epsC$ for the unique check role; reasoning-channel fidelity is
$\eta_{r,s} = \eta$ for the unique role pair.  We pick $(d_v, d_c) =
(3, 6)$ (the textbook $(3, 6)$-regular LDPC-BEC ensemble, with
noiseless BP threshold ${\epsV}^{*} \approx 0.4294$, below the BEC
capacity threshold $0.5$ at rate~$1/2$), fix $\epsC = 0.05$ and
$\eta = 0.95$ to make the verifier-side and channel-side tiers
visibly active, and sweep the variable-side erasure $\epsV$ across
$\{0.20, 0.25, \ldots, 0.55\}$.  For each
$\epsV$, we generate $30$ independent draws of the role-typed
configuration model graph at each of three problem sizes
$n \in \{200,\, 1000,\, 5000\}$.  Each instance is decoded by
extrinsic edge-specific message passing for $L = 50$ iterations, and
the empirical bit-erasure rate $\Pbit^{(L)}(\Gn)$ is averaged across
the $30$ trials.  In parallel, we evaluate the deterministic DE
recursion of \Cref{cor:xor-de} for $L$ iterations to obtain
$\Pdee^{(L)}(\lambda)$.

\rhead{Result.}
\Cref{fig:numerical-validation} plots $\Pbit^{(L)}(\Gn)$ versus
$\epsV$.  The DE prediction (solid blue curve) and the empirical
points are in tight agreement for every $\epsV$ (max absolute
deviation across the swept grid and the three problem sizes is
$0.033$ at $n = 200$ and decreases monotonically with $n$).  The
empirical spread shrinks visibly with $n$: at $\epsV = 0.40$, the
standard deviation of the empirical bit-erasure rate is $0.075$ at
$n = 200$, $0.032$ at $n = 1000$, and $0.016$ at $n = 5000$,
consistent with the $1/\sqrt{n}$ McDiarmid-rate prediction of
\Cref{thm:de}.  The plotted error bars are one-standard-deviation
across graph/noise realizations and measure finite-instance
dispersion; the standard error of the displayed empirical mean is
smaller by a factor $1/\sqrt{30}$.

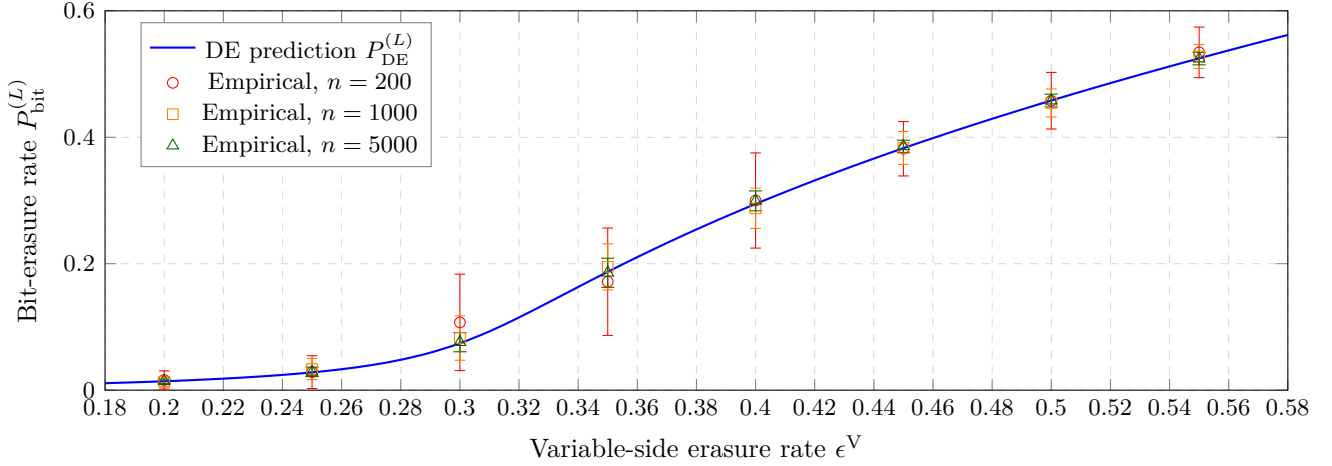
\begin{figure}[t]
\centering
\begin{tikzpicture}
\begin{axis}[
    width=0.95\linewidth, height=6.6cm,
    xlabel={Variable-side erasure rate $\epsV$},
    ylabel={Bit-erasure rate $\Pbit^{(L)}$},
    xmin=0.18, xmax=0.58,
    ymin=0, ymax=0.6,
    grid=major,
    grid style={dashed, gray!30},
    legend pos=north west,
    legend style={font=\footnotesize, fill=white, draw=black!50},
    tick label style={font=\footnotesize},
    label style={font=\small},
    cycle list name=color list,
]
\addplot[thick, blue, mark=none, smooth] table {
0.1800 0.010998
0.1900 0.012429
0.2000 0.014073
0.2100 0.015983
0.2200 0.018227
0.2300 0.020899
0.2400 0.024133
0.2500 0.028116
0.2600 0.033126
0.2700 0.039569
0.2800 0.048031
0.2900 0.059296
0.3000 0.074171
0.3100 0.092979
0.3200 0.115029
0.3300 0.138888
0.3400 0.163221
0.3500 0.187201
0.3600 0.210438
0.3700 0.232784
0.3800 0.254219
0.3900 0.274771
0.4000 0.294496
0.4100 0.313452
0.4200 0.331700
0.4300 0.349297
0.4400 0.366294
0.4500 0.382740
0.4600 0.398679
0.4700 0.414151
0.4800 0.429191
0.4900 0.443832
0.5000 0.458105
0.5100 0.472037
0.5200 0.485653
0.5300 0.498975
0.5400 0.512026
0.5500 0.524824
0.5600 0.537388
0.5700 0.549734
0.5800 0.561879
};
\addlegendentry{DE prediction $\Pdee^{(L)}$};
\addplot+[only marks, mark=o, mark size=2pt, color=red,
          error bars/.cd, y dir=both, y explicit] table[y error index=2] {
0.2000 0.016167 0.014358
0.2500 0.028500 0.025986
0.3000 0.107167 0.076275
0.3500 0.171500 0.084923
0.4000 0.300000 0.075299
0.4500 0.381833 0.043040
0.5000 0.457833 0.044753
0.5500 0.534333 0.039974
};
\addlegendentry{Empirical, $n = 200$};
\addplot+[only marks, mark=square, mark size=2pt, color=orange,
          error bars/.cd, y dir=both, y explicit] table[y error index=2] {
0.2000 0.013300 0.007412
0.2500 0.033433 0.016649
0.3000 0.082467 0.035219
0.3500 0.194967 0.036397
0.4000 0.287600 0.031869
0.4500 0.383100 0.026072
0.5000 0.454300 0.022181
0.5500 0.527700 0.018921
};
\addlegendentry{Empirical, $n = 1000$};
\addplot+[only marks, mark=triangle, mark size=2.5pt,
          color=black!60!green,
          error bars/.cd, y dir=both, y explicit] table[y error index=2] {
0.2000 0.014133 0.003359
0.2500 0.028907 0.007327
0.3000 0.075947 0.015239
0.3500 0.185707 0.023040
0.4000 0.299373 0.015778
0.4500 0.385673 0.009698
0.5000 0.458047 0.010157
0.5500 0.524467 0.009933
};
\addlegendentry{Empirical, $n = 5000$};
\end{axis}
\end{tikzpicture}
\caption{Density-evolution prediction (\Cref{thm:de},
\Cref{cor:xor-de}) vs.\ empirical bit-erasure rate from Monte-Carlo
simulation of the role-typed configuration ensemble.  Single-role
$(d_v, d_c) = (3, 6)$-regular XOR ensemble with three-tier erasure
$(\epsV, \epsC, \eta) = (\text{swept}, 0.05, 0.95)$ and $L = 50$
message-passing rounds.  Empirical points are averages over $30$
independent graph realizations at each of three problem sizes
$n \in \{200, 1000, 5000\}$, with one-standard-deviation error bars.
The empirical means track the DE prediction across the full $\epsV$
range, and the empirical spread shrinks with $n$ at the McDiarmid
rate $O(1/\sqrt{n})$ predicted by \Cref{thm:de}: the figure
validates the finite-$n$ accuracy of the DE recursion under the
model assumptions of the theorem.}
\label{fig:numerical-validation}
\end{figure}

\rhead{Verifier-side and channel-side tiers are active.}
At $\epsC = 0$, $\eta = 1$ (the noiseless-LDPC-BEC special case of
\Cref{cor:xor-de}), the threshold of the $(3, 6)$-regular ensemble
is the textbook value ${\epsV}^{*} \approx 0.4294$, at which
$\Pdee^{(\infty)}$ jumps from $0$ to a positive limit.  In the
presence of a verifier-side erasure $\epsC = 0.05$ and a channel
loss rate $1 - \eta = 0.05$, the threshold shifts left and the
above-threshold residual rises.  The next two subsections quantify
the two ways in which this rise is more than a single-scalar effect.

\subsection{AND value-conditioned DE: the certificate asymmetry}
\label{sec:numerical-and}

\rhead{Setup.}  We exercise the AND specialization
(\Cref{prop:and-de}) on a single-role $(d_v, d_c) = (3, 4)$-regular
AND ensemble with $\epsV = 0.20$, $\epsC = 0.05$, and directional
fidelities $\eta_{V \to C} = \eta_{C \to V} = 0.95$.  Value prior
$\Vbias = 0.7$ is realistic for verifier-style settings in which
most subclaims are valid.  We iterate the value-conditioned DE
recursion of \Cref{prop:and-de} for $L = 30$ rounds and plot
$\pmsg^{(1)}_\ell$ and $\pmsg^{(0)}_\ell$ separately.

\rhead{Result.}  \Cref{fig:and-history} shows that the two
value-conditioned curves diverge after one round and converge to
distinct fixed points.  The finite-graph Monte-Carlo overlay at
$n = 200$ ($40$ trials, $\pm 1\sigma$ error bars) lies within one
standard deviation of the DE prediction at every iteration,
consistent with \Cref{thm:de}.  The value-1 branch drops from $0.20$ to
$0.095$ in one round and stays there: a positive certificate
(verifier outputs $z = 1$) certifies all inputs at once, so a single
positive verdict on a check resolves every variable in its
neighborhood, and a positive verdict happens with probability
$\Vbias^{d_c - 1} = 0.7^3 \approx 0.343$ per neighbor, independent
of message-passing progress.  The value-0 branch drops more slowly:
$0.20 \to 0.149 \to 0.129$, then plateaus.  A negative certificate
(verifier outputs $z = 0$) certifies a target only when every
other input is already known to be $1$, which requires those
other variables to be resolved first, hence depends on
$\pmsg^{(1)}$.  The asymptotic gap $\pmsg^{(0)}_\infty -
\pmsg^{(1)}_\infty \approx 0.034$ visualizes the certificate
asymmetry of \Cref{prop:and-de} and the practical reading of
\Cref{rem:and-practical}.

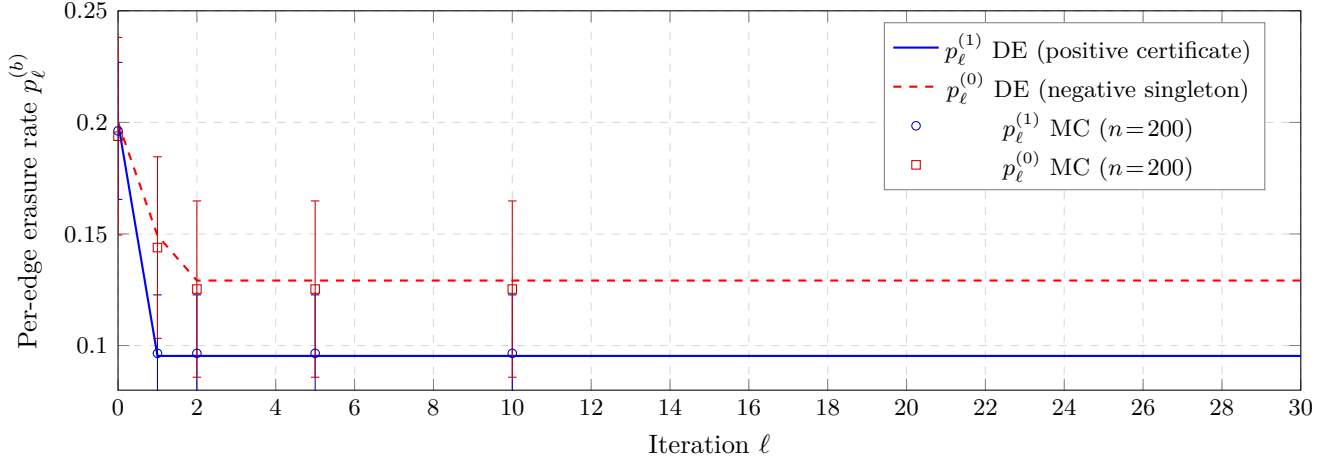
\begin{figure}[t]
\centering
\begin{tikzpicture}
\begin{axis}[
    width=0.95\linewidth, height=6.6cm,
    xlabel={Iteration $\ell$},
    ylabel={Per-edge erasure rate $\pmsg^{(b)}_\ell$},
    xmin=0, xmax=30,
    ymin=0.08, ymax=0.25,
    grid=major,
    grid style={dashed, gray!30},
    legend pos=north east,
    legend style={font=\footnotesize, fill=white, draw=black!50},
    tick label style={font=\footnotesize},
    label style={font=\small},
]
\addplot[thick, blue, mark=none] table {
0 0.200000
1 0.095342
2 0.095342
5 0.095342
10 0.095342
20 0.095342
30 0.095342
};
\addlegendentry{$\pmsg^{(1)}_\ell$ DE (positive certificate)};
\addplot[thick, red, dashed, mark=none] table {
0 0.200000
1 0.149338
2 0.129122
5 0.129122
10 0.129122
20 0.129122
30 0.129122
};
\addlegendentry{$\pmsg^{(0)}_\ell$ DE (negative singleton)};
\addplot[only marks, mark=o, mark size=1.6pt, color=blue!80!black,
         error bars/.cd, y dir=both, y explicit] table[y error index=2] {
0 0.1962 0.0307
1 0.0965 0.0262
2 0.0965 0.0262
5 0.0965 0.0262
10 0.0965 0.0262
};
\addlegendentry{$\pmsg^{(1)}_\ell$ MC ($n\!=\!200$)};
\addplot[only marks, mark=square, mark size=1.6pt, color=red!80!black,
         error bars/.cd, y dir=both, y explicit] table[y error index=2] {
0 0.1938 0.0443
1 0.1439 0.0407
2 0.1253 0.0395
5 0.1253 0.0395
10 0.1253 0.0395
};
\addlegendentry{$\pmsg^{(0)}_\ell$ MC ($n\!=\!200$)};
\end{axis}
\end{tikzpicture}
\caption{AND value-conditioned DE recursion (\Cref{prop:and-de}) on
a single-role $(d_v, d_c) = (3, 4)$-regular AND ensemble with
$\epsV = 0.20$, $\epsC = 0.05$, $\eta_{V \to C} = \eta_{C \to V} =
0.95$, and value prior $\Vbias = 0.7$.  Solid blue and dashed red
lines are the DE prediction; circle and square markers ($\pm 1\sigma$
error bars) are empirical Monte-Carlo means over $40$ trials at
$n = 200$.  The asymptotic gap visualizes the positive-versus-negative
certificate asymmetry of \Cref{prop:and-de}.}
\label{fig:and-history}
\end{figure}

\rhead{Concentration across the sweep.}  Beyond the single operating
point of \Cref{fig:and-history}, \Cref{fig:and-concentration} sweeps
the variable-side erasure $\epsV$ and overlays the value-conditioned
terminal bit-erasure rates $P^{(0)}$ and $P^{(1)}$ against direct
Monte-Carlo at three problem sizes.  The empirical rates concentrate
onto the DE prediction of \Cref{prop:and-de} (maximum absolute
deviation $\le 0.003$ at $n = 5000$), with the spread shrinking in $n$
at the $1/\sqrt{n}$ rate of \Cref{thm:de}, and the certificate
asymmetry $P^{(0)} > P^{(1)}$ persists across the whole sweep.  This is
the AND counterpart of the XOR concentration test of
\Cref{fig:numerical-validation}, on a non-value-symmetric verifier
factor that has no LDPC analog.

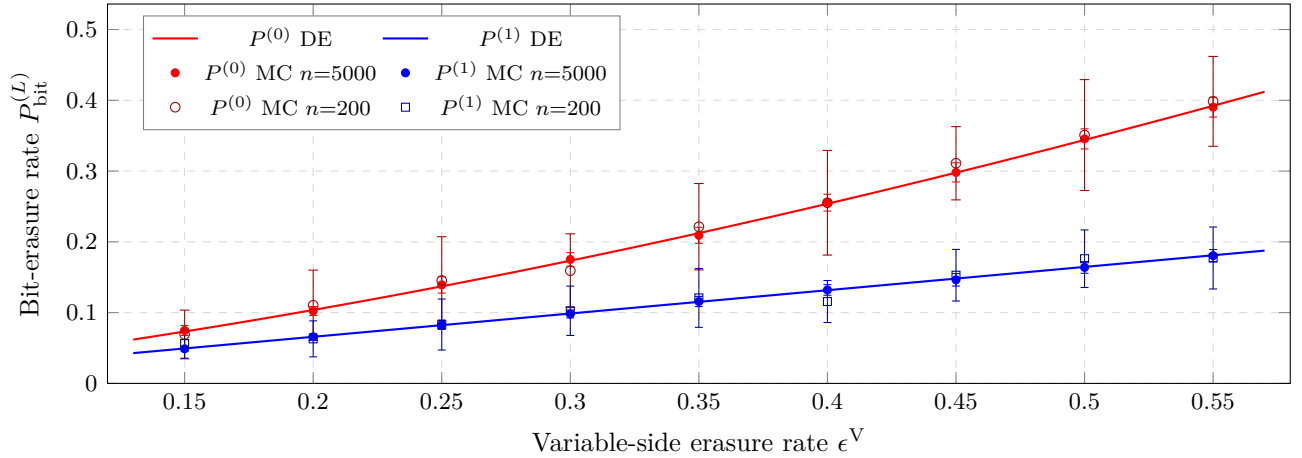
\begin{figure}[t]
\centering
\begin{tikzpicture}
\begin{axis}[
    width=0.95\linewidth, height=6.6cm,
    xlabel={Variable-side erasure rate $\epsV$},
    ylabel={Bit-erasure rate $\Pbit^{(L)}$},
    xmin=0.12, xmax=0.58, ymin=0, ymax=0.536,
    grid=major, grid style={dashed, gray!30},
    legend pos=north west,
    legend style={font=\scriptsize, fill=white, draw=black!50},
    legend columns=2,
    tick label style={font=\footnotesize}, label style={font=\small},
]
\addplot[thick, red, mark=none] table {
0.130000 0.061916
0.140000 0.067533
0.150000 0.073271
0.160000 0.079129
0.170000 0.085106
0.180000 0.091203
0.190000 0.097417
0.200000 0.103749
0.210000 0.110199
0.220000 0.116764
0.230000 0.123445
0.240000 0.130240
0.250000 0.137150
0.260000 0.144172
0.270000 0.151307
0.280000 0.158553
0.290000 0.165909
0.300000 0.173375
0.310000 0.180948
0.320000 0.188630
0.330000 0.196417
0.340000 0.204309
0.350000 0.212306
0.360000 0.220406
0.370000 0.228607
0.380000 0.236909
0.390000 0.245310
0.400000 0.253810
0.410000 0.262406
0.420000 0.271098
0.430000 0.279885
0.440000 0.288764
0.450000 0.297736
0.460000 0.306798
0.470000 0.315949
0.480000 0.325187
0.490000 0.334513
0.500000 0.343923
0.510000 0.353416
0.520000 0.362992
0.530000 0.372648
0.540000 0.382384
0.550000 0.392198
0.560000 0.402088
0.570000 0.412053
};
\addlegendentry{$P^{(0)}$ DE};
\addplot[thick, blue, mark=none] table {
0.130000 0.042788
0.140000 0.046080
0.150000 0.049371
0.160000 0.052663
0.170000 0.055954
0.180000 0.059245
0.190000 0.062537
0.200000 0.065828
0.210000 0.069120
0.220000 0.072411
0.230000 0.075703
0.240000 0.078994
0.250000 0.082285
0.260000 0.085577
0.270000 0.088868
0.280000 0.092160
0.290000 0.095451
0.300000 0.098742
0.310000 0.102034
0.320000 0.105325
0.330000 0.108617
0.340000 0.111908
0.350000 0.115199
0.360000 0.118491
0.370000 0.121782
0.380000 0.125074
0.390000 0.128365
0.400000 0.131657
0.410000 0.134948
0.420000 0.138239
0.430000 0.141531
0.440000 0.144822
0.450000 0.148114
0.460000 0.151405
0.470000 0.154696
0.480000 0.157988
0.490000 0.161279
0.500000 0.164571
0.510000 0.167862
0.520000 0.171154
0.530000 0.174445
0.540000 0.177736
0.550000 0.181028
0.560000 0.184319
0.570000 0.187611
};
\addlegendentry{$P^{(1)}$ DE};
\addplot[only marks, mark=*, mark size=1.5pt, color=red,
         error bars/.cd, y dir=both, y explicit] table[y error index=2] {
0.150000 0.074804 0.007040
0.200000 0.102290 0.006359
0.250000 0.139171 0.011623
0.300000 0.175112 0.009629
0.350000 0.209248 0.011337
0.400000 0.255410 0.011920
0.450000 0.298134 0.013527
0.500000 0.345507 0.014224
0.550000 0.390301 0.014042
};
\addlegendentry{$P^{(0)}$ MC $n{=}5000$};
\addplot[only marks, mark=*, mark size=1.5pt, color=blue,
         error bars/.cd, y dir=both, y explicit] table[y error index=2] {
0.150000 0.048874 0.003525
0.200000 0.065875 0.004615
0.250000 0.082926 0.006673
0.300000 0.098106 0.005643
0.350000 0.115718 0.007189
0.400000 0.132034 0.007467
0.450000 0.146394 0.008864
0.500000 0.163886 0.008269
0.550000 0.180540 0.008622
};
\addlegendentry{$P^{(1)}$ MC $n{=}5000$};
\addplot[only marks, mark=o, mark size=1.8pt, color=red!65!black,
         error bars/.cd, y dir=both, y explicit] table[y error index=2] {
0.150000 0.069162 0.034411
0.200000 0.110352 0.049665
0.250000 0.145412 0.061865
0.300000 0.159300 0.052061
0.350000 0.221382 0.061032
0.400000 0.255258 0.073942
0.450000 0.311052 0.051737
0.500000 0.350910 0.078325
0.550000 0.398515 0.063482
};
\addlegendentry{$P^{(0)}$ MC $n{=}200$};
\addplot[only marks, mark=square, mark size=1.5pt, color=blue!65!black,
         error bars/.cd, y dir=both, y explicit] table[y error index=2] {
0.150000 0.056240 0.020607
0.200000 0.062915 0.025427
0.250000 0.083171 0.036015
0.300000 0.102618 0.034875
0.350000 0.120794 0.041576
0.400000 0.115678 0.029679
0.450000 0.152846 0.036453
0.500000 0.176178 0.040676
0.550000 0.177190 0.043796
};
\addlegendentry{$P^{(1)}$ MC $n{=}200$};
\end{axis}
\end{tikzpicture}
\caption{AND Monte-Carlo concentration on the single-role $(d_v,d_c)=(3,4)$-regular
AND ensemble ($\epsC=0.05$, $\eta=0.95$, $\Vbias=0.7$, $L=30$).  Solid curves are the
value-conditioned DE prediction (\Cref{prop:and-de}); markers are empirical
Monte-Carlo means ($30$ trials, $\pm1\sigma$ bars) at $n=5000$ (filled) and $n=200$
(open).  The empirical rates concentrate onto the DE prediction (max
$|\mathrm{DE}-\mathrm{MC}|\le 0.003$ at $n=5000$), with spread shrinking in $n$, and
the gap $P^{(0)}>P^{(1)}$ is the positive-versus-negative certificate asymmetry of
\Cref{prop:and-de}, which has no value-symmetric XOR counterpart.}
\label{fig:and-concentration}
\end{figure}

\subsection{Deterministic-graph validation}
\label{sec:numerical-deterministic}

\rhead{Setup.}  \Cref{thm:de-deterministic} predicts that the same
bit-erasure DE prediction holds on a \emph{fixed} graph sequence, with
no appeal to graph randomness, once the local neighborhoods are
tree-like (vanishing local-cycle density).  We test this on the
$(3, 6)$-XOR ensemble at $\epsV = 0.25$ (below the noiseless threshold,
where stopping sets from short cycles are most visible), $\epsC =
0.05$, $\eta = 0.95$, $L = 50$, $n = 2400$.  Each graph is generated
once and then \emph{frozen}; only the hidden values and the three
erasure tiers are resampled across trials, so no matching exposure is
involved.  Graphs are built as disjoint unions of independently wired
blocks, so the block size is a clean knob for short-cycle density (a
single block of size $n$ is the ordinary configuration model).

\rhead{Result.}  \Cref{fig:det-errorfloor} plots the empirical
bit-erasure rate against the measured short-cycle density.  Locally
tree-like graphs reproduce the DE prediction $\Pdee^{(L)}$, while short
cycles raise the rate as an error floor, up to $4.7\times\,\Pdee^{(L)}$
at the densest; the operative hypothesis is thus local tree-likeness,
not graph randomness.  On a single frozen tree-like graph the
empirical rate concentrates onto $\Pdee^{(L)}$ as the graph grows, with
one-standard-deviation spread $0.026$, $0.010$, $0.005$ at $n = 600,
2400, 9600$ (the $1/\sqrt{n}$ rate), the deterministic-graph
counterpart of \Cref{fig:numerical-validation}.

\begin{figure}[t]
\centering
\begin{tikzpicture}
\begin{axis}[
    width=0.95\linewidth, height=6.2cm,
    xmode=log,
    xlabel={Short-cycle density (cycles per node)},
    ylabel={Bit-erasure rate $\Pbit^{(L)}$},
    xmin=0.008, xmax=2.6, ymin=0, ymax=0.158,
    grid=major, grid style={dashed, gray!30},
    legend pos=north west,
    legend style={font=\footnotesize, fill=white, draw=black!50},
    tick label style={font=\footnotesize}, label style={font=\small},
]
\addplot[black, dashed, thick, mark=none] coordinates {(0.008,0.028116) (2.6,0.028116)};
\addlegendentry{$\Pdee^{(L)}$ (locally tree-like)};
\addplot[thick, red!70!black, mark=*, mark size=1.8pt,
         error bars/.cd, y dir=both, y explicit] table[y error index=2] {
1.883889 0.131897 0.010531
1.090833 0.098911 0.012327
0.463889 0.069022 0.012990
0.255972 0.053369 0.011734
0.095278 0.038222 0.011036
0.050278 0.034047 0.010890
0.025833 0.030228 0.008175
0.012778 0.029694 0.009955
};
\addlegendentry{fixed graphs, MC $\Pbit^{(L)}$};
\end{axis}
\end{tikzpicture}
\caption{Deterministic-graph validation of \Cref{thm:de-deterministic} on the
$(3,6)$-XOR ensemble ($\epsV=0.25$, $\epsC=0.05$, $\eta=0.95$, $L=50$, $n=2400$).
Each point is a \emph{fixed} graph (only the noise is resampled; $3$ graphs $\times$
$50$ trials, $\pm1\sigma$ bars), built as a disjoint union of blocks so that block
size controls short-cycle density.  Locally tree-like graphs (left) sit on the DE
prediction $\Pdee^{(L)}$; as short cycles proliferate the empirical rate rises as an
error floor (up to $4.7\times\,\Pdee^{(L)}$), confirming that the operative
hypothesis is local tree-likeness, not graph randomness.  On a frozen tree-like
graph the empirical rate concentrates onto $\Pdee^{(L)}$ with spread
$0.026/0.010/0.005$ at $n=600/2400/9600$, the $1/\sqrt{n}$ rate.}
\label{fig:det-errorfloor}
\end{figure}
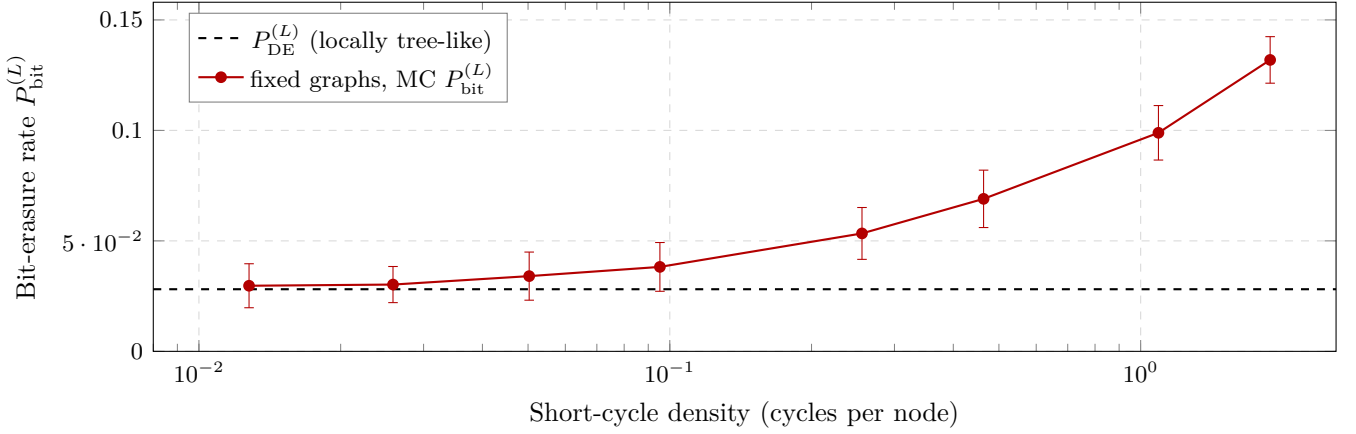

\subsection{A naive constant-sum effective-erasure heuristic fails}
\label{sec:numerical-3tier}

\rhead{Setup.}  We illustrate the operational content of
\Cref{prop:noninterchange} by ruling out the simplest one-scalar
candidate: the additive sum $\epsV + \epsC + (1 - \eta)$.  Fix the
$(3, 6)$-XOR ensemble with $\epsV = 0.40$ (just below the noiseless
threshold $\approx 0.4294$).  Fix a side budget $B_{\rm side} =
\epsC + (1 - \eta) = 0.10$.  Sweep the split $t \in [0, 1]$
between verifier and channel: $\epsC(t) = t \cdot B_{\rm side}$,
$1 - \eta(t) = (1 - t) \cdot B_{\rm side}$.  Every allocation has
the same naive effective-erasure sum $\epsV + \epsC + (1 - \eta) =
0.50$; if the three tiers were collapsible into the additive scalar,
the residual $\Pdee^{(L)}(\lambda(t))$ would be flat in $t$.

\rhead{Result.}  \Cref{fig:non-interchange} shows that the
residual varies from $0.332$ at $t = 0$ (heavy channel loss; $\epsC
= 0$, $1 - \eta = 0.10$) to $0.146$ at $t = 1$ (heavy verifier
erasure; $\epsC = 0.10$, $\eta = 1$): a $2.3\times$ ratio under the
same additive sum.  The sweep therefore numerically rejects the
additive effective-erasure scalar.  This is consistent with
the rank-$\ge 2$ content of \Cref{prop:noninterchange} but is not a
formal numerical proof of the rank statement: the proposition
forbids any smooth scalar reduction, not only the additive one.  A
direct check would evaluate the parameter Jacobian
$D_{\mathrm{par}}[\DEmap_\lambda(\zero)]$ and confirm that it has
rank at least two on this single-role slice (left to follow-on
work); the stronger rank-$\ge 3$ separation of the three tiers is
established on a multi-role ensemble in
\Cref{prop:noninterchange-roles}.

\rhead{Operational reading along this path.}  The sweep also
shows that, along this constant-sum path in the $(3,6)$-XOR
ensemble, channel loss is more damaging than verifier erasure.
Under a cost model in which the two improvements have comparable
marginal cost, this suggests improving the channel as the better
local intervention.  A true KKT shadow-price conclusion requires
the full optimization problem of \Cref{thm:optimization} with
explicit cost functions and active constraints; the constant-sum
sweep is an exploration along one specific path, not a substitute
for that calculation.

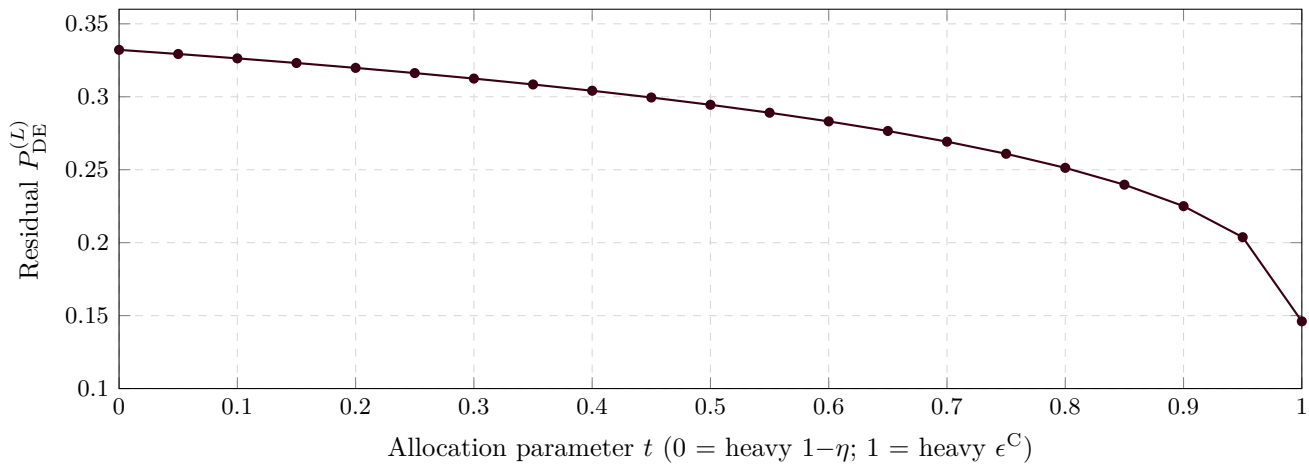
\begin{figure}[t]
\centering
\begin{tikzpicture}
\begin{axis}[
    width=0.95\linewidth, height=6.6cm,
    xlabel={Allocation parameter $t$ (0 = heavy $1{-}\eta$;
            1 = heavy $\epsC$)},
    ylabel={Residual $\Pdee^{(L)}$},
    xmin=0, xmax=1,
    ymin=0.10, ymax=0.36,
    grid=major,
    grid style={dashed, gray!30},
    tick label style={font=\footnotesize},
    label style={font=\small},
]
\addplot[thick, black!70!purple, mark=*, mark size=1.5pt] table {
0.0000 0.332155
0.0500 0.329310
0.1000 0.326311
0.1500 0.323145
0.2000 0.319795
0.2500 0.316243
0.3000 0.312467
0.3500 0.308439
0.4000 0.304129
0.4500 0.299497
0.5000 0.294496
0.5500 0.289063
0.6000 0.283119
0.6500 0.276558
0.7000 0.269231
0.7500 0.260920
0.8000 0.251284
0.8500 0.239730
0.9000 0.225042
0.9500 0.203748
1.0000 0.146077
};
\end{axis}
\end{tikzpicture}
\caption{XOR residual $\Pdee^{(L)}$ on a constant-additive-sum
surface: $\epsV = 0.40$ fixed, side budget $\epsC + (1 - \eta) =
0.10$ split as $\epsC = t \cdot 0.10$ and $1 - \eta = (1 - t) \cdot
0.10$.  $(d_v, d_c) = (3, 6)$, $L = 50$.  Residuals range over
$[0.146, 0.332]$ along the same additive sum.  The
two design points at the endpoints have the same additive
effective-erasure score but different residual erasure: the DE map
is sensitive to where erasures enter the message flow.  A flat
curve would be required if the three tiers collapsed into the
additive effective-erasure scalar; the observed monotone variation
numerically rules out that particular scalar.  The figure
illustrates the phenomenon and rules out the additive scalar; a
direct rank check of the parameter Jacobian would test the full
rank-$\ge 2$ content of \Cref{prop:noninterchange}; the stronger
rank-$\ge 3$ separation holds on a multi-role ensemble
(\Cref{prop:noninterchange-roles}).}
\label{fig:non-interchange}
\end{figure}

\section{Applications}\label{sec:applications}

We trace the framework through five representative deployed
agent systems.  Each can be seen as a calibration pilot for a
future companion paper.  \emph{Scope.}  The mappings below are
architectural correspondences, not empirical validation: they
illustrate how the variables, verifier nodes, and erasure tiers
could be instantiated in existing agent-system families.  Except in truly certifying layers (proof-kernel
checks, deterministic validators), the mapping is an idealized
abstraction and should be used only after calibration and
model-checking of the exact-or-erased assumption
(\Cref{sec:calibration}).  We do not claim that the present theory
causally explains any one benchmark result: deployed systems mix
verifier-driven certifying behavior with prompt-context effects,
aggregation effects, and confidently-wrong outputs not covered by
the erasure-only first-order theory.  What the framework does is
identify the graph-level mechanisms that can support reliable
recovery when tasks decompose into coupled subclaims and local
verifiers, and the operationally observable quantities that
calibrate the three erasure tiers.  A deployed system need not run
logical forcing for these limits to apply: because logical forcing is
the best sound local rule (\Cref{thm:converse}), a system that
certifies soundly can only match or fall short of the reliability the
theory predicts.

\subsection{Multi-agent formal theorem proving (Hilbert)}

Hilbert \cite{varambally2025hilbert} reaches $99.2\%$ on miniF2F via
a verifier-centered workflow orchestrating an informal-reasoning LLM,
a prover LLM, a Lean verifier, and a theorem retriever, with recursive
subgoal decomposition and verifier feedback.  The framework abstracts
this by mapping proof obligations to variable nodes, Lean kernel calls
to AND-monotone check nodes (the kernel accepts iff every step in its
scope type-checks together), and formatting or retrieval failures to
channel- or verifier-side erasures.  \Cref{tab:hilbert-mapping} and
\Cref{fig:hilbert-arch} give the role mapping; both are an idealized
factor-graph abstraction, not a faithful rendering of Hilbert's
control flow.  Candidate calibration: $\epsV[\mathrm{proposer}]$ from
per-step abstention rates, $\epsC[\mathrm{verifier}]$ from Lean-kernel
timeout and dependency-failure rates,
$\eta_{\mathrm{proposer}, \mathrm{verifier}}$ from refactor success
rates.  One workflow round corresponds to one extrinsic
message-passing iteration, and the round budget $T$ of \Cref{thm:de}
is the depth at which the verifier-feedback loop is truncated.

\begin{table}[t]
\centering
\renewcommand{\arraystretch}{1.25}
\caption{Mapping the Hilbert architecture
\cite{varambally2025hilbert} onto the role-typed Boolean-verifier-%
node framework of this paper.  Each row of the right column
identifies the operational quantity to estimate from agent traces.}
\label{tab:hilbert-mapping}
\small
\begin{tabular}{p{0.27\linewidth}p{0.30\linewidth}p{0.40\linewidth}}
\toprule
\textbf{Hilbert role} & \textbf{Maps to} & \textbf{Operational meaning} \\
\midrule
Decomposer / proposer
   & Variable agents $v_i$ of role ``proposer''
   & Drafts a candidate proof step.  Subclaim $\Xstar_i$ = ``does this
   step type-check?''  Abstention $\Aobs_i = \Uset$ rate is
   $\epsV[\mathrm{proposer}]$. \\
Refactorer
   & Reasoning-channel transformation
   & Reformulates a proposer-to-verifier message so the verifier can
   use it (notation, lemma inlining).  Failed refactor is an
   $\eta_{\mathrm{proposer}, \mathrm{verifier}}$ event. \\
Type-checker (Lean kernel)
   & Check agents $a_j$ of role ``verifier''
   & Runs the Lean kernel on a small group of proposed steps; constraint
   $C_a = X_{i_1} \wedge \cdots \wedge X_{i_d}$ (AND-monotone).
   Timeout / dependency failure is the $\epsC[\mathrm{verifier}]$
   event. \\
Aggregator
   & Final estimator
   & Takes terminal incoming messages at each variable and emits the
   per-step verdict (the variable-by-variable estimator of
   \Cref{eq:final-est}). \\
\bottomrule
\end{tabular}
\end{table}

\begin{figure}[t]
\centering
\begin{tikzpicture}[
  scale=0.95,
  every node/.style={font=\small},
  role/.style={draw, rounded corners=3pt, minimum width=2.0cm,
               minimum height=0.9cm, align=center, fill=blue!5},
  vnode/.style={circle, draw, minimum size=0.55cm, inner sep=0pt,
                fill=white},
  cnode/.style={draw, minimum size=0.5cm, inner sep=0pt, fill=white},
  arr/.style={-{Latex[length=2mm]}, thick},
  tier/.style={font=\footnotesize, color=red!60!black}
]
  \node[role] (D) at (0, 0)    {Decomposer\\\emph{(proposer)}};
  \node[role] (R) at (3, 0)    {Refactorer\\\emph{(channel)}};
  \node[role] (V) at (6, 0)    {Type-checker\\\emph{(Lean kernel)}};
  \node[role] (A) at (9, 0)    {Aggregator\\\emph{(estimator)}};
  \draw[arr] (D) -- (R);
  \draw[arr] (R) -- (V);
  \draw[arr] (V) -- (A);
  \node[tier, below=2pt of D] {$\epsV[\mathrm{prop}]$};
  \node[tier, below=2pt of R] {$\eta_{\mathrm{prop},\mathrm{ver}}$};
  \node[tier, below=2pt of V] {$\epsC[\mathrm{ver}]$, AND factor};
  \node[tier, below=2pt of A] {variable-wise MAP};
  \begin{scope}[shift={(0, -3.0)}]
    \node[anchor=west] at (-0.6, 0.6)
      {\emph{Underlying factor graph:}};
    \node[vnode] (v1) at (0.6, 0)   {$v_1$};
    \node[vnode] (v2) at (1.9, 0)   {$v_2$};
    \node[vnode] (v3) at (3.2, 0)   {$v_3$};
    \node[vnode] (v4) at (4.5, 0)   {$v_4$};
    \node[cnode] (a1) at (1.25, -1.4) {$a_1$};
    \node[cnode] (a2) at (3.85, -1.4) {$a_2$};
    \draw (v1) -- (a1); \draw (v2) -- (a1);
    \draw (v3) -- (a2); \draw (v4) -- (a2);
    \draw (v2) -- (a2); \draw (v3) -- (a1);
    \node[anchor=west, font=\footnotesize] at (5.4, 0)
      {variable agents = proposed steps,};
    \node[anchor=west, font=\footnotesize] at (5.4, -0.4)
      {private observation $\Aobs_i$;};
    \node[anchor=west, font=\footnotesize] at (5.4, -1.0)
      {check agents = Lean kernel calls,};
    \node[anchor=west, font=\footnotesize] at (5.4, -1.4)
      {AND template
       $C_a = \bigwedge_{i \in \partial a} X_i$;};
    \node[anchor=west, font=\footnotesize] at (5.4, -2.0)
      {check observation $\Zobs_a$.};
  \end{scope}
\end{tikzpicture}
\caption{The Hilbert architecture \cite{varambally2025hilbert} mapped
onto the role-typed Boolean-verifier-node framework.  \emph{Top:}
the four cooperating LLM roles, decomposer, refactorer,
type-checker, aggregator, each annotated with the erasure-tier
parameter it controls.  \emph{Bottom:} the underlying bipartite
factor graph that the framework analyzes.  Variable agents are
proposed proof steps; check agents are Lean kernel invocations
performing AND-monotone joint type-checks of step groups.  Refactor
failures between a proposer's output format and a verifier's
expected input format appear as reasoning-channel erasures
$1 - \eta_{\mathrm{prop}, \mathrm{ver}}$.}
\label{fig:hilbert-arch}
\end{figure}

\subsection{Multi-agent code generation (CodeR, SWE-bench)}

CodeR \cite{chen2024coder} reaches $28.33\%$ on SWE-bench-lite via a
multi-agent task graph in which generators write code, test runners
execute unit tests, and patch aggregators combine fragments.
Test-runner verdicts are AND-monotone Boolean factors at the
file/test-suite layer (the suite passes iff every test passes).  Here
$\Xstar_i$ should be defined as ``artifact $i$ satisfies local
validator $v$'' (a test outcome, type-check, or static-analysis
pass), not ``artifact $i$ is semantically correct,'' which a finite
test suite does not certify; the theory thus models recoverability of
local-validator outcomes, and full semantic correctness would need a
test suite complete for the claimed property.  SWE-bench
\cite{jimenez2024swebench}, SWE-agent \cite{yang2024sweagent}, CLEVER
\cite{thakur2025clever}, and FVAPPS \cite{dougherty2025proving} admit
the same reading.  Design questions such as ``add more test runners
or more code generators?'' map onto \Cref{thm:optimization}~(d)--(e),
and one CodeR cycle is one extrinsic DE iteration.

\subsection{Massive decomposition with cross-verification (MAKER)}

The MAKER framework \cite{meyerson2025maker} executes a
one-million-step LLM task with zero observed errors by combining
maximal decomposition of the task into single-step subtasks with a
first-to-ahead-by-$K$ voting rule on each subtask.  The deployment
is a coupled-subclaim instantiation at scale: decomposition lifts a
depth that would be catastrophic for an autoregressive single agent
into a regime where the depth-$L$ density-evolution recursion of
\Cref{thm:de} contracts to a near-zero error floor.  Per-subtask
voting supplies a cross-verification structure of the kind
\Cref{thm:augmentation} formalizes for separating augmentations.  The match
is structural rather than mechanism-level (MAKER's voting is
stochastic redundancy, not deterministic Boolean forcing), but it
indicates that the systems-engineering community has arrived
independently at the architectural principle the framework predicts.

\subsection{LLM debate on structured-output benchmarks}

Debate systems \cite{du2024debate,khan2024debate,kenton2024oversight}
alternate critiques between agents.  General LLM debate is not covered
by the erasure-only model, since debate messages are soft, persuasive,
and sometimes confidently wrong; it would need the soft-belief or
likelihood-ratio extension of \Cref{sec:outlook-bsc}.  The framework
applies as-is only to the narrow case of structured-output debate
where each round produces a deterministic local-validator outcome (a
parse check or constraint test) attached to a Boolean factor, in which
case each round is one extrinsic iteration.  The Choi--Zhu--Li
martingale analysis \cite{choi2025debate} is complementary: it studies
belief dynamics under soft exchange, whereas we study sound local
certification.

\subsection{Multi-modality classification by sensor and drone networks}

Sensor and drone networks motivate a non-erasure extension.  General
sensor classification involves soft observations, correlated noise,
false positives and negatives, and continuous measurements, none of
which are sound non-erased certificates.  A Boolean-erasure
approximation fits only when sensors produce abstaining local
certificates (validated detections or no-detection outputs after a
calibrated thresholding layer with deterministic consistency rules
between overlapping views); soft-classification and correlated-noise
regimes belong to the extensions of \Cref{sec:outlook}.  Under such a
thresholded-certificate calibration, role-typed reliabilities
correspond to per-platform sensing modalities and compute budgets and
$\eta$ to inter-platform communication loss, and the adjoint
machinery of \Cref{thm:optimization} could compare adding platforms
with improving links.  Without that calibration step the present
theorems should not be applied to soft-classification settings.

\section{Outlook}\label{sec:outlook}

\subsection{Limitations of the present model}
\label{sec:outlook-limitations}

The framework developed here is intentionally an erasure-only,
first-order reliability layer.  It does not model
confidently wrong messages (these would require the absorbing-set /
BSC machinery sketched below), correlated agent failures induced by
shared training data or shared inputs, adaptive graph construction
in which the protocol observes intermediate outcomes and reroutes
work to different roles, long-range memory or belief accumulation
across rounds beyond the once-at-$t=0$ erasure model, semantic
drift across iterations of an agent collective, or compounding
prompt-context failures specific to LLM-style verifiers.  These
omissions are not defects of the framework; they identify the
boundaries of what the first-order theory is responsible for and
the next layers of analysis (continuous-message DE, value-conditioned
absorbing sets, adaptive-routing converses, dependence-aware
ensembles).  The companion calibration claims in
\Cref{sec:calibration} should be read in this light: the parameters
$\epsV[r], \epsC[s], \eta_{r,s}$ correspond to operationally
observable failure modes and can be estimated from logs as
first-order signals, but the estimates themselves may be biased
by task difficulty (failure rates correlate with input distribution)
or by correlated failures (shared training data or shared inputs
violate the conditional-independence assumption underlying the
configuration ensemble).  The protocol is a starting point for
empirical instantiation, not a complete description of every reason
an agent system might be unreliable.

\subsection{Non-erasure channels: BSC absorbing sets and hybrid models}
\label{sec:outlook-bsc}

The erasure model is natural for abstention, timeout, dependency
failure, and unusable artifact, all empirically observable failure
modes.  Real systems also produce confident wrong messages.  The
hybrid erasure-and-flip extension introduces flip components
$\delta_r^V, \delta_s^C, \zeta_{r,s}$ alongside the erasure tiers,
with
\begin{equation*}
   \PR(\Aobs_i = \{\Xstar_i\}) = 1 - \epsV[r] - \delta_r^V,
   \quad
   \PR(\Aobs_i = \{1 - \Xstar_i\}) = \delta_r^V,
   \quad
   \PR(\Aobs_i = \Uset) = \epsV[r],
\end{equation*}
and analogous decompositions on the verifier and channel sides.
Once flips are admitted, $\Aobs_i \neq \Uset$ no longer implies
$\Aobs_i = \{\Xstar_i\}$ and $\Zobs_a \neq *$ no longer implies
$\Zobs_a = f_\theta(\Xstar_{\partial a})$, so the soundness lemma
(\Cref{lem:soundness}), the certificate-stopping-set theorem
(\Cref{thm:certstop}), the separating-augmentation theorem
(\Cref{thm:augmentation}), and the local-soundness converse
(\Cref{thm:converse}) do not apply in their present form.  The
DE state must track a full distribution over correct, incorrect,
and erased messages, not only erasure probabilities; the
finite-length obstruction combines stopping-set, absorbing-set
\cite{dolecek2010absorbing}, and trapping-set / pseudocodeword
phenomena, and is not a direct import of any one of these.  These
extensions are companion-paper material rather than
straightforward specializations of the present theory.

\subsection{Belief memory and dependence relaxation}
\label{sec:outlook-belief}

The current framework injects all noise once at $t = 0$.  A natural
extension introduces per-round noise injection with a role-specific
exponential moving-average update on a continuous belief state:
$\bm b_v^{(t)} = (1 - \alpha_{r_v})\, \bm b_v^{(t-1)} + \alpha_{r_v}\,
\bm e_{\xi_v^{(t)}}$, with $\alpha_{r_v} \in (0, 1]$ as a
temporal-smoothing / test-time-compute parameter (it is not a fourth
reliability tier alongside $\epsV, \epsC, \eta$, but a separate
axis governing how evidence accumulates across rounds).  This
matches LLM self-consistency aggregation
\cite{wang2023selfconsistency} and is the natural setting for
relaxing the conditional-independence assumption on agent failures.
The DE recursion would replace scalar erasure probabilities by
probability laws on the simplex $\Delta^{K-1}$, i.e.\ by a
continuous-state \emph{distributional} recursion
$\mathcal{L}(\bm b_v^{(t)} \mid \Xstar_v, \rolemap(v) = r)$.
Tractability may be retained under quantization, parametric
closure, or independence/latent-variable assumptions, but
correlated agent failures induced by shared training data, shared
inputs, or persistent belief states require new DE state variables
or a dependence-aware ensemble; tractability is not automatic.

\subsection{Richer Boolean-factor classes}

The abstract value-conditioned DE recursion (\Cref{thm:de}) and the
abstract certificate-stopping-set theorem (\Cref{thm:certstop})
already apply to every bounded-arity Boolean template through the
forcing operator $\force_{\theta,j}$.  What remains for OR, Horn,
implication, and other monotone or non-monotone primitives is to
derive compact closed-form forcing probabilities and interpretable
finite-length obstruction criteria analogous to
\Cref{cor:xor-stopping} and \Cref{cor:and-stopping}; the abstract
framework is in place, but the explicit corollaries that translate
the generic statements into reader-friendly criteria for each
primitive are future work.  Mixed-template ensembles (some checks
XOR, some AND, some Horn) are handled by the value-conditioned
framework with no further conceptual machinery.

\subsection{Density evolution on deterministic locally-tree-like
graphs}
\label{sec:outlook-deterministic}

\Cref{thm:de} concentrates the empirical bit-erasure rate on the
random configuration-model ensemble, via a Doob martingale over the
matching.  The same bit-erasure DE prediction holds on a
\emph{deterministic} graph sequence, with no appeal to graph
randomness, once the local
neighborhood statistics are controlled: at iteration $L$ a node's
estimate depends on its depth-$(2L{+}2)$ neighborhood, so the marginal
role-degree law alone does not suffice; what is needed is convergence
of the full local-neighborhood law.

\begin{theorem}[Deterministic-graph density-evolution concentration]
\label{thm:de-deterministic}
Let $\{G_n\}_{n \geq 1}$ be a deterministic sequence of role-typed
bipartite factor graphs, with variable degrees and check arities
uniformly bounded by $D_{\max}$, carrying role and template marks
$(\rolemap, \theta)$.  Draw the hidden values, the variable-side and
verifier-side erasure indicators, and the persistent reasoning-channel
gates mutually independently, with laws fixed by the node role, the
check role and template, and the directed role pair (the priors
$\{\Vbias_r\}$ of \eqref{eq:prior} and the erasure parameters
$\epsV[r], \epsC[s_\theta], \eta_{r,s}$), all independently of
$\{G_n\}$; the observations $\Aobs_i$ and $\Zobs_a$ are then the
derived quantities of \eqref{eq:var-obs} and \eqref{eq:check-obs}.
Suppose that for every fixed $L \in \N$:
\begin{enumerate}[label=\textup{(\alph*)},leftmargin=2.5em,itemsep=0pt]
  \item \textbf{Marginal role-degree typicality:} the empirical
    role-degree distribution of $G_n$ converges to the prescribed
    limit $\Lambda$;
  \item \textbf{Vanishing local-cycle density:} the fraction of
    variable nodes whose depth-$(2L{+}2)$ neighborhood contains a
    cycle vanishes as $n \to \infty$;
  \item \textbf{Marked Benjamini--Schramm convergence:} the empirical
    law of the rooted, socket-type-marked depth-$(2L{+}2)$
    neighborhood of a uniformly chosen variable node, each check
    half-edge carrying its socket type $\tau = (\theta, j)$ so that the
    ports of asymmetric templates are resolved, converges weakly
    to the law of the corresponding ball of the socket-type-marked
    unimodular Galton--Watson computation tree whose root and
    offspring degree, role, and socket-type marks match the laws
    $\{\pi_r^V\}, \{\pi_\theta^C\}, \Lambda$ of the DE recursion of
    \Cref{thm:de}.
\end{enumerate}
Then for every fixed $L$ the empirical bit-erasure rate
$\Pbit^{(L)}(G_n)$ converges in probability to $\Pdee^{(L)}$.
Quantitatively, with $b_n^{\mathrm{det}} \defeq
\big|\E\,\Pbit^{(L)}(G_n) - \Pdee^{(L)}\big| \to 0$, there is a
constant $a_L > 0$ depending only on $L$ and $D_{\max}$ such that for
every $t > 0$ and every $n$,
\begin{equation}
\label{eq:det-tail}
   \PR\big\{\big|\Pbit^{(L)}(G_n) - \Pdee^{(L)}\big|
   > b_n^{\mathrm{det}} + t\big\}
   \le 2\exp(-a_L\, n\, t^2).
\end{equation}
\end{theorem}

\rhead{Discussion.}  Hypotheses (a)--(c) are the socket-type-marked
Benjamini--Schramm convergence of $\{G_n\}$ to the socket-type-marked
unimodular Galton--Watson computation tree
\cite{aldous2004objective,bordenave2010resolvent}; under bounded
degree (c) is the operative condition and implies (a) (its degree
marginal) and (b) (a tree-supported limit forces vanishing cyclic
mass), which we list separately for readability.  Because the hidden
values are drawn from the priors independently of $G_n$, the
value-conditioning of the non-XOR recursion is automatic and no value
mark is needed in (c); the XOR specialization (\Cref{cor:xor-de}) is
value-symmetric in any case.  The proof (\Cref{app:de-deterministic})
has two parts: a bounded-differences concentration of $\Pbit^{(L)}$
around its mean on the \emph{fixed} graph (no matching exposure is
needed, since the graph is not random), and a bias estimate
$b_n^{\mathrm{det}} \to 0$ that holds because bounded degree makes the
depth-$(2L{+}2)$ socket-type-marked balls range over a finite set, on which (c) is
atom-wise convergence and the per-root unresolved probability is a
bounded function whose Galton--Watson average is exactly $\Pdee^{(L)}$
(\Cref{app:de-recursion}).  Like \Cref{thm:de}, the bound is centered
at $\Pdee^{(L)}$ only up to the deterministic offset $b_n^{\mathrm{det}}$,
which inherits the (here unspecified) rate of (a)--(c); a quantitative
neighborhood-law rate would make it explicit.  The theorem concentrates
the global bit-erasure rate $\Pbit^{(L)}$; the per-socket
value-conditioned message rates $\widehat\pmsg_{L, \tau}^{(b)}$ of
\eqref{eq:phat-conv} concentrate by the identical argument applied to a
uniformly chosen directed socket of type $\tau$ at radius $2L{+}1$,
which we do not spell out.

Biregular bipartite graphs realize the role-typed degree structure of
(a), and biregular families of large girth additionally satisfy the
locally-tree-like hypotheses (b)--(c).  Large girth is what (b)--(c)
require, and it is distinct from the spectral-expansion (Ramanujan)
property \cite{marcus2015interlacing}, which a biregular family may
have without large girth.  Because the
hidden values are drawn from the priors, no value-typicality of the
graph is required.  Explicit constructions for irregular role-typed
degree laws are open.

\Cref{thm:de-deterministic} complements the random-ensemble DE
concentration of the multi-edge-type LDPC \cite{richardson2008modern}
and noisy-message-passing \cite{tarighati2015noisy,dupraz2021noisy}
literatures, connecting the framework to deterministic local-weak
limits, and it applies to deployed systems whose task graph is fixed.
\Cref{thm:threshold}, \Cref{thm:certstop}, and \Cref{thm:augmentation}
are proved per-graph and apply directly to graphs in class (a)--(c);
the architecture-optimization theorem (\Cref{thm:optimization})
optimizes over distributional designs and remains a separate question.

\subsection{Fano-cut-set converse against unbounded-alphabet local
protocols}
\label{sec:outlook-fano}

\Cref{thm:converse} is a local-soundness converse: it bounds, via
the sound Bayes-optimal (support) estimator on the depth-$2T$
computation tree, the asymptotic per-variable terminal abstention
rate of any $T$-round sound local protocol on the role-typed
configuration
ensemble in the erasure model.  We pose the following open
problem: under a per-edge alphabet bound $|\mathcal{Y}| \le Y$ on
local messages, in the joint limit $T \to \infty$, $Y \to \infty$,
does an information-theoretic lower bound match the threshold
surface of \Cref{thm:threshold}~(b)?

The natural starting point is a role-pair-capacity cut-set bound
\cite{elgamalkim2011network}
on the transcript mutual information of the form
\[
   I\!\left(\Xstar;\, \mathrm{transcript}_T \,\Big|\, \{\Aobs_i\},
   \{\Zobs_a\}\right)
   \;\le\; \sum_{(r, s)} N_{r, s}\, T\, C_{r, s}(Y),
\]
where $N_{r, s}$ counts role-$r$-to-role-$s$ directed edges and
$C_{r,s}(Y)$ is the per-edge capacity of the persistent erasure
channel $\eta_{r,s}$ at alphabet bound $Y$ (private variable and
verifier observations are conditioned out as side information),
combined with Fano's inequality \cite{coverthomas2006elements}
on the per-variable estimator.
Whether such a role-typed cut-set bound matches the
spectral-radius threshold $\rho(\DEjac_\lambda(\zero))$ in the joint
$T, Y \to \infty$ limit is open; matching a global
mutual-information cut-set bound to a local BP-stability condition is
not automatic, and in general may hold only under
additional symmetry, regularity, or optimality assumptions on the
ensemble.  Repeated transmission over persistent open edges,
correlations introduced by the matching, and the order of limits
in $n, T, Y$ are technical complications that a future analysis
must address.  The expected-posterior-divergence /
conditional-mutual-information characterization for delegated
multi-agent DAGs in \cite{ao2026reliability} is a natural
reference point for the information-theoretic side of this
problem.

\section{Conclusion}\label{sec:conclusion}

This paper developed a density-evolution and finite-length
reliability theory for the certifying layer of sparse agent networks.
The central modeling step was to replace a single hidden answer with
a vector of coupled binary subclaims and to represent local
verification by role-typed Boolean verifier nodes.  Each component
carries out a distinct information-theoretic operation: variable
agents acquire partial evidence about the hidden subclaim vector,
verifier nodes process local constraints, reasoning channels transmit
certificates between roles, and the decoder combines the surviving
evidence into a per-subclaim certification or a safe abstention.  In
this representation, variable-side abstention, verifier-side failure,
and role-pair reasoning-channel loss are three structurally distinct
erasure mechanisms, not one scalar noise level
(\Cref{prop:noninterchange}).

For bounded-degree role-typed configuration ensembles, the
value-conditioned logical-forcing decoder admits a density-evolution
recursion with fixed-round concentration (\Cref{thm:de}).  The XOR
specialization recovers the LDPC-BEC baseline.  For non-linear factors
such as AND, the value-conditioned recursion reveals
positive-versus-negative certificate asymmetries absent from
parity-only models.  At finite length, failures are characterized
by certificate-stopping sets (\Cref{thm:certstop}), which gives a
direct route to separating augmentation (\Cref{thm:augmentation})
and to cost-constrained architecture optimization with adjoint
sensitivities and KKT shadow prices (\Cref{thm:optimization},
\Cref{rem:investment-variables}).  Within fixed-round sound local
protocols on the observable channel-gated computation tree, the
local-soundness converse (\Cref{thm:converse}) identifies logical
forcing as the asymptotically optimal certifying local rule.  The
bounded-difference concentration step and the adjoint/KKT framework
are channel-agnostic; this paper specializes them to the
certifying-layer logical-forcing recursion.  Numerical experiments
(\Cref{fig:numerical-validation}, \Cref{sec:numerical-and},
\Cref{sec:numerical-3tier}) illustrate the finite-$n$ concentration
predicted by \Cref{thm:de} and the value-conditioned and three-tier
asymmetries predicted by \Cref{prop:and-de} and
\Cref{prop:noninterchange}.

The theory is intentionally first-order.  Non-erased certificates are
assumed sound, and missing or unusable evidence is modeled as erasure.
Confidently wrong messages, correlated agent failures, adaptive graph
construction, and soft belief exchange lie outside the present erasure
theory.  Natural follow-on directions include hybrid erasure-and-flip
DE with absorbing-set obstructions (\Cref{sec:outlook-bsc}),
deterministic-graph DE under marked Benjamini--Schramm convergence
(\Cref{sec:outlook-deterministic}), dependence-aware ensembles for
correlated agent failures, finite-compute constraints, richer
Boolean-template analyses (OR, Horn, implication, mixed templates),
empirical calibration on formal-proof and verified-code benchmarks,
and stronger information-theoretic converses including the open
Fano-cut-set problem of \Cref{sec:outlook-fano}.  These directions
preserve the central viewpoint: certifying agent networks are sparse
systems of noisy local verification whose reliability can be studied
by the same asymptotic and finite-length tools that made sparse
graphical codes analyzable, provided the model assumptions are made
explicit and empirically checked.

\appendices
\crefalias{section}{appendix}
\crefalias{subsection}{appendix}
\section{A Useful Monotonicity Fact}\label{app:monotonicity}

The proofs above repeatedly use a monotonicity property of the
logical-forcing operator.

\begin{lemma}[Monotonicity of feasible sets]
\label{lem:gamma-monotone}
Fix a template $\theta$, target socket $j$, and verifier output
$z \in \{0, 1\}$.  If $M_k \subseteq M'_k$ for all $k \neq j$, then
\begin{equation}
\label{eq:gamma-monotone}
   \force_{\theta, j}(z;\, M_{-j})
   \subseteq
   \force_{\theta, j}(z;\, M'_{-j}).
\end{equation}
Moreover, under the \emph{sound-transcript} condition that
$z = f_\theta(\Xstar_{\partial a})$ and $\Xstar_k \in M_k$ for every
$k \neq j$, if $\force_{\theta, j}(z;\, M_{-j}) = \Uset$, then
$\force_{\theta, j}(z;\, M'_{-j}) = \Uset$ for every enlargement
$M_{-j} \subseteq M'_{-j}$.  In particular, on sound message
configurations, replacing singleton incoming messages by unresolved
messages cannot create a new singleton certificate.
\end{lemma}

\begin{proof}
The feasible set $\force_{\theta, j}$ is defined by existential
quantification over the Cartesian product of the incoming candidate
sets in \eqref{eq:forcing}.  Enlarging any candidate set only enlarges
the product over which the existential ranges, so the feasible set
can only grow, giving \eqref{eq:gamma-monotone}.  Without further
hypotheses, an enlargement may move $\force$ from $\emptyset$ to a
singleton (when the original input configuration was inconsistent);
the second clause therefore requires the sound-transcript invariant.
Under that invariant, $\Xstar_j \in \force(z; M_{-j})$ at every
input (the true assignment is always feasible by definition of $z$),
so $\force \in \{\{\Xstar_j\}, \Uset\}$ at every input.  Enlargement
preserves $\Uset$ as the largest set in $\{0, 1\}$, completing the
proof.
\end{proof}

\section{Detailed Proof of Theorem \ref{thm:de}}\label{app:de-proof}

This appendix gives the full proof of \Cref{thm:de} in three steps.
\Cref{app:de-tree} shows that the directed depth-$2L{+}1$ neighborhood
of a uniformly chosen socket converges in total variation to the
corresponding typed Galton-Watson tree.  \Cref{app:de-recursion}
identifies the message distribution at the root of the typed tree
with the value-conditioned recursion of
\eqref{eq:pbar}--\eqref{eq:pinit}.  \Cref{app:de-mcdiarmid} applies
McDiarmid's bounded-differences inequality to obtain the exponential
concentration bound \eqref{eq:bit-tail}.

\rhead{Radius convention.}  The proof distinguishes two
neighborhood radii: the directed edge-message
$V^{(L)}_{i \to a}$ is a function of the depth-$R_{\mathrm{edge}}(L) =
2L + 1$ directed neighborhood rooted at the edge, while the terminal
node estimate $\widehat M_i^{(L)} = \Aobs_i \cap \bigcap_{a \in
\partial i} \widetilde C_{a \to i}^{(L)}$ is a function of the
depth-$R_{\mathrm{node}}(L) = 2L + 2$ rooted variable neighborhood
(one extra hop is needed to read all incoming check messages at the
final round).  Per-socket message convergence uses
$R_{\mathrm{edge}}(L)$; the empirical bit-erasure rate
$\Pbit^{(L)}$ uses $R_{\mathrm{node}}(L)$.  All bounded-difference
constants below scale with the relevant radius.

\subsection{Computation neighborhoods and tree convergence}
\label{app:de-tree}

Fix $L \in \N$ and a socket type $\tau \in \Sockets$.  The directed
depth-$2L{+}1$ neighborhood of an edge $(i, a)$ is the subgraph
spanned by all directed paths of length at most $2L{+}1$ starting at
$(i, a)$ and alternating between variable and check nodes.  Because
all degrees are bounded by $D_{\max}$ and the template arities are
bounded by the same $D_{\max}$ (without loss of generality, by taking
the larger of the two), the number of nodes in such a neighborhood is
at most
\begin{equation}
\label{eq:nbhd-size}
   N_L \defeq 2 \sum_{r=0}^{2L+1} D_{\max}^r \;\le\; 4\, D_{\max}^{2L+1},
\end{equation}
and the neighborhood is independent of $n$.

\begin{lemma}[Typed tree convergence, quantitative form]
\label{lem:tree-tv}
Let $\mathcal{N}^{\tau}_{2L+1}(\Gn)$ denote the directed depth-%
$2L{+}1$ neighborhood of a uniformly chosen socket of type $\tau$ in
the role-typed configuration ensemble of \Cref{sec:model-ensemble},
and let $\widetilde{\mathcal{N}}^{\tau}_{2L+1}$ denote the corresponding typed
Galton-Watson tree truncated at depth $2L{+}1$ under the limiting
typed degree/socket law $\bm\pi^{\infty}$.  For each fixed $L$,
\begin{equation}
\label{eq:tv-rate}
   d_{\mathrm{TV}}\!\Bigl(
      \mathcal{L}\bigl(\mathcal{N}^{\tau}_{2L+1}(\Gn)\bigr),\;
      \mathcal{L}\bigl(\widetilde{\mathcal{N}}^{\tau}_{2L+1}\bigr)
   \Bigr)
   \;\le\; \frac{C_L}{n}\;+\;\Delta_{n, L},
\end{equation}
where $C_L = c_0\, N_L^2$ for an absolute constant $c_0$ depending
only on the role and template structure (not on $n$), and the
empirical-law discrepancy
\begin{equation}
\label{eq:Delta-n}
   \Delta_{n, L}
   \,\defeq\,
   d_{\mathrm{TV}}\!\Bigl(
      \mathcal{L}\bigl(\widetilde{\mathcal{N}}^{\tau}_{2L+1}\,
        ;\,\bm\pi_n\bigr),\;
      \mathcal{L}\bigl(\widetilde{\mathcal{N}}^{\tau}_{2L+1}\,
        ;\,\bm\pi^{\infty}\bigr)
   \Bigr)
\end{equation}
captures the difference between the typed Galton-Watson tree built
from the finite-$n$ empirical typed degree/socket law $\bm\pi_n$
(realized by the configuration model after the
$O(\sqrt n)$-rounded socket-balance step) and the same tree built
from the limiting law $\bm\pi^{\infty}$.  When the finite-$n$ ensemble
construction enforces deterministic socket-balance up to $O(1)$
counts, $\Delta_{n,L} = O(1/\sqrt n) \cdot \mathrm{poly}(N_L)$; in
the standard $O(\sqrt n)$-rounded construction, $\Delta_{n,L}
= O(1/\sqrt n) \cdot \mathrm{poly}(N_L)$ holds under any standard
$\bm\pi_n \to \bm\pi^\infty$ rate, and $\Delta_{n,L} \to 0$ as
$n \to \infty$ for any consistent estimator $\bm\pi_n$.
\end{lemma}

\begin{proof}
We assume positive limiting socket-type mass: $m_\tau \defeq \lim_n
n^{-1} \#\{\text{sockets of type }\tau\} > 0$ for every socket type
$\tau \in \Sockets$ that appears in the DE state, with
$m_{\min} \defeq \min_\tau m_\tau > 0$.  The TV bound
\eqref{eq:tv-rate} decomposes by triangle inequality into two
contributions:
\[
   d_{\mathrm{TV}}\!\Bigl(
      \mathcal{L}\bigl(\mathcal{N}^{\tau}_{2L+1}(\Gn)\bigr),\;
      \mathcal{L}\bigl(\widetilde{\mathcal{N}}^{\tau}_{2L+1};
        \bm\pi^{\infty}\bigr)
   \Bigr)
   \le
   \underbrace{d_{\mathrm{TV}}\!\Bigl(
      \mathcal{L}\bigl(\mathcal{N}^{\tau}_{2L+1}(\Gn)\bigr),\;
      \mathcal{L}\bigl(\widetilde{\mathcal{N}}^{\tau}_{2L+1};
        \bm\pi_n\bigr)
   \Bigr)}_{\text{collision error: } \le C_L/n}
   + \underbrace{\Delta_{n,L}}_{\text{empirical-law error}}.
\]
\emph{Collision error ($C_L/n$).}  We construct a coupling between
the configuration-model neighborhood and the typed Galton-Watson
neighborhood \emph{built from $\bm\pi_n$}, which fails to be valid
only when, while exploring the depth-$2L{+}1$ neighborhood
breadth-first, the configuration model proposes a socket pairing
already used in the partial exploration.  Each new pairing is
chosen uniformly from the remaining unpaired sockets of the
matching type; the number of remaining same-type sockets is
$\Theta(m_\tau\, n) \ge \Theta(m_{\min}\, n)$ and the number of
already-explored sockets is at most $N_L$.  The probability that
the next pairing collides with the partial exploration is
therefore at most $N_L / \Theta(m_{\min}\, n)$.  Union-bounding
over the $N_L$ pairings constructed during the breadth-first
exploration gives the collision-error bound $C_L/n$ with $c_0$
depending only on $m_{\min}^{-1}$ and the absolute constants
hidden in the $\Theta(\cdot)$.
\smallskip\par
\emph{Empirical-law error ($\Delta_{n,L}$).}  The Galton-Watson
tree under $\bm\pi_n$ uses finite-$n$ empirical socket-type
proportions and degree distributions; the tree under
$\bm\pi^{\infty}$ uses the limiting proportions.  By a standard
typed-tree TV stability argument (the tree law is a continuous
function of the degree law in the TV topology on truncated
neighborhoods), $\Delta_{n,L}$ is bounded by the TV distance
$d_{\mathrm{TV}}(\bm\pi_n, \bm\pi^{\infty})$ times a polynomial
factor in $N_L$.  Under the standard $O(\sqrt n)$-rounded
socket-balance construction $d_{\mathrm{TV}}(\bm\pi_n,
\bm\pi^{\infty}) = O(1/\sqrt n)$, so $\Delta_{n,L} =
O(\mathrm{poly}(N_L)/\sqrt n)$.  Combining the two terms yields
\eqref{eq:tv-rate}.
\end{proof}

\begin{corollary}[Variable-rooted typed tree convergence]
\label{cor:tree-tv-node}
Fix a role $r \in \RolesV$ and a depth $R \ge 1$.  Let
$\mathcal{N}^{r}_{R}(\Gn)$ denote the depth-$R$ rooted variable
neighborhood of a uniformly chosen variable of role $r$ in
the role-typed configuration ensemble of \Cref{sec:model-ensemble},
and let $\widetilde{\mathcal{N}}^{r}_{R}$ denote the typed
Galton-Watson tree truncated at depth $R$ with root offspring law
the full node degree law $D^{(r), \mathrm{node}}$ at the root and
the size-biased excess-degree laws at all subsequent levels.  Then
\begin{equation}
\label{eq:tv-rate-node}
   d_{\mathrm{TV}}\!\Bigl(
      \mathcal{L}\bigl(\mathcal{N}^{r}_{R}(\Gn)\bigr),\;
      \mathcal{L}\bigl(\widetilde{\mathcal{N}}^{r}_{R}\bigr)
   \Bigr)
   \;\le\; \frac{C_R}{n} + \Delta_{n, R},
\end{equation}
with the same collision constant $C_R$ (up to a factor bounded by the
maximum role-$r$ variable degree) and the same empirical-law term
$\Delta_{n, R}$ as in \Cref{lem:tree-tv}; in particular
$\Delta_{n, R} = 0$ under the deterministic socket-balance
construction, in which case the rate is the sharper $C_R/n$.
\end{corollary}

\begin{proof}
A uniformly chosen role-$r$ variable has the unbiased node-degree law
$D^{(r), \mathrm{node}}$; this is why the root offspring law differs
from the size-biased excess-degree law at deeper levels, which arises
from following a uniformly chosen socket rather than a uniformly
chosen node.  Conditional on the root degree $d \le D_{\max}$ and its
incident socket types, the depth-$R$ variable-rooted neighborhood is
the union of $d$ socket-rooted depth-$(R-1)$ neighborhoods.  Applying
\Cref{lem:tree-tv} to each and a union bound over the $d$ socket
explorations gives joint TV distance at most
$D_{\max} \cdot \bigl(C_{R-1}/n + \Delta_{n, R-1}\bigr)$; the
collision term absorbs into $C_R/n$, and the empirical-law term
absorbs into $\Delta_{n, R}$ (it is the same
$\bm\pi_n \to \bm\pi^{\infty}$ degree-law discrepancy as in
\Cref{lem:tree-tv}, up to a polynomial factor in the root degree).
\end{proof}

\subsection{The recursion on the typed tree}
\label{app:de-recursion}

We compute the message-erasure probability at the root of the typed
Galton-Watson tree given the root's hidden value, and verify that the
result matches \eqref{eq:pbar}--\eqref{eq:pinit}.

Let the root be a socket of type $\tau = (\theta, j)$ attached to a
variable $i$ of role $r = \socketrole{\tau}$ with hidden value $X_i =
b$.  The tree has $L$ alternating layers; we compute by induction on
$L$.

\emph{Base case $L = 0$.}  At iteration $0$ no messages have been
exchanged; the only message available at the root is $V_{i \to a}^{(0)}
= \Aobs_i$, which equals $\Uset$ with probability $\epsV[r]$ by
\eqref{eq:var-obs}.  Hence
$\pmsg_{0, \tau}^{(b)} = \epsV[r]$, matching \eqref{eq:pinit}.

\emph{Inductive step.}  Assume that after $\ell$ rounds, every socket
type $\tau' \in \Sockets$ and value $b' \in \{0, 1\}$ satisfy
\eqref{eq:p-def} with the value $\pmsg_{\ell, \tau'}^{(b')}$ given by
the recursion.  Consider iteration $\ell + 1$.  The variable $i$ at
the root of the tree sends $V_{i \to a}^{(\ell+1)}$ along socket $j$
to its parent check $a$ of template $\theta$.  By \eqref{eq:var-update},
\[
   V_{i \to a}^{(\ell+1)}
   = \Aobs_i \cap \bigcap_{c \in \partial i \setminus \{a\}}
   \widetilde C_{c \to i}^{(\ell)}.
\]
Conditional on $X_i = b$, the soundness lemma (\Cref{lem:soundness})
implies each incoming message lies in $\{\{b\}, \Uset\}$, so the
intersection equals $\{b\}$ as soon as any one of
$\{\Aobs_i\} \cup \{\widetilde C_{c \to i}^{(\ell)} : c \in \partial i
\setminus \{a\}\}$ is the singleton $\{b\}$, and equals $\Uset$
otherwise.  Equivalently,
\[
   V_{i \to a}^{(\ell+1)} = \Uset
   \;\Longleftrightarrow\;
   \Aobs_i = \Uset \;\text{ and every extrinsic inbound check
   message is }\Uset.
\]
On the typed tree, $\Aobs_i$ and the children $c$ of $i$ are
conditionally independent given the role types and incoming socket
type, so the product factors over excess-degree socket types $\tau'$:
\begin{equation}
\label{eq:V-equals-U}
   \PR\{V_{i \to a}^{(\ell+1)} = \Uset \mid X_i = b\}
   = \epsV[r]\, \E\!\left[
   \prod_{\tau' \in \Sockets_r}
   \big(\hmsg_{\ell, \tau'}^{(b)}\big)^{D_{\tau'}^{(r), \tau,
   \mathrm{ex}}}\right],
\end{equation}
which is \eqref{eq:pupdate}.  The complementary singleton probability
is $\PR\{V_{i \to a}^{(\ell+1)} = \{b\} \mid X_i = b\} = 1 - $
\eqref{eq:V-equals-U}.

It remains to verify the form of $\hmsg_{\ell, \tau'}^{(b)}$.  Fix one
of the root's incident excess-degree sockets $\tau' = (\theta', k) \in
\Sockets_r$, which connects the root variable $i$ to a descendant
check $a'$ of template $\theta'$ at socket $k$; since $i$ has role $r$,
the socket role is $\socketrole{\tau'} = r$.  By \eqref{eq:check-update},
the inbound message $\widetilde C_{a' \to i}^{(\ell)}$ equals $\Uset$
unless (i) the verifier observation $\Zobs_{a'}$ is not erased, (ii)
the return channel $a' \to i$ delivers, and (iii) the forcing operator
$\force_{\theta', k}$ at the target socket $k$ outputs the singleton at
value $b$.  Independence on the tree gives the product form
$(1 - \epsC[s_{\theta'}]) \cdot \eta_{s_{\theta'}, r}$ for (i)--(ii).
For (iii), conditional on the root value $X_i = b$ and on independent
draws of the descendant true values from the role priors $\Vbias$,
the inbound messages from the other sockets $k' \neq k$ of $a'$ are
singletons or $\Uset$ according to \eqref{eq:pbar}.  The probability
that the operator outputs the singleton at the target value $b$ is
exactly the forcing probability $\forceprob_{\ell, \theta', k}^{(b)}$
of \eqref{eq:forcing-prob}.  Composing,
\[
   \PR\{\widetilde C_{a' \to i}^{(\ell)} = \Uset \mid X_i = b\}
   = 1 - (1 - \epsC[s_{\theta'}])\, \eta_{s_{\theta'}, r}\,
   \forceprob_{\ell, \theta', k}^{(b)},
\]
which is \eqref{eq:hupdate}.  This closes the induction.

By \Cref{lem:tree-tv}, the message at the root of
$\mathcal{N}^{\tau}_{2L+1}(\Gn)$ has total-variation distance at most
$C_L / n$ from the message at the root of
$\widetilde{\mathcal{N}}^{\tau}_{2L+1}$, so
\begin{equation}
\label{eq:exp-conv}
   \E\big[\one\{V^{(L)} = \Uset, X = b\}\big]
   \,\xrightarrow[n \to \infty]{}\, \PR(X = b)\, \pmsg_{L, \tau}^{(b)},
\end{equation}
which is the convergence of expectations promised by \Cref{thm:de}.

\subsection{Bounded-differences concentration via sequential
exposure}
\label{app:de-mcdiarmid}

\rhead{Why this argument needs care.}  The role-typed
configuration model is the uniform measure on a finite set of socket
matchings, together with independent draws of all observation and
channel variables.  The observation and channel variables are
mutually independent and admit a direct McDiarmid bound, but the
socket-pairing matching $\Pi$ is \emph{not} a tuple of independent
coordinates: its $\binom{|\text{sockets}|}{2}$ pair indicators are
strongly correlated.  We therefore expose $\Pi$ sequentially through
a Doob martingale, a standard device for permutation- and
matching-valued random objects \cite{mcdiarmid1989}, and use
McDiarmid for the remaining independent ingredients.

\rhead{Random ingredients.}  Index the random ingredients of
one realization in canonical order.  Let $S$ be the set of all
sockets (variable-side and check-side); fix a canonical enumeration
$s_1, s_2, \ldots, s_{|S|}$ that proceeds, say, lexicographically.
We expose the matching $\Pi$ as a sequence of pair revelations
\begin{equation}
\label{eq:matching-exposure}
   E_1, E_2, \ldots, E_{|S|/2},
\end{equation}
where $E_t = (s_{i_t}, \pi(s_{i_t}))$ is the partner of the smallest
not-yet-matched socket $s_{i_t}$ at step $t$ (chosen uniformly among
the remaining unmatched sockets).  After the matching is exposed, we
expose the independent ingredient block:
\begin{itemize}[leftmargin=2em]
\item the hidden values $\Xstar_i$ for $i \in \Vset_n$ ($M_X = n$
independent draws under the role prior $\Vbias$);
\item the variable-side erasure indicators $\one\{\Aobs_i = \Uset\}$
($M_A = n$ Bernoulli draws with parameter $\epsV[r_i]$);
\item the verifier-side erasure indicators $\one\{\Zobs_a = *\}$ for
$a \in \Cset_n$ ($M_Z = \alpha n$ Bernoulli draws with parameter
$\epsC[s_a]$);
\item the channel-erasure indicators $\one\{\widetilde M = \Uset \mid
M = \{b\}\}$ for each directed edge ($M_\eta = \Theta(n D_{\max})$
Bernoulli draws with parameters $1 - \eta_{r, s}$).
\end{itemize}

\rhead{Doob martingale on the matching.}  Let $f(\Gn) =
\Pbit^{(L)}(\Gn)$ denote the empirical bit-erasure rate after $L$
rounds.  Define the Doob martingale
\begin{equation}
\label{eq:doob}
   Z_t \,\defeq\, \E\!\left[ f \,\Big|\, E_1, \ldots, E_t,\,
   \Xstar,\, \{\one\{\Aobs_i = \Uset\}\}_i,\,
   \{\one\{\Zobs_a = *\}\}_a,\, \{\text{channel indicators}\} \right],
\end{equation}
filtered by the matching-exposure and ingredient $\sigma$-algebras
in canonical order.  The conditioning is on the independent primitive
ingredients only, namely the hidden values, the variable- and
verifier-side erasure indicators, and the channel-erasure indicators.
The observations $\Aobs_i$ and $\Zobs_a$ themselves are \emph{not}
conditioned on as free coordinates: they are derived from these
primitives together with the matching ($\Zobs_a = f_\theta(\Xstar_{\partial a})$
when its erasure indicator does not fire, and $\Aobs_i \in
\{\{\Xstar_i\}, \Uset\}$ according to its variable-side indicator), so
exposing the matching updates them automatically and consistently.  We
show below that the martingale increment $|Z_t - Z_{t - 1}|$ is bounded
by a constant $\widetilde\Delta_L$ independent of the realization.

\begin{lemma}[Type-preserving single-pair switching coupling]
\label{lem:switching}
Fix a step $t \in \{1, \ldots, |S|/2\}$ in the matching-exposure
sequence \eqref{eq:matching-exposure}.  Let $\Pi, \Pi'$ be two
configuration-model perfect matchings agreeing on $E_1, \ldots,
E_{t-1}$ and differing at $E_t$ only in the partner of $s_{i_t}$.
Then there exists a coupling between $\Pi$ and $\Pi'$ under which
the completed matchings differ by exactly two paired edges:
\[
   \Pi \,=\, \Pi'_{0} \cup \{(s_{i_t}, u), (s', v)\},
   \qquad
   \Pi' \,=\, \Pi'_{0} \cup \{(s_{i_t}, v), (s', u)\},
\]
for some sockets $u, v$ of the matching type required by
$s_{i_t}$, some socket $s'$ outside $\{s_{i_t}, u, v\}$, and a
common matching $\Pi'_{0}$ on the remaining sockets.  In
particular, $\Pi$ and $\Pi'$ are related by a single
\emph{type-preserving 2-edge switch}.
\end{lemma}

\begin{proof}
The configuration-model perfect matching is the uniform measure on
type-compatible perfect matchings of $S$.  Conditional on
$E_1, \ldots, E_{t-1}$, the remaining matching is uniform over
type-compatible perfect matchings of the unexposed sockets.  Let
$u$ and $v$ be the round-$t$ partners of $s_{i_t}$ under $\Pi$ and
$\Pi'$ respectively, and let $s'$ be the round-$(>t)$ partner of $v$
under $\Pi$ (which is the partner of $u$ under $\Pi'$ by exchange).
Type-compatibility of $\Pi$ at the $(s_{i_t}, u)$ pair forces $u$ to
have the type required by $s_{i_t}$; similarly for $v$.  The
2-edge switch $(s_{i_t}, u), (s', v) \leftrightarrow (s_{i_t}, v),
(s', u)$ exchanges $\Pi$ and $\Pi'$ on these two pairs while
preserving all other pairs and all socket-type counts.  This is
the standard switching lemma for configuration-model matchings.
\end{proof}

\begin{lemma}[Local-statistic Lipschitz constant under switching]
\label{lem:switch-lipschitz}
Let $h: \mathcal{G}_R \to [0, B]$ be a bounded depth-$R$
rooted-variable local statistic, and let $g(\Gn) \defeq n^{-1}
\sum_{i \in \Vset_n} h\bigl(\mathcal{N}^{r_i}_R(\Gn, i)\bigr)$ be
its empirical mean.  Define
\begin{equation}
\label{eq:CDmaxR}
   C(D_{\max}, R) \,\defeq\, 4 \cdot 2\, D_{\max}^{R}
   \,=\, 8\, D_{\max}^{R},
\end{equation}
the maximum number of variables within graph distance $R$ of any
fixed set of four sockets in a graph with maximum degree
$D_{\max}$.  Then:
\begin{enumerate}[label=(\roman*),leftmargin=2em]
\item Under a type-preserving 2-edge switch (\Cref{lem:switching})
the empirical mean satisfies
$|g(\Gn) - g(\Gn')| \le B \cdot C(D_{\max}, R) / n$.
\item Replacing a single hidden value $\Xstar_j$, variable-side
erasure indicator $\one\{\Aobs_j = \Uset\}$, verifier-side erasure
indicator $\one\{\Zobs_a = *\}$, or channel-erasure indicator
changes $g$ by at most $B \cdot C(D_{\max}, R+1) / n$, where the
$(R+1)$-radius accounts for the propagation of a $\Xstar_j$ change
through the true verifier values $T_a = f_\theta(\Xstar_{\partial
a})$ at adjacent checks.
\end{enumerate}
\end{lemma}

\begin{proof}
\emph{(i)}  The 2-edge switch alters $\Gn$ only on the two pairs
$(s_{i_t}, u)$ and $(s', v)$, whose endpoints lie at four sockets
incident to at most four variables.  A variable $i$'s depth-$R$
rooted neighborhood includes the swapped edges only if $i$ is
within graph distance $R$ of one of these four sockets; the maximum
such count is $C(D_{\max}, R)$ by the bounded-degree volume bound.
For variables $i$ outside this set, $h(\mathcal{N}^{r_i}_R(\Gn, i))
= h(\mathcal{N}^{r_i}_R(\Gn', i))$.  Summing over the at-most
$C(D_{\max}, R)$ affected variables and dividing by $n$ gives
the bound.

\emph{(ii)}  A variable-side erasure-indicator, verifier-side
erasure-indicator, or channel-indicator change affects $g$ only
through variables $i$ whose depth-$R$
neighborhood contains the changed ingredient; the bound follows
as in (i) at depth $R$.  An $\Xstar_j$ change additionally flips
the verifier truth $T_a = f_\theta(\Xstar_{\partial a})$ at each
adjacent non-erased check $a$, and these $T_a$-changes propagate
through the depth-$R$ neighborhoods of variables one hop further
out from $j$.  This enlarges the radius from $R$ to $R+1$ and
yields the $C(D_{\max}, R+1)$ bound.
\end{proof}

\begin{lemma}[Bounded martingale increment]
\label{lem:bd-constant}
Let $\widetilde\Delta_L \defeq C(D_{\max}, R_{\mathrm{node}}(L)+1)
/ n = 8\, D_{\max}^{2L+3} / n$.  For every matching-exposure step
$t$ in \eqref{eq:matching-exposure}, the Doob martingale $Z_t$
defined in \eqref{eq:doob} satisfies $|Z_t - Z_{t-1}| \le
\widetilde\Delta_L$ almost surely.  For every independent
observation, hidden-value, or channel ingredient $W_j$, replacing
$W_j$ while holding all others fixed changes $f$ by at most
$\widetilde\Delta_L$ in absolute value.
\end{lemma}

\begin{proof}
The empirical bit-erasure rate $f(\Gn) = \Pbit^{(L)}(\Gn) =
n^{-1} \sum_i \one\{\widehat M_i^{(L)} = \Uset\}$ is the empirical
mean of a depth-$R_{\mathrm{node}}(L) = 2L+2$ rooted-variable local
statistic with $B = 1$.  By \Cref{lem:switch-lipschitz}(i),
adjacent matching realizations differing by a single type-preserving
2-edge switch (which is the form of any two completions admissible
in the Doob filtration at step $t$, by \Cref{lem:switching}) yield
$|f(\Gn) - f(\Gn')| \le C(D_{\max}, 2L+2)/n \le \widetilde\Delta_L$;
averaging over the swap-coupling gives $|Z_t - Z_{t-1}| \le
\widetilde\Delta_L$ almost surely.  By \Cref{lem:switch-lipschitz}(ii),
replacing any independent ingredient (including $\Xstar_j$, with the
$T_a$-propagation effect) changes $f$ by at most $C(D_{\max},
2L+3)/n = \widetilde\Delta_L$.
\end{proof}

\begin{proof}[Proof of \Cref{thm:de}]
By \Cref{lem:bd-constant}, $Z_t - Z_{t-1}$ is a bounded martingale
increment for the matching-exposure block, and the independent
observation and channel ingredients form a bounded-differences
sequence with the same constant $\widetilde\Delta_L$.  Combining the
Azuma--Hoeffding inequality on the matching martingale with
McDiarmid's bounded-differences inequality on the independent
ingredients \cite{mcdiarmid1989},
\begin{align*}
   \PR\big\{|f - \E f| > \delta\big\}
   &\le 2 \exp\!\left(-\frac{2 \delta^2}{N\,
   \widetilde\Delta_L^2}\right),
\end{align*}
where $N = |S|/2 + M_X + M_A + M_Z + M_\eta = O(n D_{\max})$ is the
total number of exposure steps.  Substituting $\widetilde\Delta_L
= 8\, D_{\max}^{2L+3} / n$ (\Cref{lem:bd-constant}),
\begin{align*}
   \PR\big\{|f - \E f| > \delta\big\}
   &\le 2 \exp\!\left(-\frac{2 \delta^2 n^2}
   {O(n D_{\max}) \cdot \big(8 D_{\max}^{2L+3}\big)^2}\right) \\
   &\le 2 \exp\!\left(-\frac{n\, \delta^2}{C\, D_{\max}^{4L+7}}\right)
\end{align*}
for a positive constant $C$ depending only on absolute constants in
the $O(\cdot)$ above; this is the around-mean concentration with rate
$a_L = 1 / (C\, D_{\max}^{4L+7})$, the displayed radius $\delta$
playing the role of $t$ in \eqref{eq:bit-tail}.  Combined with the
bias bound $|\E f - \Pdee^{(L)}| \le b_{n, L}$ (\Cref{lem:tree-tv},
which gives the limit \eqref{eq:exp-conv}), the triangle inequality
yields the offset tail \eqref{eq:bit-tail} and convergence of
$f = \Pbit^{(L)}(\Gn)$ in probability to $\Pdee^{(L)}$.

The edge-message convergence statement is identical with the
depth-$2L{+}1$ neighborhood and the excess-degree socket law in place
of the node-degree law; \Cref{lem:bd-constant} carries through with
the same constant up to absorbed factors, and the Azuma--McDiarmid
tail goes through unchanged.
\end{proof}

\rhead{Standing assumptions used by \Cref{thm:de}.}  For the
record, the precise standing assumptions of \Cref{thm:de} are: (i)
finite role and template sets ($|\RolesV|, |\RolesC|, |\Templates|
< \infty$); (ii) finite second moment of the role-typed degree
distribution and a uniform-in-$n$ degree bound $D_{\max} < \infty$;
(iii) socket-balance equalities holding asymptotically with
$O(\sqrt{n})$ rounding fluctuations; (iv) hidden values drawn
independently from a role-typed prior with $\Vbias_r \in (0, 1)$
under value-degree independence
$\Xstar_i \perp\!\!\!\perp D_i \mid \rolemap(i)$ (or,
equivalently, deterministic hidden vectors whose
\emph{socket-conditional} empirical value frequencies converge to
$\Vbias_r$); (v) directed-edge erasure variables independent across
edges and persistent in time, so
that any per-round erasure is chosen at $t = 0$ and held thereafter
(the alternative i.i.d.-per-round model is a separate setting and not
analyzed here).

\subsection{Deterministic-graph concentration}
\label{app:de-deterministic}

This appendix proves \Cref{thm:de-deterministic}.  Fix $L$, write
$R = 2L + 2$, and let $U_i \in \{0, 1\}$ indicate that variable node
$i$ is unresolved after $L$ rounds, so that
$\Pbit^{(L)}(G_n) = \frac1n \sum_i U_i$ summed over the $n$ variable
nodes.  The graph $G_n$ is deterministic; the independent primitive
random variables are the hidden values $\{\Xstar_i\}$, the
variable-side erasure indicators, the verifier-side erasure
indicators, and the persistent channel gates $\{B_e\}$, mutually
independent (with role-, template-, and role-pair-dependent laws) and
independent of $G_n$.  The observations $\Aobs_i$ of \eqref{eq:var-obs}
and $\Zobs_a$ of \eqref{eq:check-obs} are \emph{derived} from the
hidden values and these erasure indicators (e.g.\ $\Zobs_a = *$ when
its erasure indicator fires and $\Zobs_a = f_\theta(\Xstar_{\partial a})$
otherwise), so they are not themselves among the independent
ingredients.

\emph{Part A: concentration around the mean.}  $\Pbit^{(L)}(G_n)$ is a
function of the independent ingredients above (hidden values,
variable- and verifier-side erasure indicators, and channel gates).
Because $U_i$ depends only
on the depth-$R$ neighborhood of $i$ and the ingredients inside it,
changing one ingredient alters $U_i$ only for roots $i$ whose
depth-$R$ neighborhood contains the changed item; the degree/arity
bound $D_{\max}$ caps the number of such roots by
$\kappa_L = O(D_{\max}^{R})$, independent of
$n$.  Bounded degrees and arities also make the numbers of variables,
checks, and edges all $\Theta(n)$, so there are $\Theta(n)$ independent
ingredients.  Thus $\Pbit^{(L)}$ is a bounded-differences function with
per-ingredient Lipschitz constant $\kappa_L / n$, and McDiarmid's
inequality \cite{mcdiarmid1989} gives a
constant $a_L > 0$ (depending only on $L$ and $D_{\max}$) with
\begin{equation}
\label{eq:det-mean-tail}
   \PR\big\{\big|\Pbit^{(L)}(G_n) - \E\,\Pbit^{(L)}(G_n)\big| > t\big\}
   \le 2\exp(-a_L\, n\, t^2), \qquad t > 0.
\end{equation}
An equivalent local-dependence concentration argument
\cite{janson2004poisson, chatterjee2007stein}, using that each $U_i$
depends only on a bounded-radius neighborhood, yields the same tail.
No matching exposure is needed: the graph is deterministic, so the
sequential-exposure martingale of \Cref{thm:de} is replaced by direct
bounded differences over the independent ingredients.

\emph{Part B: the mean converges to $\Pdee^{(L)}$.}  For a finite
rooted socket-type-marked ball $\mathcal{T}$ of depth $R$, let
$\psi_L(\mathcal{T}) \in [0, 1]$ be the probability that the root is
unresolved after $L$ rounds when message passing is run on
$\mathcal{T}$ with hidden values and erasures drawn from the model;
$\psi_L$ is a fixed bounded function of the marked ball.  Since the
depth-$R$ estimate is a deterministic function of the depth-$R$ ball
and its ingredients, $\E[U_i] = \psi_L(\mathcal{T}_i)$ for the marked
ball $\mathcal{T}_i$ rooted at $i$, whence
\[
   \E\,\Pbit^{(L)}(G_n)
   = \frac1n \sum_i \psi_L(\mathcal{T}_i)
   = \E_{\mathcal{T} \sim \mu_n}\!\big[\psi_L(\mathcal{T})\big],
\]
with $\mu_n$ the empirical law of the depth-$R$ marked ball of a
uniformly chosen variable root.  Bounded degrees and arities and the
finite role and socket-type mark sets make the set $\mathcal{B}_R$ of
possible depth-$R$ marked balls \emph{finite}; on a finite set the weak
convergence $\mu_n \Rightarrow \mu_{\mathrm{GW}}$ of hypothesis~(c) is
atom-wise, $\mu_n(\mathcal{T}) \to \mu_{\mathrm{GW}}(\mathcal{T})$ for
every $\mathcal{T} \in \mathcal{B}_R$, so for the bounded $\psi_L$,
\[
   \E_{\mu_n}[\psi_L]
   = \sum_{\mathcal{T} \in \mathcal{B}_R}
     \psi_L(\mathcal{T})\,\mu_n(\mathcal{T})
   \;\longrightarrow\;
   \sum_{\mathcal{T} \in \mathcal{B}_R}
     \psi_L(\mathcal{T})\,\mu_{\mathrm{GW}}(\mathcal{T})
   = \E_{\mu_{\mathrm{GW}}}[\psi_L].
\]
The limit $\E_{\mu_{\mathrm{GW}}}[\psi_L]$ is the probability that the
root of the marked computation tree is unresolved after $L$ rounds.
The tree's degree, role, and socket-type marks match the DE recursion
(hypothesis~(c)) and the hidden values are drawn from the priors, so
\Cref{app:de-recursion} identifies this probability with
$\Pdee^{(L)}$.  Hence
$b_n^{\mathrm{det}} = |\E\,\Pbit^{(L)}(G_n) - \Pdee^{(L)}| \to 0$.

\emph{Combining.}  For every $t > 0$,
$\{|\Pbit^{(L)}(G_n) - \Pdee^{(L)}| > b_n^{\mathrm{det}} + t\}
\subseteq \{|\Pbit^{(L)}(G_n) - \E\,\Pbit^{(L)}(G_n)| > t\}$ because
$|\E\,\Pbit^{(L)}(G_n) - \Pdee^{(L)}| \le b_n^{\mathrm{det}}$, so
\eqref{eq:det-mean-tail} gives \eqref{eq:det-tail}, and
$b_n^{\mathrm{det}} \to 0$ gives convergence in probability.  If the
neighborhood-law convergence in (a)--(c) is quantified, then
$b_n^{\mathrm{det}} \le \sum_{\mathcal{T} \in \mathcal{B}_R}
|\mu_n(\mathcal{T}) - \mu_{\mathrm{GW}}(\mathcal{T})|$ makes the offset
rate explicit.

\section{Adjoint Equations: Derivation Detail}\label{app:adjoint-detail}

This appendix supplies the full reverse-mode chain-rule derivation of
the adjoint equations \eqref{eq:adjoint}--\eqref{eq:adjoint-grad} of
\Cref{thm:optimization}~(d), gives the explicit block-sparse forms of
the Jacobians $D_{\bm p}\, \DEmap_\lambda$ and $D_\lambda\, \DEmap_\lambda$,
and traces the KKT multiplier interpretation of part (e) back to unit
increments of the three erasure-tier knobs.

\subsection{Reverse-mode chain rule}
\label{app:adjoint-chain}

Fix a finite-horizon objective $J_L(\lambda) = \psi(\bm p_L, \lambda)$
and the recursion $\bm p_{\ell+1} = \DEmap_\lambda(\bm p_\ell)$ with
$\bm p_0 = \bm p_0(\lambda)$.  By the chain rule,
\begin{equation}
\label{eq:total-grad}
   \nabla_\lambda\, J_L
   = \nabla_\lambda\, \psi(\bm p_L, \lambda)
   + \big(D_{\bm p}\, \psi(\bm p_L, \lambda)\big)\,
   \frac{\mathrm{d} \bm p_L}{\mathrm{d} \lambda},
\end{equation}
where $\mathrm{d} \bm p_L / \mathrm{d} \lambda$ is the total derivative.
Iterating the recursion's chain rule,
\begin{equation}
\label{eq:tan}
   \frac{\mathrm{d} \bm p_{\ell+1}}{\mathrm{d} \lambda}
   = D_{\bm p}\, \DEmap_\lambda(\bm p_\ell)\,
   \frac{\mathrm{d} \bm p_\ell}{\mathrm{d} \lambda}
   + D_\lambda\, \DEmap_\lambda(\bm p_\ell),
   \qquad \ell = 0, \ldots, L - 1,
\end{equation}
with initial condition $\mathrm{d} \bm p_0 / \mathrm{d} \lambda =
D_\lambda\, \bm p_0(\lambda)$.  Substituting \eqref{eq:tan}
recursively into \eqref{eq:total-grad} gives forward-mode
sensitivity; the result is a sum of $L+1$ terms, the $\ell$-th term
being a product of $L - \ell$ Jacobians.  Forward mode is computationally
expensive when $\dim \lambda$ is large.

The reverse-mode (adjoint) computation is dual.  Define the adjoint
state $\xi_\ell \in \R^{2|\Sockets|}$ by the backward recursion
\begin{align}
\label{eq:adjoint-recur}
   \xi_L &= \big(\nabla_{\bm p}\, \psi(\bm p_L, \lambda)\big), \\
   \xi_\ell &= \big(D_{\bm p}\, \DEmap_\lambda(\bm p_\ell)\big)^\top\,
   \xi_{\ell+1}, \qquad \ell = L-1, L-2, \ldots, 0. \notag
\end{align}
We claim that
\begin{equation}
\label{eq:adjoint-sub}
   \nabla_\lambda\, J_L
   = \nabla_\lambda\, \psi(\bm p_L, \lambda)
   + \sum_{\ell=0}^{L-1} \big(D_\lambda\, \DEmap_\lambda(\bm p_\ell)\big)^\top\,
   \xi_{\ell+1}
   + \big(D_\lambda\, \bm p_0(\lambda)\big)^\top\, \xi_0,
\end{equation}
which is \eqref{eq:adjoint-grad} (the last term vanishes when the
initial condition does not depend on $\lambda$, e.g., when $\bm p_0$
is the zero-erasure state $\zero$).

\emph{Verification.}  Define the Lagrangian-style functional
\[
   \Lambda(\lambda)
   = \psi(\bm p_L, \lambda)
   + \sum_{\ell=0}^{L-1} \xi_{\ell+1}^\top\,
   \big(\DEmap_\lambda(\bm p_\ell) - \bm p_{\ell+1}\big),
\]
which equals $\psi(\bm p_L, \lambda)$ on the trajectory satisfying
the recursion.  The sign convention in the constraint term is
chosen so that the parameter Jacobian contribution enters
$\nabla_\lambda \Lambda$ with a $+$ sign (matching
\eqref{eq:adjoint-sub}); the alternative convention
$\xi_{\ell+1}^\top(\bm p_{\ell+1} - \DEmap_\lambda)$ would flip the
sign of every $D_\lambda \DEmap$ term and the corresponding adjoint
formula.  Differentiating with respect to $\lambda$, the sensitivity
terms involving $\mathrm{d} \bm p_\ell / \mathrm{d} \lambda$ group
into
\[
   -\sum_{\ell=1}^{L} \xi_\ell^\top\,
   \frac{\mathrm{d} \bm p_\ell}{\mathrm{d} \lambda}
   + \sum_{\ell=0}^{L-1}
   \xi_{\ell+1}^\top\, D_{\bm p}\, \DEmap_\lambda(\bm p_\ell)\,
   \frac{\mathrm{d} \bm p_\ell}{\mathrm{d} \lambda}
   + \nabla_{\bm p}\, \psi(\bm p_L, \lambda)\,
   \frac{\mathrm{d} \bm p_L}{\mathrm{d} \lambda},
\]
plus the explicit parameter contribution $\sum_{\ell=0}^{L-1}
\xi_{\ell+1}^\top D_\lambda \DEmap_\lambda(\bm p_\ell)$ from the
$\DEmap_\lambda$ term.  Reindexing the first sum by
$\ell \to \ell - 1$ and using the adjoint recursion
\eqref{eq:adjoint-recur} to identify $\xi_\ell = (D_{\bm p}\,
\DEmap_\lambda(\bm p_\ell))^\top\, \xi_{\ell+1}$ for $\ell = 0,
\ldots, L - 1$, all the $\mathrm{d} \bm p_\ell / \mathrm{d} \lambda$
terms cancel except the boundary $\xi_0^\top\, \mathrm{d} \bm p_0 /
\mathrm{d} \lambda$.  What remains is exactly the right-hand side
of \eqref{eq:adjoint-sub} with the correct $+$ sign on the
$D_\lambda \DEmap$ sum.  Since $\Lambda(\lambda) \equiv J_L(\lambda)$
on the trajectory, this is $\nabla_\lambda\, J_L$, which is the
gradient formula \eqref{eq:adjoint-grad} of
\Cref{thm:optimization}~(d).

\subsection{Block-sparse structure of the state Jacobian
\texorpdfstring{$D_p\, \DEmap_\lambda$}{D\_p Phi}}
\label{app:adjoint-Dp}

The state $\bm p_\ell \in [0, 1]^{2|\Sockets|}$ has coordinates
indexed by socket type $\tau$ and value branch $b$.  The recursion
$\bm p_{\ell+1} = \DEmap_\lambda(\bm p_\ell)$ decomposes coordinatewise as
$\pmsg_{\ell+1, \tau}^{(b)}$ given by \eqref{eq:pupdate}, which is a
function of $\{\hmsg_{\ell, \tau'}^{(b)}\}_{\tau' \in \Sockets_{\socketrole{\tau}}}$.
Each $\hmsg_{\ell, (\theta, k)}^{(b)}$ in turn depends on the forcing
probability $\forceprob_{\ell, \theta, k}^{(b)}$, which depends on
$\{\pmsg_{\ell, (\theta, k')}^{(b')}\}_{k' \neq k, b' \in \{0, 1\}}$
through the inbound erasure probabilities $\pbar_{\ell, \theta, k'}^{(x_{k'})}$
of \eqref{eq:pbar}.

The chain rule gives
\begin{equation}
\label{eq:Dp-element}
   \frac{\partial \pmsg_{\ell+1, \tau}^{(b)}}
   {\partial \pmsg_{\ell, \tau'}^{(b')}}
   = \sum_{(\theta, k):\, (\theta, k') \in \Sockets, k' = k}
   \frac{\partial \pmsg_{\ell+1, \tau}^{(b)}}
   {\partial \hmsg_{\ell, (\theta, k)}^{(b)}}\,
   \frac{\partial \hmsg_{\ell, (\theta, k)}^{(b)}}
   {\partial \forceprob_{\ell, \theta, k}^{(b)}}\,
   \frac{\partial \forceprob_{\ell, \theta, k}^{(b)}}
   {\partial \pmsg_{\ell, \tau'}^{(b')}},
\end{equation}
with the summation only over check templates $\theta$ and target
sockets $k$ such that $(\theta, k')$ matches $\tau'$ for some
$k' \neq k$.  The block sparsity arises because:
\begin{itemize}[leftmargin=2em]
\item $\partial \pmsg_{\ell+1, \tau}^{(b)} / \partial \pmsg_{\ell, \tau'}^{(b')}$
is nonzero only if $\socketrole{\tau} = \socketrole{\tau'}$ (variable
update preserves variable role) and $b = b'$ (variable update is
diagonal in the value branch);
\item $\partial \hmsg_{\ell, (\theta, k)}^{(b)} / \partial
\forceprob_{\ell, \theta, k}^{(b)} = -(1 - \epsC[s_\theta])\,
\eta_{s_\theta, r_{\theta, k}}$, scalar;
\item $\partial \forceprob_{\ell, \theta, k}^{(b)} / \partial
\pmsg_{\ell, (\theta, k')}^{(b')}$ is nonzero only if $k' \neq k$
(socket asymmetry of the forcing rule) and depends on the Boolean
primitive $f_\theta$ through the value-conditioned form of
\eqref{eq:forcing-prob}.
\end{itemize}
For \emph{XOR} templates the forcing probability factors as
\eqref{eq:xor-forcing}, so
\[
   \frac{\partial \forceprob_{\ell, \theta, k}^{\mathrm{XOR}}}
   {\partial \pmsg_{\ell, (\theta, k')}}
   = -\eta_{r_{\theta, k'}, s_\theta}\,
   \prod_{k'' \neq k, k'} \eta_{r_{\theta, k''}, s_\theta}\,
   (1 - \pmsg_{\ell, (\theta, k'')}),
\]
which is independent of $b$.  For \emph{AND} templates with $b = 1$,
the partial derivative is zero (positive certificates do not depend
on the inbound erasure rates by \eqref{eq:and-positive}).  For AND
templates with $b = 0$,
\[
   \frac{\partial \forceprob_{\ell, \theta, k}^{(0)}}
   {\partial \pmsg_{\ell, (\theta, k')}^{(1)}}
   = -\Vbias_{r_{\theta, k'}}\,
   \eta_{r_{\theta, k'}, s_\theta}\,
   \prod_{k'' \neq k, k'} \Vbias_{r_{\theta, k''}}\,
   \eta_{r_{\theta, k''}, s_\theta}\,
   (1 - \pmsg_{\ell, (\theta, k'')}^{(1)}),
\]
nonzero only on the value-$1$ inbound branch.  This last asymmetry
is the analytical fingerprint of the verifier asymmetry of
\Cref{rem:and-practical}.

The matrix $D_{\bm p}\, \DEmap_\lambda(\bm p_\ell) \in
\R^{2|\Sockets| \times 2|\Sockets|}$ is therefore block-diagonal in
the value branch under XOR (the $b = 0$ and $b = 1$ blocks are
identical), and block-coupled-but-triangular in the value branch
under AND (the $b = 0$ block reads from the $b = 1$ block but not
vice versa).  Mixed-template ensembles inherit the union of these
sparsity patterns.

\subsection{Block-sparse structure of the parameter Jacobian
\texorpdfstring{$D_\lambda\, \DEmap_\lambda$}{D\_lambda Phi}}
\label{app:adjoint-Dlambda}

The design parameter $\lambda$ collects role proportions
$\{\pi_r^V\}, \{\pi_s^C\}$, template proportions $\{\pi_\theta^C\}$,
degree distributions $\{P_{\deg \mid r}\}$, value priors $\{\Vbias_r\}$,
erasure probabilities $\{\epsV[r]\}, \{\epsC[s]\}$, and channel
fidelities $\{\eta_{r, s}\}$.  Differentiating
\eqref{eq:hupdate}--\eqref{eq:pupdate} with respect to each
parameter gives:

\noindent\emph{Variable-side erasure $\epsV[r]$.}  The variable update
\eqref{eq:pupdate} factors $\epsV[\socketrole{\tau}]$ multiplicatively, so
\[
   \frac{\partial \pmsg_{\ell+1, \tau}^{(b)}}{\partial \epsV[r]}
   = \one\{\socketrole{\tau} = r\} \,
   \E\!\left[\prod_{\tau' \in \Sockets_r}
   \big(\hmsg_{\ell, \tau'}^{(b)}\big)^{D_{\tau'}^{(r), \tau, \mathrm{ex}}}\right].
\]
This is the multiplicative-pre-factor signature recorded in
\Cref{prop:noninterchange}.

\noindent\emph{Verifier-side erasure $\epsC[s]$.}  Through \eqref{eq:hupdate},
\[
   \frac{\partial \hmsg_{\ell, (\theta, k)}^{(b)}}{\partial \epsC[s]}
   = \one\{s_\theta = s\}\, \eta_{s, r_{\theta, k}}\,
   \forceprob_{\ell, \theta, k}^{(b)},
\]
acting symmetrically across both value branches with a $\forceprob^{(b)}$-dependent
prefactor.  The variable update then propagates this through
$\partial \pmsg / \partial \hmsg$ as in \eqref{eq:Dp-element}.

\noindent\emph{Reasoning-channel fidelity $\eta_{r, s}$.}  Through
\eqref{eq:pbar} and \eqref{eq:hupdate}, $\eta$ enters in two distinct
positions: once on the inbound side (modulating $\pbar^{(x_k)}$ and
hence $\forceprob$) and once on the return side (multiplying
$\forceprob$ to form $\hmsg$).  These two entry positions act
asymmetrically across the value branches under non-symmetric
Boolean factors, under AND, the $b = 1$ branch's $\forceprob$ does
not depend on the inbound channel $\eta_{V \to C}$ (cf.\
\eqref{eq:and-positive}), while the $b = 0$ branch's $\forceprob$
depends on it as a $(d-1)$-fold product (cf.\ \eqref{eq:and-negative}).
The return channel acts on both branches identically.

This is the structural reason the channel-side column of the
parameter Jacobian acts asymmetrically across $b$ in the proof of
\Cref{prop:noninterchange}.

\noindent\emph{Other parameters.}  Partial derivatives with respect to
$\Vbias_r, \pi_r^V, \pi_s^C, \pi_\theta^C$, and the degree-law
parameters follow by direct differentiation of
\eqref{eq:terminal-de} and \eqref{eq:pupdate}.  In each case, the
partial is supported on the sockets and value branches whose role
or template index matches the parameter index, a sparsity pattern
inherited directly from the role-typed structure of the recursion.

\subsection{KKT shadow prices, traced to unit erasure-tier increments}
\label{app:adjoint-kkt}

Consider the constrained optimization
\begin{equation}
\label{eq:opt-program}
   \minimize_{\lambda \in \Designspace}\; J_L(\lambda)
   \quad \text{subject to}\quad
   \Cost(\lambda) \le B,
\end{equation}
with a smooth $\Cost$ and a smooth $J_L$ on a smooth stratum of
$\Designspace$ (which is the case wherever the support of the
degree law is fixed; see \Cref{thm:optimization}~(b)).  The
Lagrangian is
\[
   \mathcal{L}(\lambda, \mu)
   = J_L(\lambda) + \mu \cdot \big(\Cost(\lambda) - B\big),
   \qquad \mu \ge 0.
\]
First-order necessary conditions for $\lambda^*$ to be a regular
local minimum are
\begin{equation}
\label{eq:kkt-firstorder}
   \nabla_\lambda\, J_L(\lambda^*) + \mu^*\, \nabla_\lambda\,
   \Cost(\lambda^*) = 0,
   \qquad
   \mu^* \ge 0,
   \qquad
   \mu^* \cdot \big(\Cost(\lambda^*) - B\big) = 0.
\end{equation}
By \eqref{eq:adjoint-sub}, $\nabla_\lambda\, J_L$ is computable in
$O(L)$ adjoint steps once $\bm p_L$ is forward-propagated.

\emph{Interpretation of the budget multiplier $\mu^*$.}  When the
budget constraint is active ($\Cost(\lambda^*) = B$), the standard
sensitivity analysis gives
\[
   \frac{\mathrm{d}}{\mathrm{d} B}\, J_L(\lambda^*(B))
   = -\mu^*,
\]
i.e., $\mu^*$ is the marginal reliability gain per unit additional
budget, the shadow price of agent-system budget.

\emph{Interpretation of the tier-parameter sensitivities.}  Each of
the three erasure-tier parameters $\epsV[r], \epsC[s], 1 - \eta_{r, s}$
appears as one (or several) coordinates of $\lambda$.  The corresponding
component of $\nabla_\lambda\, J_L(\lambda^*)$ is the local
sensitivity of the residual objective to that parameter (\emph{not}
a Lagrange multiplier; multipliers attach to constraints), computed
via the adjoint formula \eqref{eq:adjoint-sub}:
\begin{align}
\label{eq:adjoint-tier}
   \frac{\partial J_L}{\partial \epsV[r]}\bigg|_{\lambda^*}
   &= \sum_{\ell=0}^{L-1}
   \big(D_{\epsV[r]}\, \DEmap_\lambda(\bm p_\ell^*)\big)^\top\,
   \xi_{\ell+1}^*
   \;+\; \big(D_{\epsV[r]}\, \bm p_0(\lambda^*)\big)^\top\, \xi_0^*, \\
   \frac{\partial J_L}{\partial \epsC[s]}\bigg|_{\lambda^*}
   &= \sum_{\ell=0}^{L-1}
   \big(D_{\epsC[s]}\, \DEmap_\lambda(\bm p_\ell^*)\big)^\top\,
   \xi_{\ell+1}^*, \notag\\
   \frac{\partial J_L}{\partial \eta_{r, s}}\bigg|_{\lambda^*}
   &= \sum_{\ell=0}^{L-1}
   \big(D_{\eta_{r, s}}\, \DEmap_\lambda(\bm p_\ell^*)\big)^\top\,
   \xi_{\ell+1}^*. \notag
\end{align}
The variable-side formula carries the boundary term
$(D_{\epsV[r]}\, \bm p_0(\lambda^*))^\top \xi_0^*$ because the
initial state $\bm p_0$ is initialized from the variable-side
erasure rates (\eqref{eq:pinit}); the verifier-side and
reasoning-channel formulas have no boundary contribution because
$\epsC$ and $\eta$ do not enter $\bm p_0$.  By
\Cref{prop:noninterchange} and \Cref{rem:rank-scope}, the
parameter Jacobian $D_{\mathrm{par}}[\DEmap_\lambda(\zero)]$ has
rank at least two on a generic open subregion of parameter space,
so on that subregion the three components above are not all
collinear: no single-scalar reduction summarizes them.  In a
single-role single-template slice, $\epsC$ and $\eta_{C \to V}$
are locally confounded through the product $(1-\epsC)\eta_{C\to V}$;
distinguishing all three operational tiers in a rank-$\ge 3$ sense
requires nondegenerate role structure and is established in
\Cref{prop:noninterchange-roles}.
For shadow-price interpretations of these sensitivities, an
explicit cost model with investment variables $u_r^V, u_s^C,
u_{r,s}^\eta$ and reliability response curves is needed, as
detailed in \Cref{rem:investment-variables}; without that model,
$\nabla_\lambda J_L$ supplies sensitivities, not multipliers.

In practice, the architect reads off the three tier sensitivities
at the optimum $\lambda^*$ and identifies the locally most
influential intervention as the tier whose sensitivity is largest in
absolute value relative to its per-unit cost; under an explicit
investment model, the corresponding stationarity equation
identifies the active investment.  When more than one tier
sensitivity is large, the optimum reallocates across all such tiers
simultaneously, a phenomenon that has no
single-effective-noise-parameter analog in MET-LDPC or in the
noisy-message-passing-decoder line, by \Cref{prop:noninterchange}.

\section*{Acknowledgments}

This work was supported in part by the U.S.\ National Science
Foundation under Grants CNS-2528914 and CNS-2150832.

\bibliographystyle{IEEEtran}
\bibliography{paper}

\end{document}